\documentclass[reprint, amsmath,amssymb,showpacs,floatfix,longbibliography, aps, twocolumn, superscriptaddress,nofootinbib]{revtex4-1}
\usepackage{CJKutf8}

\usepackage{amsthm}
\newtheorem{theorem}{Theorem}
\newtheorem{proposition}{Proposition}

\newtheorem{corollary}{Corollary}[theorem]
\newtheorem{lemma}[theorem]{Lemma}

\usepackage{multirow}
\usepackage[utf8]{inputenc}
\usepackage{amsmath,amssymb,amsfonts,stmaryrd}
\usepackage{physics}
\usepackage{dsfont}
\usepackage{graphicx}
\usepackage{xcolor}
\usepackage{colortbl}
\usepackage{graphicx}
\usepackage{cancel}
\usepackage{natbib}
\usepackage{color}
\usepackage{tikz}
\usepackage{mathrsfs}
\usepackage{relsize}

\usepackage[all]{xy}
\usepackage{yhmath}
\usepackage{blkarray}
\usepackage{tabularx}
\usepackage{array}
\usepackage{makecell}

\usepackage{diagbox}
\usepackage{algorithm}
\usepackage{algpseudocode}
\usepackage{lipsum}
\newcommand*\colvec[3][]{
    \begin{pmatrix}\ifx\relax#1\relax\else#1\\\fi#2\\#3\end{pmatrix}
}
\algrenewcommand\algorithmicrequire{\textbf{Input:}}
\algrenewcommand\algorithmicensure{\textbf{Output:}}


\usepackage{amsmath}
\usepackage{bm} %allows for bold math type
\usepackage{subfigure} %allows for 2x2 figures
\usepackage{environ}
\NewEnviron{eqs}{%
\begin{equation}\begin{split}
    \BODY
\end{split}\end{equation}}
\usepackage{url}
\usepackage[margin=0.75in]{geometry}
\usepackage{scalerel}
\usepackage{stackengine,wasysym}

\usepackage{hyperref}
\hypersetup{colorlinks=true,citecolor=blue,linkcolor=blue, urlcolor=blue, breaklinks=true}
\hypersetup{linktocpage}

\allowdisplaybreaks

\usepackage{mathtools}

\newcommand{\ZZ}{{\mathbb Z}}

\newcommand{\g}{{\gamma}}

\graphicspath{{./Figures/}}

\newcommand{\prlsection}[1]{\vspace{6pt}\noindent\textbf{\textsc{#1.}}—}

\usepackage{capt-of}

\begin{document}

\author{Zijian Liang}
\affiliation{International Center for Quantum Materials, School of Physics, Peking University, Beijing 100871, China}

\author{Yu-An Chen}
\email[E-mail: ]{yuanchen@pku.edu.cn}
\affiliation{International Center for Quantum Materials, School of Physics, Peking University, Beijing 100871, China}

\date{\today}

\title{Topological subsystem bivariate bicycle codes with four-qubit check operators}

\begin{abstract}
High-rate bivariate bicycle (BB) codes are promising low-overhead quantum memories, but their stabilizer checks typically have weight $6$ or higher, making syndrome extraction challenging.
We introduce subsystem bivariate bicycle (SBB) codes, a translation-invariant CSS subsystem construction that realizes BB-code logical structure using local weight-$4$ gauge measurements. Their stabilizer syndromes are inferred by multiplying the corresponding gauge outcomes.
We further show that nonlocal stabilizers in translation-invariant CSS subsystem codes can be detected using a determinantal-ideal criterion based on the gauge-operator commutation matrix.
When this criterion excludes nonlocal stabilizers, a finite-depth Clifford circuit decouples gauge qubits and identifies the protected subsystem with a corresponding BB stabilizer code.
An SBB code is topological, meaning that it has no nontrivial local logical operators, if and only if the corresponding BB code is topological.
A finite search yields low-overhead examples including $[[27,6,3]]$, $[[75,10,5]]$, and $[[108,12,6]]$; the latter encodes six times more logical qubits than a subsystem surface code at the same block length and distance.
These results show how gauge degrees of freedom can make high-rate BB logical structure compatible with local weight-$4$ syndrome extraction.
\end{abstract}
\maketitle

%%%%%%%%%%%%%%%%%%%%%%%%%%%%%%%%%%%%%%%%%%%%%%%%%%%%%%%%%%%%%%%%%%%%%%%%%%%%%%%%%%%%%%%%%%%%%%%%%%%%%%%%%%%%%%%%

\prlsection{Introduction}
The realization of large-scale fault-tolerant quantum computers relies on quantum error correction~\cite{Shor1995Scheme, Steane1996Error, Knill1997Theory, gottesman1997stabilizer, kitaev2003fault}.
Topological codes, especially surface and toric codes, have become leading benchmarks because of their high thresholds and low-weight local stabilizer checks~\cite{bravyi1998quantum, dennis2002topological, terhal2015quantum, semeghini2021probing, Verresen2021PredictionTC, bluvstein2022quantum, google2023suppressing, Google2023NonAbelian, Google2024surface, iqbal2023topological, iqbal2024NonAbelian, Cong2024EnhancingTO}.
Their main drawback is overhead: implementing large-scale algorithms may require a substantial number of physical qubits~\cite{Fowler2012Surfacecodes, Litinski2019gameofsurfacecodes}.

Bivariate bicycle (BB) codes have recently emerged as a compelling low-overhead alternative.
They can achieve order-of-magnitude overhead improvements over the surface code~\cite{Kovalev2013QuantumKronecker, Pryadko2022DistanceGB, Bravyi2024HighThreshold, wang2024coprime, Wang2024Bivariate, tiew2024low, wolanski2024ambiguity, gong2024toward, maan2024machine, cowtan2024ssip, shaw2024lowering, cross2024linear, voss2024multivariate, berthusen2025toward, Eberhardt2025Logical, lin2025single, liang2025generalized, liang2025selfdual, liang2026generalizedmathbbzptoriccodes}.
However, their stabilizer measurements typically involve weight-$6$ or higher-weight checks, which are more demanding from a hardware perspective than the weight-$4$ checks of surface codes.
This motivates the search for constructions that retain the favorable logical structure of BB codes while reducing the weight of the operators measured directly.

Subsystem codes, originating from the operator quantum error-correction formalism~\cite{Kribs2005Unified, Kribs2006Operator, Poulin2005subsystem, Bacon2006Operator}, provide a natural mechanism for this reduction~\cite{Bombin2010Topologicalsubsystemcodes, Bombin2012Universaltopologicalphase, zhang2026coupledlayerconstructionquantumproduct}.
Rather than measuring stabilizers directly, one measures lower-weight gauge operators and combines their outcomes to infer the stabilizer syndrome.
Known examples include the subsystem surface code~\cite{bravyi2013subsystemsurfacecodesthreequbit} and subsystem hyperbolic codes~\cite{Higgott2021subsystem}, both of which admit weight-3 gauge checks.
%The trade-off is to require additional physical qubits.

In this work, we combine low-overhead BB codes with the measurement advantages of subsystem codes.
We introduce \textbf{subsystem bivariate bicycle (SBB) codes}. We propose a family of translation-invariant CSS subsystem codes on a square lattice with three qubits per unit cell with weight-4 gauge checks.
To make the construction concrete, we first present the $[[75,10,5]]$ SBB code as a guiding example: it measures only local weight-4 gauge operators, encodes five times more logical qubits than the $[[75,2,5]]$ subsystem surface code at the same block length and distance.
After this example, we develop the general Laurent-polynomial criterion that explains how the construction works.

Our explicit constructions show that this approach can substantially reduce overhead at a modest number of qubits.
% The $[[75,10,5]]$ code also reaches a comparable improvement factor to the $[[8064,338,10]]$ subsystem hyperbolic code.
Examples and comparisons are summarized in Table~\ref{tab:code_comparison}.

\begin{table}[thb]
\centering
\renewcommand{\arraystretch}{1.2}
\begin{tabular}{cccc}
\hline
Code family & $[[n,k,d]]$ & $kd/n$ & $kd^2/n$ \\
\hline
Subsystem surface code~\cite{bravyi2013subsystemsurfacecodesthreequbit}
& $[[3L^2,2,L]]$ & $2/(3L)$ & $0.67$ \\
\hline
\multirow{3}{*}{\begin{tabular}{c}
Subsystem \\
hyperbolic code~\cite{Higgott2021subsystem}
\end{tabular}}
& $[[384,18,4]]$ & $0.188$ & $0.75$ \\
& $[[1536,66,8]]$ & $0.344$ & $2.75$ \\
& $[[8064,338,10]]$ & $0.420$ & $4.19$ \\
\hline
\multirow{6}{*}{\begin{tabular}{c}
Subsystem bivariate\\
bicycle code
\end{tabular}}
& $[[27,6,3]]$ & $0.667$ & $2$ \\
& $[[60,10,4]]$ & $0.667$ & $2.67$ \\
& $[[75,10,5]]$ & $0.667$ & $3.33$ \\
& $[[90,12,5]]$ & $0.667$ & $3.33$ \\
& $[[108,12,6]]$ & $0.667$ & $4$ \\
& $[[126,14,6]]$ & $0.667$ & $4$ \\
%& $[[126,6,9]]$ & $0.429$ & $3.86$ \\
%& $[[147,14,6]]$ & $0.571$ & $3.43$ \\
\hline
\end{tabular}
\caption{
Comparison of subsystem code families.
The columns $kd/n$ and $kd^2/n$ are the normalized quantities appearing in the Bravyi--Poulin--Terhal tradeoff bounds~\cite{Bravyi2011Subsystem, Bravyi2010Tradeoffs} for two-dimensional local subsystem and stabilizer codes, respectively.
For SBB codes, $d$ is the dressed distance.
Explicit gauge-generator data and verification of the listed SBB examples are provided in Appendix~\ref{app: examples_verification}.
}
\label{tab:code_comparison}
\end{table}

\begin{figure}[t]
    \centering
    \subfigure[$G_{X,1}$]{\includegraphics[width=0.44\linewidth]{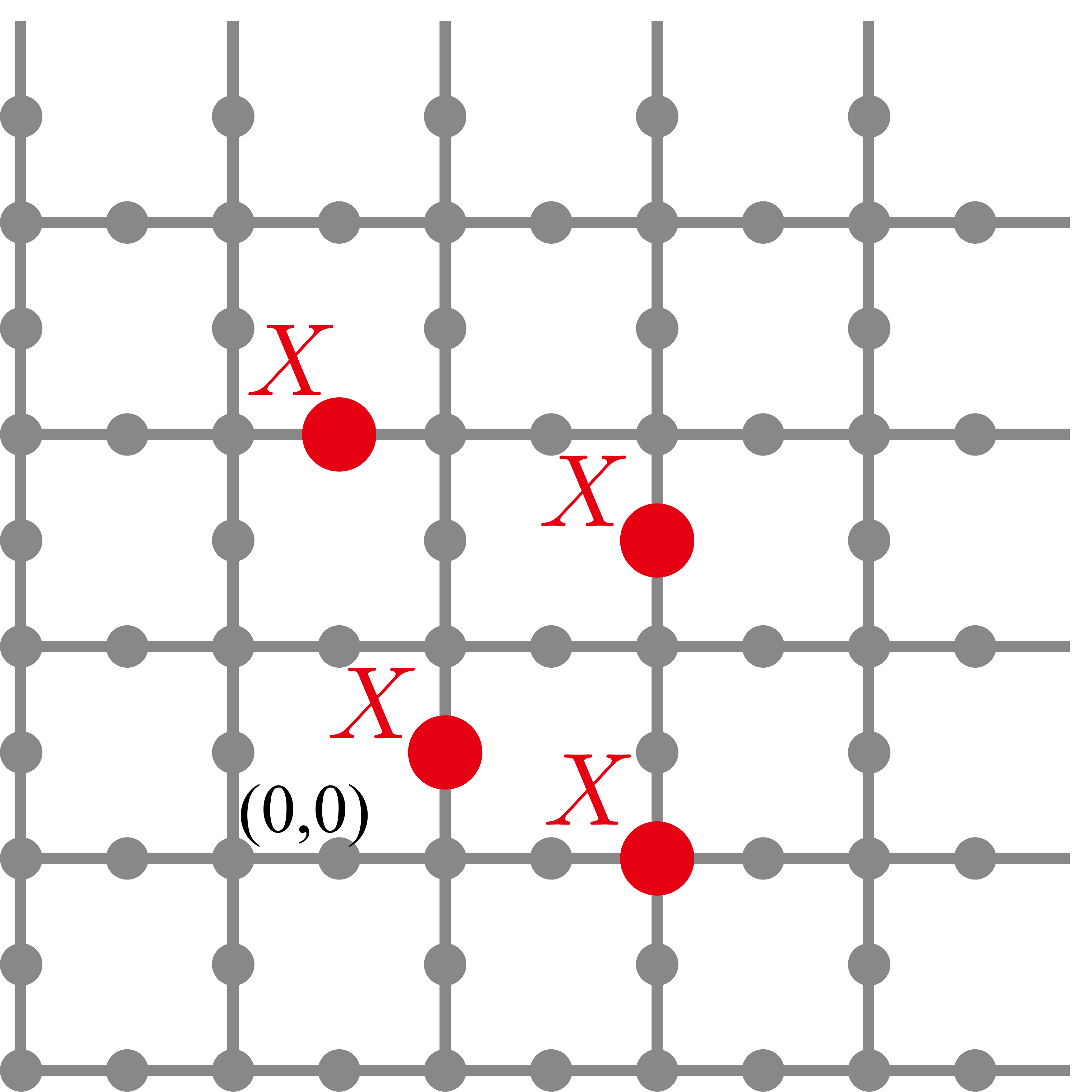}}
    \hspace{0.25cm}
    \subfigure[$G_{X,2}$]{\includegraphics[width=0.44\linewidth]{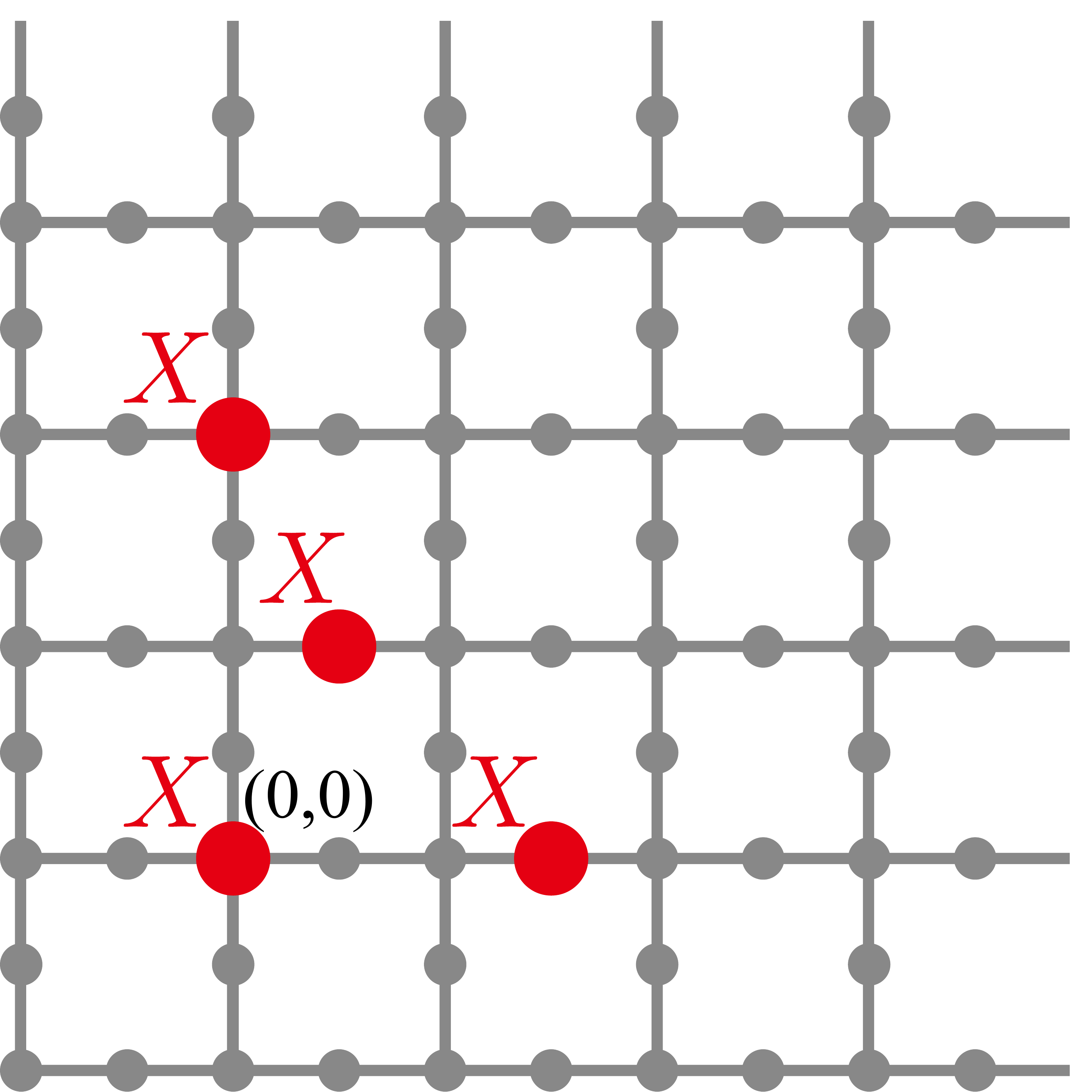}}\\
    \subfigure[$G_{Z,1}$]{\includegraphics[width=0.44\linewidth]{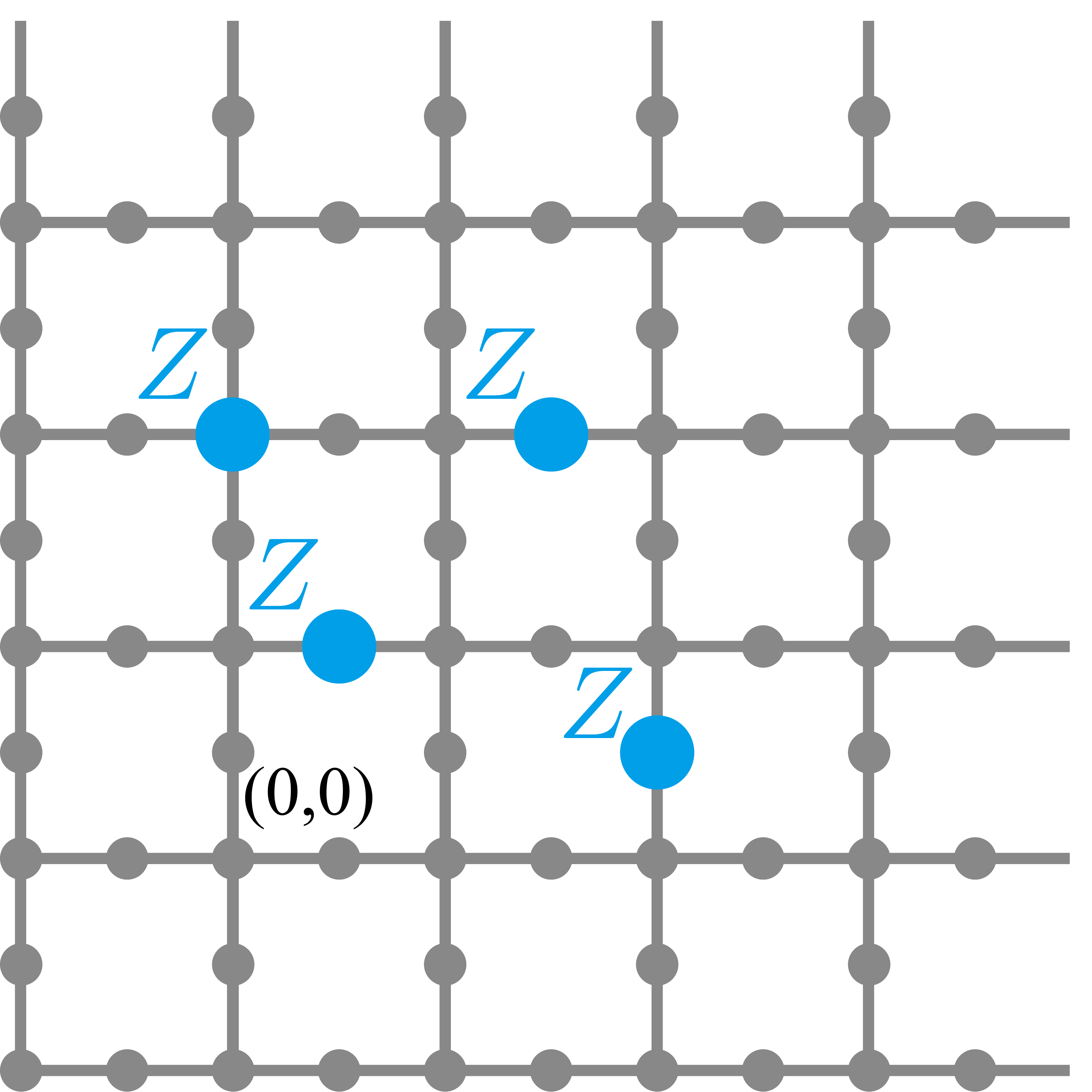}}
    \hspace{0.25cm}
    \subfigure[$G_{Z,2}$]{\includegraphics[width=0.44\linewidth]{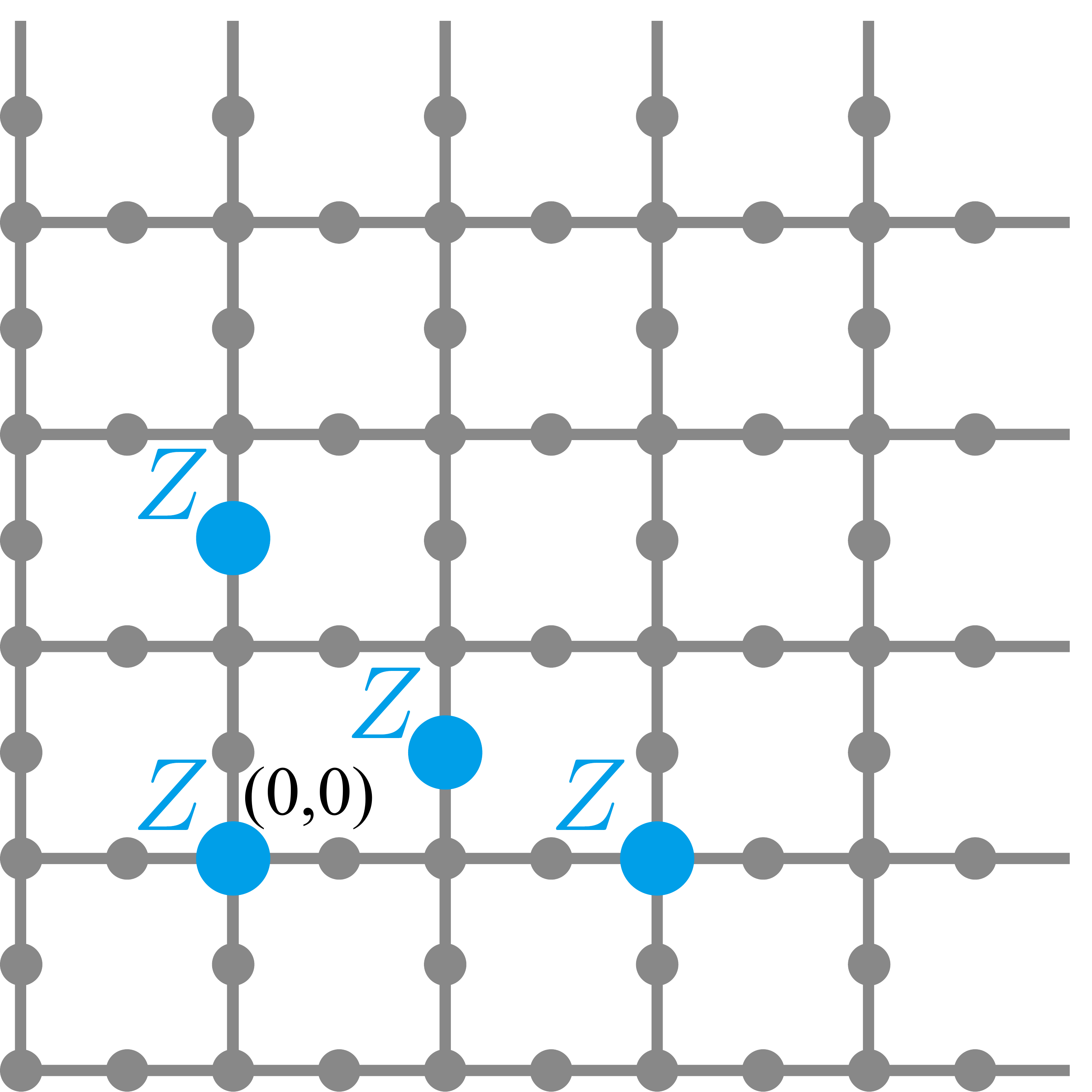}}
    \caption{
    Four types of local gauge generators of the $[[75,10,5]]$ SBB code.
    All other gauge generators are obtained by lattice translations.
    In the convention of Eq.~\eqref{eq: general_subsystem_gauge_operators}, they are specified by polynomials:
    (a) $(f_1,g_1,h_1)=(x^2,y^2,x+x^2y)$;
    (b) $(f_2,g_2,h_2)=(1+y^2,x+y,0)$;
    (c) $(f_3,g_3,h_3)=(y^2,y+xy^2,x^2)$;
    (d) $(f_4,g_4,h_4)=(1+x^2,0,x+y)$.
    }
    \label{fig: gauge operator of [[75, 10, 5]] code}
\end{figure}

\begin{figure}[t]
    \centering
    \subfigure[$X$-type stabilizer]{\includegraphics[width=0.44\linewidth]{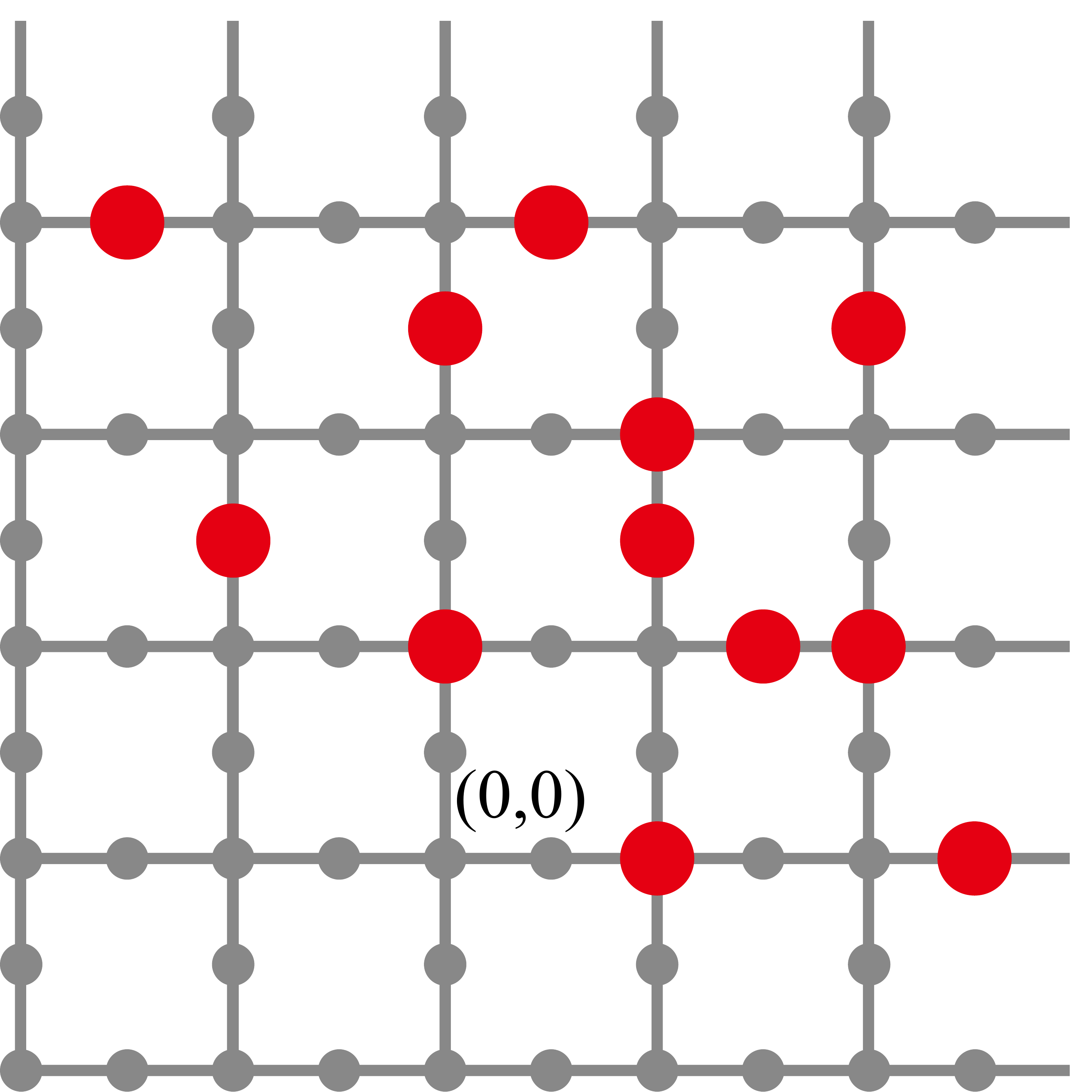}\label{fig: 2(a)}}
    \hspace{0.25cm}
    \subfigure[$Z$-type stabilizer]{\includegraphics[width=0.44\linewidth]{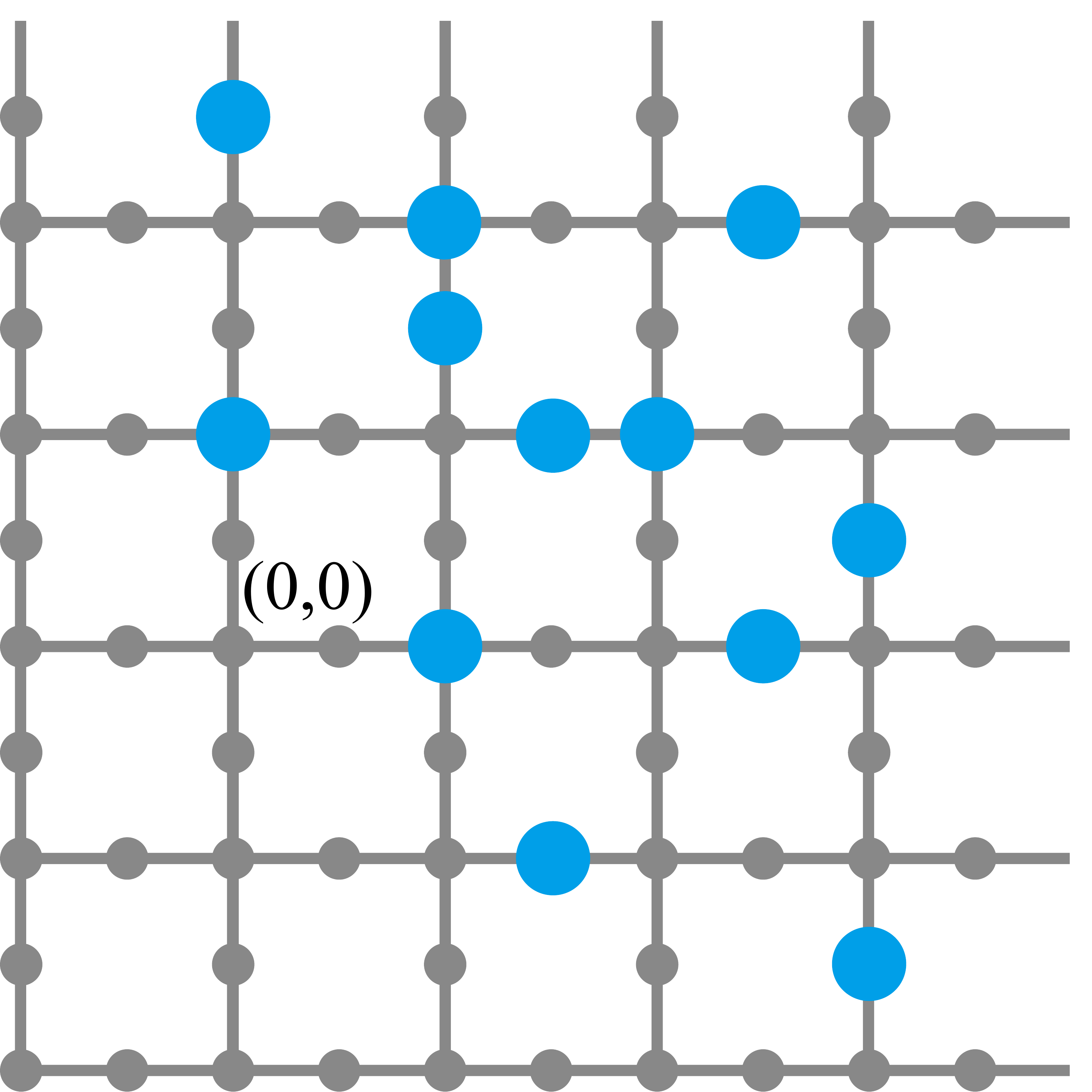}\label{fig: 2(b)}}\\
    \subfigure[Dressed logical $X$]{\includegraphics[width=0.44\linewidth]{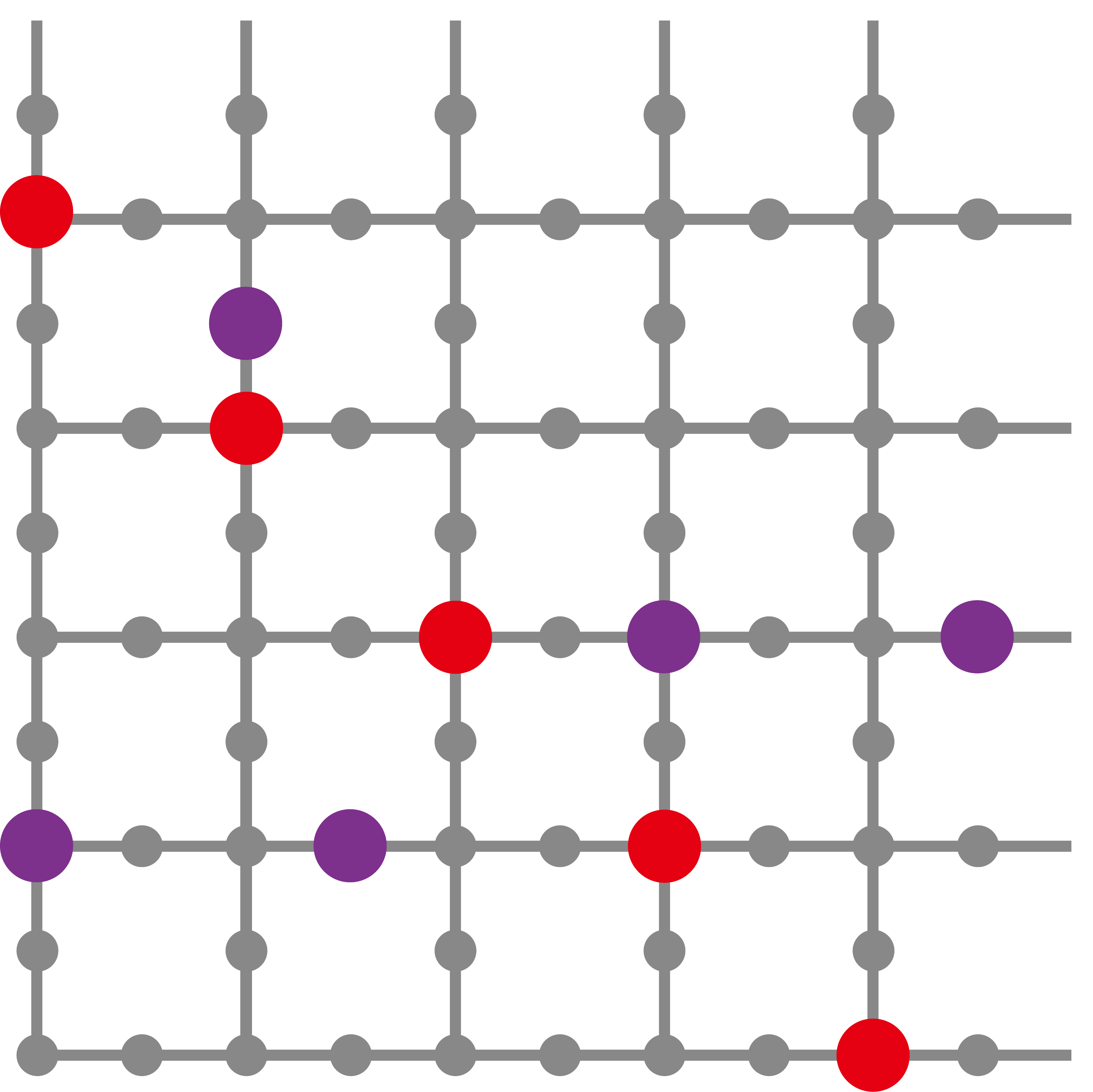}\label{fig: 2(c)}}
    \hspace{0.25cm}
    \subfigure[Dressed logical $Z$]{\includegraphics[width=0.44\linewidth]{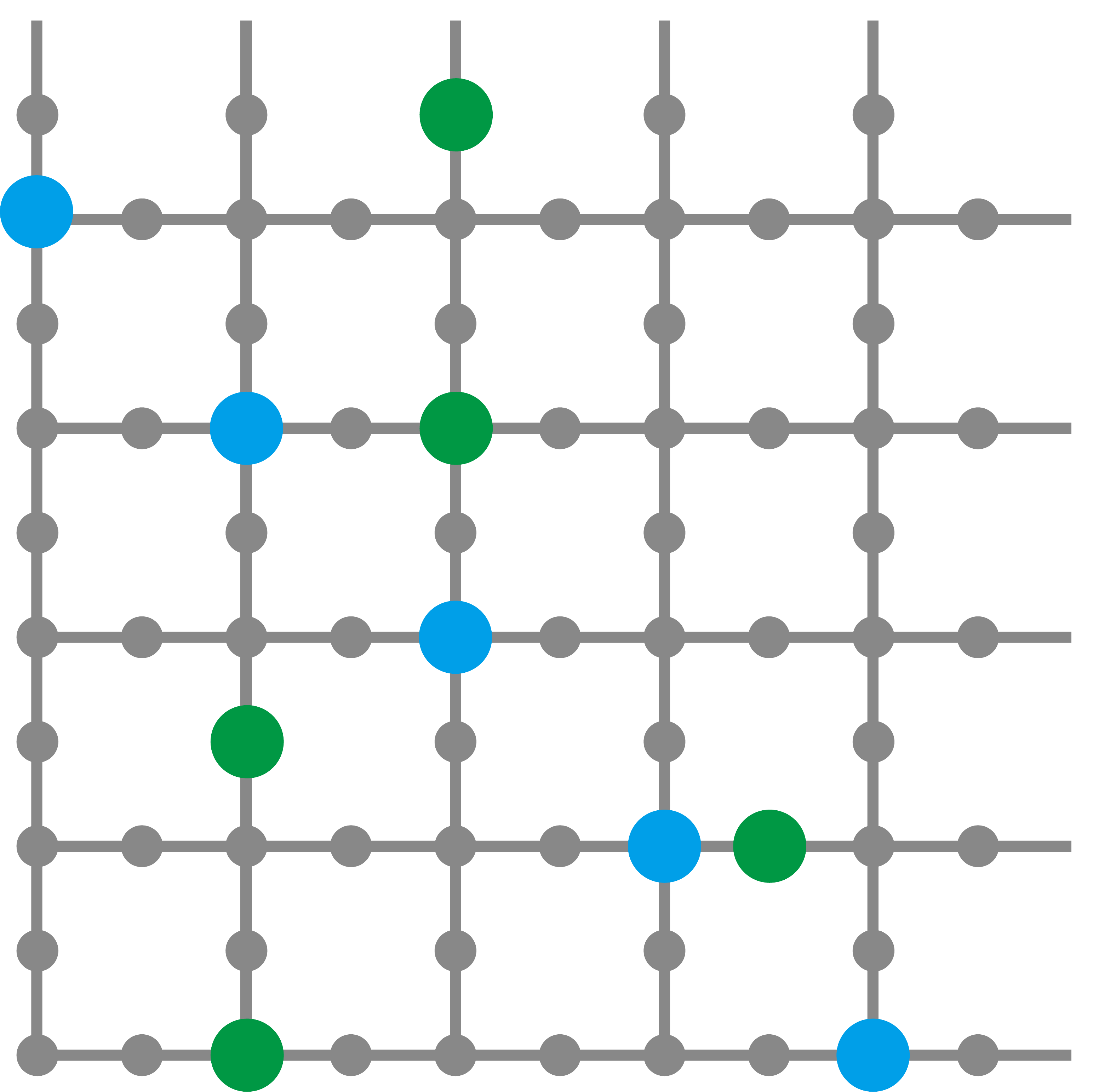}\label{fig: 2(d)}}
    \caption{
    Weight-$12$ stabilizers and representative weight-$5$ dressed logical operators of the $[[75,10,5]]$ SBB code.
    Panels (a) and (b) show local $X$- and $Z$-type stabilizers, respectively.
    These stabilizers are products of nearby gauge generators and lie in the center of the gauge group; hence they commute with all gauge operators.
    Panels (c) and (d) show representative minimum-weight dressed logical operators.
    In panel (c), red and purple both denote Pauli-$X$ operators and represent two inequivalent weight-$5$ logical-$X$ operators.
    The $5$ translations of each operator in the $y$ direction give $10$ independent dressed logical-$X$ operators modulo multiplication by gauge operators.
    Similarly, in panel (d), blue and green both denote Pauli-$Z$ operators and represent two inequivalent weight-$5$ logical-$Z$ operators.
    Their translations in the $x$ direction give $10$ independent dressed logical-$Z$ operators.
    }
     \label{fig: Stabilizers and dressed logical operators of [[75, 10, 5]] code}
\end{figure}

%%%%%%%%%%%%%%%%%%%%%%%%%%%%%%%%%%%%%%%%%%%%%%%%%%%%%%%%%%%%%%%%%%%%%%%%%%%%%%%%%%%%%%%%%%%%%%%%%%%%%%%%%%%%%%%%
\prlsection{Guiding example}
\label{sec:guiding_example_75}
We use standard subsystem-code terminology throughout; a brief review is given in Appendix~\ref{app: preliminaries}.
We begin with a concrete example: the $[[75,10,5]]$ SBB code.
The code is defined on a $5\times5$ square lattice with periodic boundary conditions.
Each unit cell contains three physical qubits, one on the vertex and two on the edges, so the total qubit number is $n=75$.
% Thus the total number of physical qubits is $n=3\times 5^2=75$.

The syndrome-extraction measurements are local gauge checks.
%rather than stabilizer checks.
There are four translation-invariant families of gauge generators per unit cell, denoted
$G_{X,1}$, $G_{X,2}$, $G_{Z,1}$, and $G_{Z,2}$.
Two are $X$-type and two are $Z$-type, and each is weight-$4$, as shown in Fig.~\ref{fig: gauge operator of [[75, 10, 5]] code}.

The gauge checks need not commute with each other, but suitable products of gauge checks can commute with every gauge generator.
Such products lie in the center of the gauge group and define stabilizers.
For the $[[75,10,5]]$ code, each local $X$-stabilizer is the product of three  $X$-gauge checks.
Relative to a reference unit cell, it is formed from two copies of $G_{X,1}$ shifted by $(-2,1)$ and $(0,1)$, together with one copy of $G_{X,2}$ shifted by $(1,0)$.
Similarly, each local $Z$-stabilizer is the product of three  $Z$-gauge checks: two copies of $G_{Z,1}$ shifted by $(1,-2)$ and $(1,0)$, together with one copy of $G_{Z,2}$ shifted by $(0,1)$.
The resulting stabilizers have weight $12$ and are shown in Figs.~\ref{fig: 2(a)} and~\ref{fig: 2(b)}.
In addition, Figs.~\ref{fig: 2(c)} and~\ref{fig: 2(d)} show representative dressed logical operators of minimum weight $5$.
There are $10$ independent dressed logical-$X$ and logical-$Z$ operators, respectively.
Thus, the subsystem code encodes $k=10$ logical qubits, and its code distance is $d=5$.

\begin{figure*}[thb]
    \centering
    \subfigure[CNOT gates $U_1$]{\includegraphics[width=0.22\linewidth]{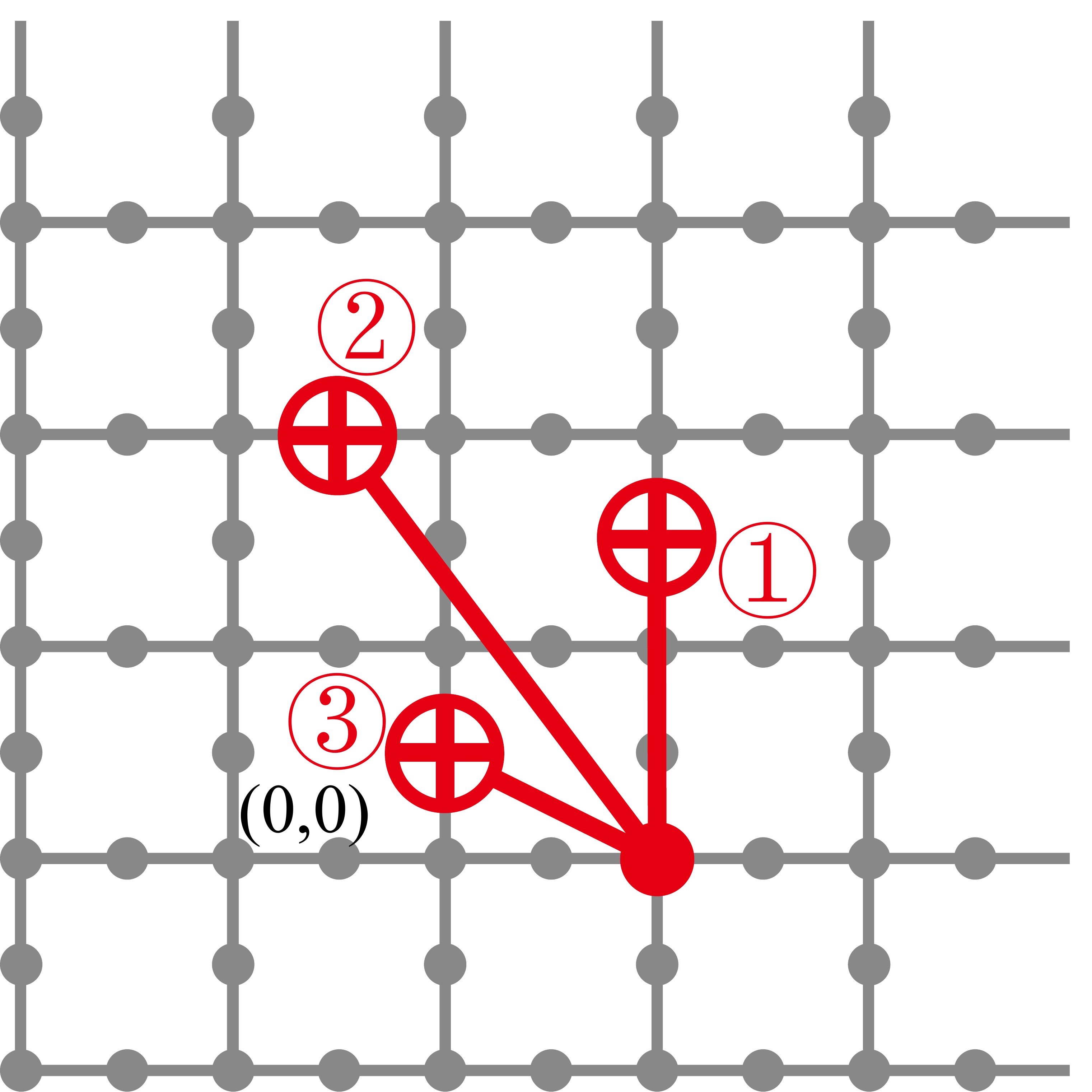}}
    \hspace{0.25cm}
    \subfigure[CNOT gates $U_2$]{\includegraphics[width=0.22\linewidth]{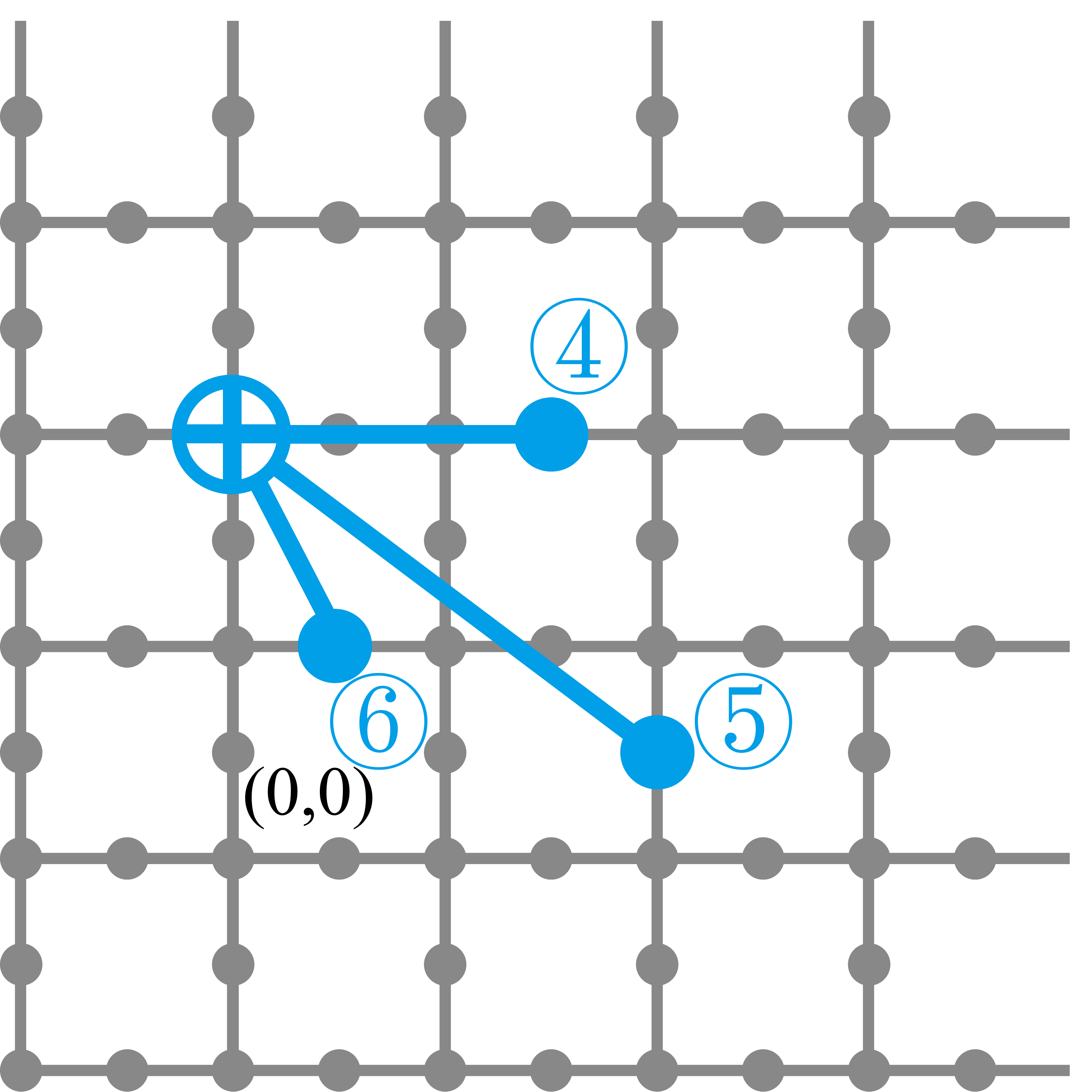}}
    \hspace{0.25cm}
    \subfigure[$X$-type stabilizer after applying $U$]{\includegraphics[width=0.22\linewidth]{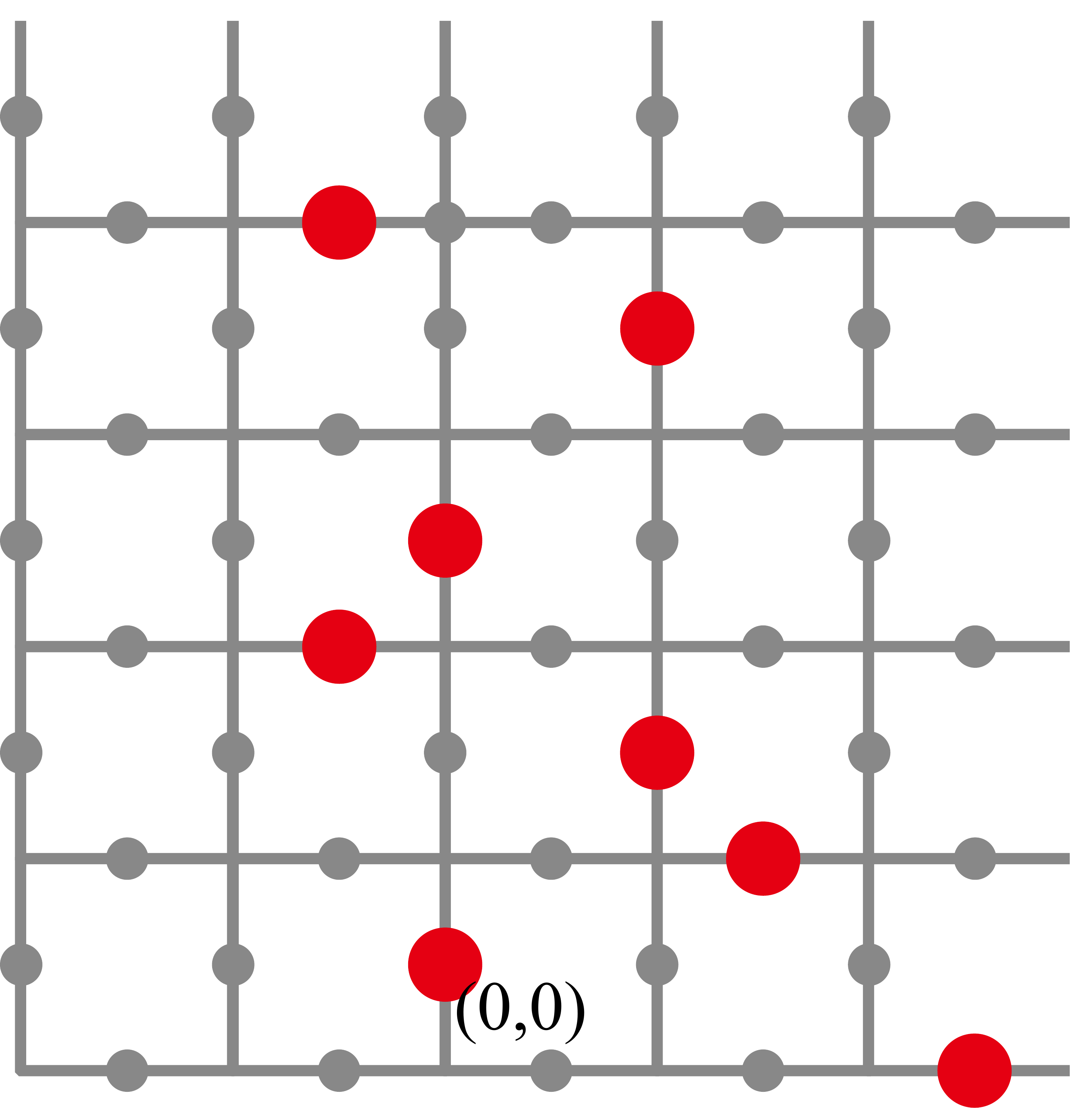}}
    \hspace{0.25cm}
    \subfigure[$Z$-type stabilizer after applying $U$]{\includegraphics[width=0.22\linewidth]{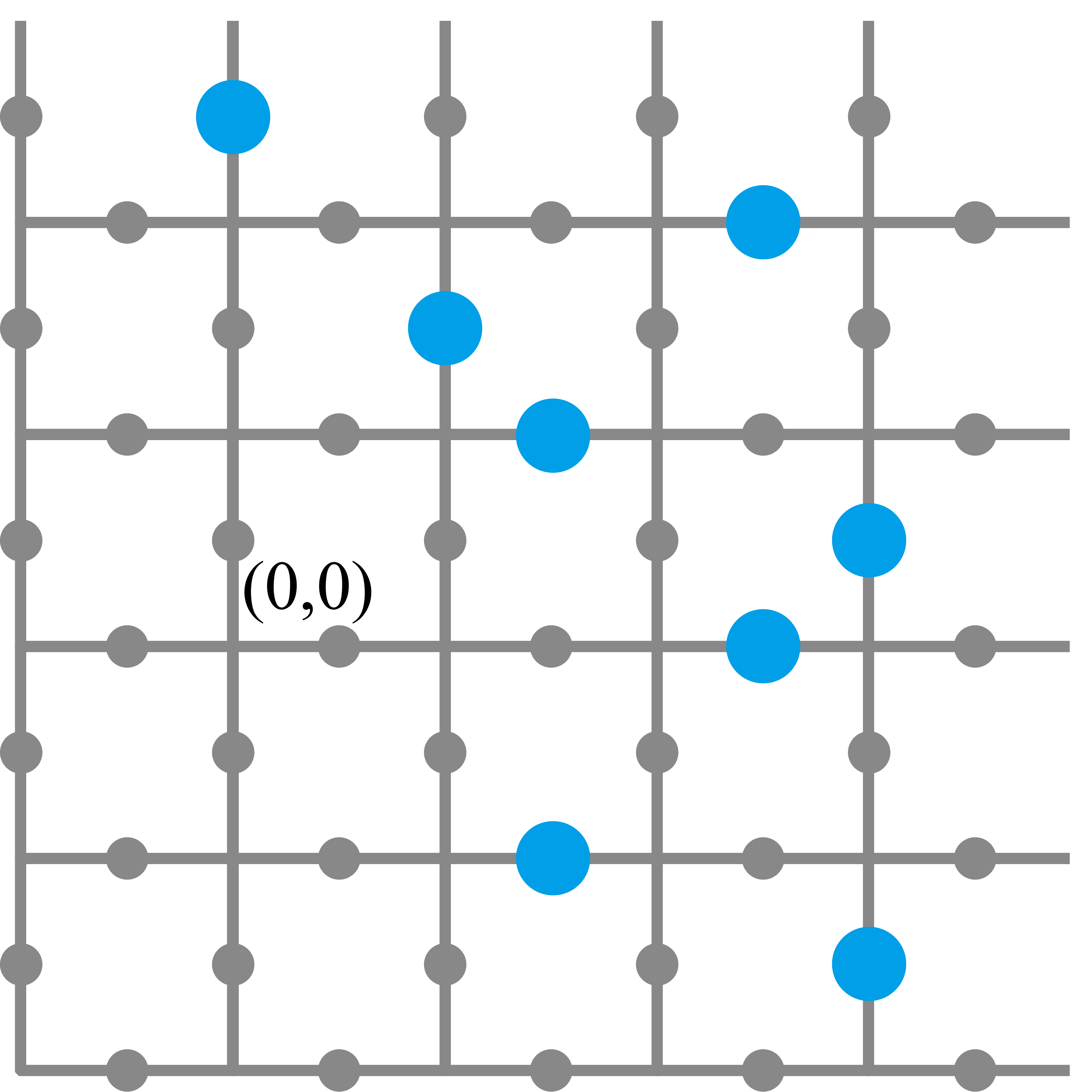}}
    \caption{
    Clifford reduction of the $[[75,10,5]]$ SBB code.
    (a)(b) The circuit $U=U_2U_1$ consists of two translation-invariant layers of CNOT gates.
    Within each layer, the CNOT gates mutually commute, so both $U_1$ and $U_2$ are depth-3 circuits.
    The circuit maps the gauge pair $G_{X,1},G_{Z,1}$ to a single-qubit Pauli $X,Z$ pair on the vertex gauge qubit.
    These gauge qubits can then be decoupled from the protected subsystem.
    (c)(d) After removing the vertex gauge qubits, the transformed stabilizers act on the edge qubits and have weight $8$, reduced from weight $12$ in the subsystem description.
    On the $5\times5$ lattice, the resulting stabilizer code has parameters $[[50,10,5]]$.
    }
    \label{fig: clifford gate and stabilizers}
\end{figure*}

To relate this subsystem construction to a more familiar stabilizer code description, consider the finite-depth Clifford circuit shown in Fig.~\ref{fig: clifford gate and stabilizers}.
This circuit maps the local gauge pair $G_{X,1},G_{Z,1}$ to a single-qubit Pauli $X,Z$ pair on the vertex qubit in each unit cell.
These vertex qubits are gauge qubits and can therefore be decoupled from the protected logical subsystem.
After removing them, the remaining edge qubits support an induced BB stabilizer code with $50$ physical qubits, weight-$8$ stabilizers, and $k=10$ logical qubits.
For this example, the induced stabilizer code has parameters $[[50,10,5]]$.

Thus, the $[[75,10,5]]$ SBB code realizes the logical content and distance of a BB code, but exposes only local weight-$4$ gauge checks for direct measurement.
Compared with the $[[75,2,5]]$ subsystem surface code at the same block length and distance, it encodes 5 times more logical qubits.
We now turn to the general construction underlying this example.

%%%%%%%%%%%%%%%%%%%%%%%%%%%%%%%%%%%%%%%%%%%%%%%%%%%%%%%%%%%%%
\prlsection{General construction and SBB correspondence}
\label{sec: Weight-4 subsystem bivariate bicycle codes}
We now use the Laurent-polynomial formalism reviewed in Appendix~\ref{app:laurent_formalism} to derive the local stabilizers, exclude nonlocal stabilizers on finite tori, and identify the corresponding BB code.

The square lattice has three qubits per unit cell.
We consider CSS gauge generators with two $X$-type and two $Z$-type generators per unit cell.
In the Laurent-polynomial representation, the first three components record the $X$ support and the last three components record the $Z$ support:
\begin{eqs}
    G_{X,1} =
    \left[\begin{array}{c}
        \rule{0pt}{1.1em}f_1(x,y) \\
        \rule{0pt}{1.1em}g_1(x,y) \\
        \rule{0pt}{1.1em}h_1(x,y) \\
        \hline
        0 \\
        0 \\
        0
    \end{array}\right],
    \qquad
    G_{X,2} =
    \left[\begin{array}{c}
        \rule{0pt}{1.1em}f_2(x,y) \\
        \rule{0pt}{1.1em}g_2(x,y) \\
        \rule{0pt}{1.1em}h_2(x,y) \\
        \hline
        0 \\
        0 \\
        0
    \end{array}\right], \\
    G_{Z,1} =
    \left[\begin{array}{c}
        0 \\
        0 \\
        0 \\
        \hline
        \rule{0pt}{1.1em} f_3(x,y) \\
        \rule{0pt}{1.1em} g_3(x,y) \\
        \rule{0pt}{1.1em}h_3(x,y)
    \end{array}\right],
    \qquad
    G_{Z,2} =
    \left[\begin{array}{c}
        0 \\
        0 \\
        0 \\
        \hline
        \rule{0pt}{1.1em} f_4(x,y) \\
        \rule{0pt}{1.1em} g_4(x,y) \\
        \rule{0pt}{1.1em} h_4(x,y)
    \end{array}\right],
    \label{eq: general_subsystem_gauge_operators}
\end{eqs}
where $f_i(x,y), g_i(x,y), h_i(x,y)$ belong to the \textbf{Laurent polynomial ring} $$R=\mathbb{Z}_2[x^{\pm1},y^{\pm1}].$$
%The polynomial arguments $(x,y)$ are suppressed below.
For two Pauli operators represented by Laurent-polynomial vectors $v_1$ and $v_2$, define the product
\begin{eqs}
    v_1\cdot v_2 :=\overline{v}_1^{\mathsf T}\Lambda v_2,
\end{eqs}
where $\Lambda$ is the standard $6\times6$ symplectic form, and the antipode map is defined by
\begin{eqs}
 x^a y^b \mapsto \overline{x^a y^b}:=x^{-a}y^{-b}.
\end{eqs}
The constant term of $v_1\cdot v_2$ determines whether the two Pauli operators commute, while the full Laurent polynomial records the commutators between $v_1$ and all lattice translates of $v_2$.

The commutation data between the two $X$-type and two $Z$-type gauge generators are encoded in the $2\times 2$ matrix
\begin{eqs}
M_c
% &=
% \begin{pmatrix}
% G_{X,1}\\
% G_{X,2}
% \end{pmatrix}
% \cdot
% \begin{pmatrix}
% G_{Z,1} & G_{Z,2}
% \end{pmatrix}\\
&=
\begin{pmatrix}
G_{X,1}\cdot G_{Z,1} & G_{X,1}\cdot G_{Z,2} \\
G_{X,2}\cdot G_{Z,1} & G_{X,2}\cdot G_{Z,2}
\end{pmatrix}
\in M_2(R). 
\label{eq: commutation matrix M_c}
\end{eqs}
Stabilizers are linear combinations of gauge generators that commute with all gauge generators. $X$-type stabilizers come from the left kernel of $M_c$, while $Z$-type stabilizers come from the right kernel of $M_c$:
\begin{lemma}[Kernel-stabilizer correspondence]
\label{lemma: subsystem code}
If $u=(u_1,u_2)$ satisfies $uM_c=0$, then
\begin{eqs}
 S_X=\bar{u}_1G_{X,1}+\bar{u}_2G_{X,2}
\end{eqs}
is an $X$-type stabilizer.
Similarly, if $v=(v_1,v_2)^{\mathsf T}$ satisfies $M_cv=0$, then
\begin{eqs}
 S_Z=v_1G_{Z,1}+v_2G_{Z,2}
\end{eqs}
is a $Z$-type stabilizer.
\end{lemma}
A sufficient condition for such kernels is that the commutation matrix is singular.
Write
\begin{eqs}
M_c=
\begin{pmatrix}
 a & b \\
 c & d
\end{pmatrix}\in M_2(R),
\end{eqs}
and suppose that $\det(M_c) = 0$.
% \begin{eqs}
% \det(M_c)=ad+bc=0 .
% \label{eq: det(M_c)=0}
% \end{eqs}
It follows that
\begin{eqs}
(c,a)M_c=0,\qquad
M_c
\begin{pmatrix}
b\\
a
\end{pmatrix}
=0 .
\end{eqs}
Thus $(c,a)$ gives a left kernel vector and $(b,a)^{\mathsf T}$ gives a right kernel vector.
% The corresponding stabilizers are
% \begin{eqs}
% S_X&=\bar{c}~G_{X,1}+\bar{a}~G_{X,2},\\
% S_Z&=b~G_{Z,1}+a~G_{Z,2}.
% \label{eq: SBB_local_stabilizers}
% \end{eqs}
Equivalently, in the matrix form,
\begin{eqs}
P:=
\begin{pmatrix}
1&0\\
c&a
\end{pmatrix},~~
Q:=
\begin{pmatrix}
1&b\\
0&a
\end{pmatrix}
\Rightarrow
PM_cQ=
\begin{pmatrix}
a&0\\
0&0
\end{pmatrix}.
\end{eqs}
The vanishing second row and second column encode the stabilizer combinations above.
If the pivot $a(x,y)$ is a monomial in $R$, then $P$ and $Q$ are invertible, since their determinants are units in $R$. Hence, replacing $G_{X,2}$ and $G_{Z,2}$ by $S_X$ and $S_Z$ is merely a change of basis:
\begin{eqs}
    \langle G_{X,1},G_{Z,1},G_{X,2},G_{Z,2}\rangle
    =
    \langle G_{X,1},G_{Z,1},S_X,S_Z\rangle .
    \label{eq: same gauge group}
\end{eqs}
%If another entry of $M_c$ is a monomial, we can relabel the $X$- and $Z$-type gauge generators and use that entry as the pivot. 
% More generally, the existence of such a basis change is equivalent to the four entries of $M_c$ generating the unit ideal. This gives the following criterion, proved in Appendix~\ref{app: proof of Lemma monomial condition}.
More generally, whether such a basis change exists is characterized by the following lemma, proved in Appendix~\ref{app: proof of Lemma monomial condition}.
\begin{lemma}[Unimodular-entry reduction criterion]
\label{lemma: monomial condition}
Let $M\in M_2(R)$ satisfy $M\neq 0$ and $\det(M)=0$.
For $M=\begin{pmatrix}
 a & b \\
 c & d
\end{pmatrix}$, define the \textbf{entry ideal}
\begin{equation}
I_1(M):=\langle a, b, c, d\rangle \subseteq R
\end{equation}
to be the ideal generated by the entries of $M$. Then the following conditions are equivalent:
\begin{enumerate}
    \item There exist invertible matrices $P,Q\in GL_2(R)$ such that
    \begin{equation}
        PMQ =
        \begin{pmatrix}
            1&0\\
            0&0
        \end{pmatrix}~.
    \end{equation}
    \item The entries of $M$ generate the unit ideal, i.e.,
    \begin{equation}
        I_1(M)=R.
    \end{equation}
\end{enumerate}
\end{lemma}
%We say that $M$ satisfies the \textbf{unimodular-entry condition} when $I_1(M)=R$.
For a general $p\times q$ matrix $M$ of generic rank $r$, the analogous condition is $I_r(M)=R$, where $I_r(M)$ is the ideal generated by all $r\times r$ minors of $M$.
In that case, invertible row and column operations bring $M$ to the canonical rank-$r$ form; see Appendix~\ref{app: Generalization of Lemma monomial condition}.

%$\begin{psmallmatrix} I_r&0\\0&0 \end{psmallmatrix}$

Thus far, we have worked on the infinite plane and identified local stabilizers. On a finite torus, additional stabilizers may appear with support winding around noncontractible cycles; we refer to them as \textbf{nonlocal stabilizers}. The Bacon-Shor code provides a standard example~\cite{Bacon2006Operator}.
The condition $I_1(M)=R$ rules out such stabilizers, as stated in the following theorem, proved in Appendix~\ref{app: proof of Theorem nonlocal stabilizer}.
\begin{theorem}[Entry-ideal criterion for nonlocal stabilizers]
\label{thm: nonlocal stabilizer}
Let $M\in M_2(R)$ satisfy $M\neq 0$ and $\det(M)=0$.
For an $n\times m$ torus, define
\begin{equation}
R_{n,m}:=R/(x^n-1,y^m-1),
\end{equation}
and let $M^{(n,m)}$ denote the image of $M$ in $M_2(R_{n,m})$.
Let $I_1(M)$ be the entry ideal of $M$.
\begin{enumerate}
    \item If $I_1(M)=R$, then, for every choice of positive integers $n$ and $m$, imposing the periodic boundary conditions $x^n=y^m=1$ does not create kernel vectors beyond those generated by the infinite-plane local kernels. Equivalently, no additional nonlocal stabilizers appear.
    \item If $I_1(M)$ is a proper ideal of $R$, then there exist positive integers $n$ and $m$ such that the finite-torus kernel of $M^{(n,m)}$ is strictly larger than the reduction of the infinite-plane local kernel. Any nonzero Pauli operator obtained from such an extra kernel vector is an additional nonlocal stabilizer.
\end{enumerate}
\end{theorem}
For a general $p\times q$ matrix $M$ of generic rank $r$, the no-nonlocal-stabilizer condition is $I_r(M)=R$; see Appendix~\ref{app: Generalization of Theorem nonlocal stabilizer}.

In what follows, we focus on nonzero singular $2\times 2$ commutation matrices $M_c$ satisfying $I_1(M_c)=R$, i.e., no nonlocal stabilizers.
%By Theorem~\ref{thm: nonlocal stabilizer}, this condition rules out nonlocal stabilizers on any finite torus; hence, the stabilizer group is generated by the local kernels already identified on the infinite plane.
Lemma~\ref{lemma: monomial condition} allows an invertible change of gauge-generator bases such that
\begin{equation}
    M_c =
    \begin{pmatrix}
        1 & 0 \\
        0 & 0
    \end{pmatrix}.
\label{eq: canonical Mc}
\end{equation}
We therefore work in this canonical basis. The next theorem shows that the remaining noncommuting gauge pair can then be decoupled by a finite-depth Clifford circuit.

\begin{theorem}[SBB-BB Clifford correspondence]
\label{thm: exist Clifford gate}
When $M_c$ is in the canonical form~\eqref{eq: canonical Mc}, there exists a finite-depth Clifford circuit $U$ that maps the noncommuting gauge pair in each unit cell to a single-qubit Pauli pair:
\begin{eqs}
    U G_{X,1}(\mathbf r) U^\dagger &= X_{q(\mathbf r)},\\
    U G_{Z,1}(\mathbf r) U^\dagger &= Z_{q(\mathbf r)},
\end{eqs}
where $q(\mathbf r)$ denotes the corresponding gauge qubit in unit cell $\mathbf r$.
The remaining transformed gauge generators have no support on the qubits $q(\mathbf r)$ and generate the stabilizer group on the complementary qubits. Hence, the qubits $q(\mathbf r)$ are disentangled gauge qubits; removing them yields the corresponding BB stabilizer code.
\end{theorem}
The proof is given in Appendix~\ref{app: proof of Theorem exist Clifford gate}.

Theorem~\ref{thm: exist Clifford gate} identifies the protected subsystem of the SBB code with the corresponding BB code. Indeed, the qubits $q(\mathbf r)$ support only the canonical gauge pairs $X_{q(\mathbf r)},Z_{q(\mathbf r)}$, so removing them deletes only gauge degrees of freedom and leaves the logical dimension $k$ unchanged.
The code distance, however, need not be invariant under a finite-depth Clifford circuit. To constrain code distances, we therefore impose the topological order condition: a subsystem code is topological if it has no nontrivial local logical operator~\cite{Bombin2010Topologicalsubsystemcodes, Bombin2012Universaltopologicalphase, bombin_Stabilizer_14}.\footnote{Mathematically, if $\mathcal S_{\mathrm{loc}}$ and $\mathcal G$ denote the groups generated by local stabilizers and gauge operators, respectively, then the topological order condition is
$\mathcal Z(\mathcal S_{\mathrm{loc}})=\mathcal G$,
up to phases. Here $\mathcal Z(\mathcal S)$ denotes local Pauli operators commuting with all elements of $\mathcal S$. For an ordinary stabilizer code, this reduces to
$\mathcal Z(\mathcal S)=\mathcal S$.}

Since a finite-depth Clifford circuit maps local Pauli operators to local Pauli operators, and since any support on gauge qubits $q(\mathbf r)$ can be removed by multiplying by the gauge operators
$X_{q(\mathbf r)}$ and $Z_{q(\mathbf r)}$, nontrivial local logical operators exist in the SBB description if and only if they exist in the corresponding BB description. This gives the following corollary.
\begin{corollary}[Preservation of topological order]
\label{cor:topological_iff}
Under the Clifford correspondence of Theorem~\ref{thm: exist Clifford gate},
the SBB code is a topological subsystem code if and only if the corresponding
BB code is a topological stabilizer code.
\end{corollary}

%%%%%%%%%%%%%%%%%%%%%%%%%%%%%%%%%%%%%%%%%%%%%%%%%%%%%%%%
\prlsection{Search for reflection-symmetric SBB codes}
\label{sec:search_symmetric_sbb}
We now describe a finite search for SBB codes with a combined reflection symmetry: spatial reflection about the diagonal $y=x$ accompanied by the exchange of $X$- and $Z$-gauge generators.
This constraint reduces the search space and can support fold-transversal Clifford gates~\cite{Eberhardt2025Logical}.
Thus, for an $X$-gauge generator
\begin{equation}
    G_X(f,g,h)=(f,g,h,0,0,0)^{\mathsf T},
\end{equation}
the corresponding $Z$-gauge generator is fixed to be
\begin{equation}
    G_Z(f,g,h)
    =
    (0,0,0,f^\sigma,h^\sigma,g^\sigma)^{\mathsf T},
\end{equation}
where $p^\sigma (x,y):=p(y,x)$.
Each candidate is therefore specified by two $X$-type generators
\begin{equation*}
    G_{X,1}=G_X(f_1,g_1,h_1),
    \qquad
    G_{X,2}=G_X(f_2,g_2,h_2),
\end{equation*}
together with their reflected partners $G_{Z,1}$ and $G_{Z,2}$.

For the first generator, its reflected self-pairing is
\begin{equation}
    a
    :=
    G_{X,1}\cdot G_{Z,1}
    =
    \bar f_1 f_1^\sigma
    +
    \bar g_1 h_1^\sigma
    +
    \bar h_1 g_1^\sigma .
\end{equation}
We require $G_{X,1}$ to have weight 4 and $a=1$.\footnote{Generally, one could allow $I_1(M_c)= R$. We impose the stronger condition $a=1$ to simplify the search; the resulting search space is still large enough to produce good SBB code examples.} By Proposition~\ref{prop:weight_four_monomial_pivot}, proved in the Appendix~\ref{app:reflection_symmetric_search}, this forces $f_1$ to be a monomial. We remove the resulting translation freedom by setting $f_1 = 1$, implying the mixed term vanishes,
\begin{equation}
    \bar g_1 h_1^\sigma
    +
    \bar h_1 g_1^\sigma
    =
    0.
\end{equation}
It remains to enumerate the possible pairs $(g_1,h_1)$ with $\operatorname{wt}(g_1)+\operatorname{wt}(h_1)=3$.
Up to interchanging $g_1$ and $h_1$, $G_{X,1}$ falls into one of two cases:
\begin{enumerate}
    \item[(i)] $g_1=0$ and $\operatorname{wt}(h_1)=3$;
    \item[(ii)] $\operatorname{wt}(g_1)=1$ and $h_1=u g_1$, with $u\in R^\rho$ and $\operatorname{wt}(u)=2$.
\end{enumerate}
Here $p^\rho(x,y):=p(y^{-1},x^{-1})$ and $R^\rho:=\{p\in R:p^\rho=p\}$.
% Every enumerated $G_{X,1}$ satisfies $G_{X,1}\cdot G_{Z,1}=1$ and hence the commutation matrix has a monomial pivot.

The second generator is taken to be an arbitrary weight-4 operator, modulo a common monomial translation.
% \begin{equation}
%     \operatorname{wt}(f_2)+\operatorname{wt}(g_2)+\operatorname{wt}(h_2)=4,
% \end{equation}
% again modulo a common monomial translation.
For each pair $(G_{X,1},G_{X,2})$, we compute the commutation matrix
\begin{equation}
    M_c
    =
    \begin{pmatrix}
        G_{X,1}\cdot G_{Z,1} & G_{X,1}\cdot G_{Z,2}\\
        G_{X,2}\cdot G_{Z,1} & G_{X,2}\cdot G_{Z,2}
    \end{pmatrix}
    =
    \begin{pmatrix}
        1 & b\\
        b^\rho & d
    \end{pmatrix}.
\end{equation}
We retain only candidates satisfying $\det M_c=0$.
This condition gives local kernel stabilizers. Since the upper-left entry is already $1$, the candidate has the SBB-BB Clifford correspondence.
The stabilizers are
\begin{equation}
    S_X=b^\sigma \,G_{X,1}+G_{X,2},
    \qquad
    S_Z=b\,G_{Z,1}+G_{Z,2}.
\end{equation}

We evaluate the candidates on twisted tori\footnote{The preceding results were stated for standard tori with $x^n=y^m=1$. The arguments extend directly to twisted tori specified by two independent lattice translation vectors.} with translation vectors~\cite{liang2025generalized}
\begin{equation}
    \mathbf a_1=(0,m),
    \qquad
    \mathbf a_2=(\ell,q),
    \qquad
    0\le q<m .
\end{equation}
Laurent polynomials are reduced in the quotient ring
\begin{equation}
    %R_{m,\ell,q}=
    R/(y^m-1,\;x^\ell y^q-1).
\end{equation}
The torus contains $n=3m\ell$ physical qubits.
For each twisted torus, we construct the finite binary matrices
$H_{G_X}$ and $H_{G_Z}$ from all translates of the $X$- and $Z$-type gauge
generators, and $H_{S_X}$ and $H_{S_Z}$ from all translates of the stabilizers. The number of protected
logical qubits is computed as
\begin{equation}
\begin{aligned}
    k
    &=
    n-\operatorname{rank}_{\mathbb Z_2}(H_{G_X})
      -\operatorname{rank}_{\mathbb Z_2}(H_{S_Z}) \\
    &=
    n-\operatorname{rank}_{\mathbb Z_2}(H_{S_X})
      -\operatorname{rank}_{\mathbb Z_2}(H_{G_Z}) .
\end{aligned}
\end{equation}

Finally, we compute the dressed distance. Dressed logical $X$ operators are
obtained from the CSS code with check matrices $(H_{G_X},H_{S_Z})$, while
dressed logical $Z$ operators are obtained from $(H_{S_X},H_{G_Z})$. The minimum weight among nontrivial dressed logical operators defines the SBB code
distance $d$.
%giving subsystem parameters $[[n,k,d]]$ on the chosen twisted torus.

An explicit formula for the Clifford circuit mapping these SBB codes to their
corresponding BB codes is given in Appendix~\ref{app:reflection_symmetric_search}. Representative search results
are summarized in Table~\ref{tab:code_comparison} and
Appendix~\ref{app: examples_verification}.
% %%%%%%%%%%%%%%%%%%%%%%%%%%%%%%%%%%%%%%%%%%%%%%%%%%%%%%%%%%%%%%%%%%%%%%%%%%%%%%%%%%%%%%%%%%%%%%%%%%%%%%%%%%%%%%%%

%%%%%%%%%%%%%%%%%%%%%%%%%%%%%%%%%%%%%%%%%
\prlsection{Summary and outlook}
We introduced subsystem bivariate bicycle codes that realize the logical structure of BB codes through local gauge measurements.
In the weight-$4$ constructions studied here, products of gauge checks generate local stabilizers, and a determinantal-ideal condition excludes nonlocal stabilizers.
Under these conditions, a finite-depth Clifford circuit decouples the gauge qubits and identifies the protected subsystem with an associated BB stabilizer code.
Thus, within this class, the SBB and BB descriptions have the same logical dimension and satisfy the topological order condition simultaneously.

Several directions remain open. First, efficient syndrome-extraction circuits should be designed and analyzed under realistic circuit-level noise, with attention to measurement depth, wire crossings, correlated faults, and practical decoding strategies. 
Second, understanding whether the observed low-overhead patterns extend systematically to larger sizes is worth studying.
Third, candidates with nonlocal stabilizers, which are excluded in our SBB correspondence, may have their own logical structure and deserve further study.
Finally, allowing more gauge qubits per unit cell may lead to even lower-weight gauge generators or new SBB code families.

\prlsection{Acknowledgements}
We would like to thank Nikolas P. Breuckmann, Jens Niklas Eberhardt, and Zongyuan Wang for their valuable discussions.
Y.-A.C. is supported by the National Natural Science Foundation of China (Grant No.~12474491) and by the Fundamental Research Funds for the Central Universities, Peking University.

\bibliography{bibliography}

@article{haah_module_13,
	author = {Haah, Jeongwan},
	da = {2013/12/01},
	doi = {10.1007/s00220-013-1810-2},
	id = {Haah2013},
	isbn = {1432-0916},
	journal = {Communications in Mathematical Physics},
	number = {2},
	pages = {351--399},
	title = {Commuting Pauli Hamiltonians as Maps between Free Modules},
	ty = {JOUR},
	url = {https://doi.org/10.1007/s00220-013-1810-2},
	volume = {324},
	year = {2013}}

@article{bombin_Stabilizer_14,
   title={Structure of 2D Topological Stabilizer Codes},
   volume={327},
   ISSN={1432-0916},
   url={http://dx.doi.org/10.1007/s00220-014-1893-4},
   DOI={10.1007/s00220-014-1893-4},
   number={2},
   journal={Communications in Mathematical Physics},
   publisher={Springer Science and Business Media LLC},
   author={Bombín, Héctor},
   year={2014},
   month=Mar, pages={387–432} }

@article{iqbal2024NonAbelian,
   title={Non-Abelian topological order and anyons on a trapped-ion processor},
   volume={626},
   ISSN={1476-4687},
   url={http://dx.doi.org/10.1038/s41586-023-06934-4},
   DOI={10.1038/s41586-023-06934-4},
   number={7999},
   journal={Nature},
   publisher={Springer Science and Business Media LLC},
   author={Iqbal, Mohsin and Tantivasadakarn, Nathanan and Verresen, Ruben and Campbell, Sara L. and Dreiling, Joan M. and Figgatt, Caroline and Gaebler, John P. and Johansen, Jacob and Mills, Michael and Moses, Steven A. and Pino, Juan M. and Ransford, Anthony and Rowe, Mary and Siegfried, Peter and Stutz, Russell P. and Foss-Feig, Michael and Vishwanath, Ashvin and Dreyer, Henrik},
   year={2024},
   month=feb, pages={505–511} }

@article{Bombin2010Topologicalsubsystemcodes,
  title = {Topological subsystem codes},
  author = {Bombin, H.},
  journal = {Phys. Rev. A},
  volume = {81},
  issue = {3},
  pages = {032301},
  numpages = {15},
  year = {2010},
  month = {Mar},
  publisher = {American Physical Society},
  doi = {10.1103/PhysRevA.81.032301},
  url = {https://link.aps.org/doi/10.1103/PhysRevA.81.032301}
}

@article{Bombin2012Universaltopologicalphase,
doi = {10.1088/1367-2630/14/7/073048},
url = {https://doi.org/10.1088/1367-2630/14/7/073048},
year = {2012},
month = {jul},
publisher = {IOP Publishing},
volume = {14},
number = {7},
pages = {073048},
author = {Bombin, H and Duclos-Cianci, Guillaume and Poulin, David},
title = {Universal topological phase of two-dimensional stabilizer codes},
journal = {New Journal of Physics}
}

@misc{bravyi1998quantum,
      title={Quantum codes on a lattice with boundary}, 
      author={S. B. Bravyi and A. Yu. Kitaev},
      year={1998},
      eprint={quant-ph/9811052},
      archivePrefix={arXiv},
      primaryClass={quant-ph},
      url={https://arxiv.org/abs/quant-ph/9811052}, 
}

@article{dennis2002topological,
    author = {Dennis, Eric and Kitaev, Alexei and Landahl, Andrew and Preskill, John},
    title = "{Topological quantum memory}",
    journal = {Journal of Mathematical Physics},
    volume = {43},
    number = {9},
    pages = {4452-4505},
    year = {2002},
    month = {09},
    issn = {0022-2488},
    doi = {10.1063/1.1499754},
    url = {https://doi.org/10.1063/1.1499754},
}

@article{kitaev2003fault,
title = {Fault-tolerant quantum computation by anyons},
journal = {Annals of Physics},
volume = {303},
number = {1},
pages = {2-30},
year = {2003},
issn = {0003-4916},
doi = {https://doi.org/10.1016/S0003-4916(02)00018-0},
url = {https://www.sciencedirect.com/science/article/pii/S0003491602000180},
author = {A.Yu. Kitaev}
}

@article{google2023suppressing,
	author = {Google Quantum AI and Collaborators},
	da = {2023/02/01},
	doi = {10.1038/s41586-022-05434-1},
	id = {Acharya2023},
	isbn = {1476-4687},
	journal = {Nature},
	number = {7949},
	pages = {676--681},
	title = {Suppressing quantum errors by scaling a surface code logical qubit},
	ty = {JOUR},
	url = {https://doi.org/10.1038/s41586-022-05434-1},
	volume = {614},
	year = {2023}}

@article{bluvstein2022quantum,
   title={A quantum processor based on coherent transport of entangled atom arrays},
   volume={604},
   ISSN={1476-4687},
   url={http://dx.doi.org/10.1038/s41586-022-04592-6},
   DOI={10.1038/s41586-022-04592-6},
   number={7906},
   journal={Nature},
   publisher={Springer Science and Business Media LLC},
   author={Bluvstein, Dolev and Levine, Harry and Semeghini, Giulia and Wang, Tout T. and Ebadi, Sepehr and Kalinowski, Marcin and Keesling, Alexander and Maskara, Nishad and Pichler, Hannes and Greiner, Markus and Vuletić, Vladan and Lukin, Mikhail D.},
   year={2022},
   month=apr, pages={451–456} }

@article{Ellison2023paulitopological,
  doi = {10.22331/q-2023-10-12-1137},
  url = {https://doi.org/10.22331/q-2023-10-12-1137},
  title = {Pauli topological subsystem codes from {A}belian anyon theories},
  author = {Ellison, Tyler D. and Chen, Yu-An and Dua, Arpit and Shirley, Wilbur and Tantivasadakarn, Nathanan and Williamson, Dominic J.},
  journal = {{Quantum}},
  issn = {2521-327X},
  publisher = {{Verein zur F{\"{o}}rderung des Open Access Publizierens in den Quantenwissenschaften}},
  volume = {7},
  pages = {1137},
  month = oct,
  year = {2023}
}

@article{iqbal2023topological,
	author = {Iqbal, Mohsin and Tantivasadakarn, Nathanan and Gatterman, Thomas M. and Gerber, Justin A. and Gilmore, Kevin and Gresh, Dan and Hankin, Aaron and Hewitt, Nathan and Horst, Chandler V. and Matheny, Mitchell and Mengle, Tanner and Neyenhuis, Brian and Vishwanath, Ashvin and Foss-Feig, Michael and Verresen, Ruben and Dreyer, Henrik},
	da = {2024/06/25},
	doi = {10.1038/s42005-024-01698-3},
	id = {Iqbal2024},
	isbn = {2399-3650},
	journal = {Communications Physics},
	number = {1},
	pages = {205},
	title = {Topological order from measurements and feed-forward on a trapped ion quantum computer},
	ty = {JOUR},
	url = {https://doi.org/10.1038/s42005-024-01698-3},
	volume = {7},
	year = {2024}}

@article{terhal2015quantum,
   title={Quantum error correction for quantum memories},
   volume={87},
   ISSN={1539-0756},
   url={http://dx.doi.org/10.1103/RevModPhys.87.307},
   DOI={10.1103/revmodphys.87.307},
   number={2},
   journal={Reviews of Modern Physics},
   publisher={American Physical Society (APS)},
   author={Terhal, Barbara M.},
   year={2015},
   month=apr, pages={307–346} }

@misc{gottesman1997stabilizer,
      title={Stabilizer Codes and Quantum Error Correction}, 
      author={Daniel Gottesman},
      year={1997},
      eprint={quant-ph/9705052},
      archivePrefix={arXiv},
      primaryClass={quant-ph}
}

@article{haah2016algebraic,
  title={Algebraic methods for quantum codes on lattices},
  author={Haah, Jeongwan},
  journal={Revista colombiana de matematicas},
  volume={50},
  number={2},
  pages={299--349},
  year={2016},
  publisher={Universidad Nacional de Colombia y Sociedad Colombiana de Matem{\'a}ticas},
  url = {https://doi.org/10.15446/recolma.v50n2.62214}
}

@article{Poulin2005subsystem,
  title = {Stabilizer Formalism for Operator Quantum Error Correction},
  author = {Poulin, David},
  journal = {Phys. Rev. Lett.},
  volume = {95},
  issue = {23},
  pages = {230504},
  numpages = {4},
  year = {2005},
  month = {Dec},
  publisher = {American Physical Society},
  doi = {10.1103/PhysRevLett.95.230504},
  url = {https://link.aps.org/doi/10.1103/PhysRevLett.95.230504}
}

@article{liang2023extracting,
  title = {Extracting Topological Orders of Generalized Pauli Stabilizer Codes in Two Dimensions},
  author = {Liang, Zijian and Xu, Yijia and Iosue, Joseph T. and Chen, Yu-An},
  journal = {PRX Quantum},
  volume = {5},
  issue = {3},
  pages = {030328},
  numpages = {36},
  year = {2024},
  month = {Aug},
  publisher = {American Physical Society},
  doi = {10.1103/PRXQuantum.5.030328},
  url = {https://link.aps.org/doi/10.1103/PRXQuantum.5.030328}
}

@article{Shor1995Scheme,
  title = {Scheme for reducing decoherence in quantum computer memory},
  author = {Shor, Peter W.},
  journal = {Phys. Rev. A},
  volume = {52},
  issue = {4},
  pages = {R2493--R2496},
  numpages = {0},
  year = {1995},
  month = {Oct},
  publisher = {American Physical Society},
  doi = {10.1103/PhysRevA.52.R2493},
  url = {https://link.aps.org/doi/10.1103/PhysRevA.52.R2493}
}

@article{Steane1996Error,
  title = {Error Correcting Codes in Quantum Theory},
  author = {Steane, A. M.},
  journal = {Phys. Rev. Lett.},
  volume = {77},
  issue = {5},
  pages = {793--797},
  numpages = {0},
  year = {1996},
  month = {Jul},
  publisher = {American Physical Society},
  doi = {10.1103/PhysRevLett.77.793},
  url = {https://link.aps.org/doi/10.1103/PhysRevLett.77.793}
}

@article{Knill1997Theory,
  title = {Theory of quantum error-correcting codes},
  author = {Knill, Emanuel and Laflamme, Raymond},
  journal = {Phys. Rev. A},
  volume = {55},
  issue = {2},
  pages = {900--911},
  numpages = {0},
  year = {1997},
  month = {Feb},
  publisher = {American Physical Society},
  doi = {10.1103/PhysRevA.55.900},
  url = {https://link.aps.org/doi/10.1103/PhysRevA.55.900}
}

@article{Google2023NonAbelian,
	author = {Google Quantum AI and Collaborators},
	da = {2023/06/01},
	doi = {10.1038/s41586-023-05954-4},
	id = {Andersen2023},
	isbn = {1476-4687},
	journal = {Nature},
	number = {7964},
	pages = {264--269},
	title = {Non-Abelian braiding of graph vertices in a superconducting processor},
	ty = {JOUR},
	url = {https://doi.org/10.1038/s41586-023-05954-4},
	volume = {618},
	year = {2023}}

@article{Bravyi2024HighThreshold,
	author = {Bravyi, Sergey and Cross, Andrew W. and Gambetta, Jay M. and Maslov, Dmitri and Rall, Patrick and Yoder, Theodore J.},
	da = {2024/03/01},
	doi = {10.1038/s41586-024-07107-7},
	id = {Bravyi2024},
	isbn = {1476-4687},
	journal = {Nature},
	number = {8005},
	pages = {778--782},
	title = {High-threshold and low-overhead fault-tolerant quantum memory},
	ty = {JOUR},
	url = {https://doi.org/10.1038/s41586-024-07107-7},
	volume = {627},
	year = {2024}}

@article{Cong2024EnhancingTO,
	author = {Cong, Iris and Maskara, Nishad and Tran, Minh C. and Pichler, Hannes and Semeghini, Giulia and Yelin, Susanne F. and Choi, Soonwon and Lukin, Mikhail D.},
	da = {2024/02/20},
	doi = {10.1038/s41467-024-45584-6},
	id = {Cong2024},
	isbn = {2041-1723},
	journal = {Nature Communications},
	number = {1},
	pages = {1527},
	title = {Enhancing detection of topological order by local error correction},
	ty = {JOUR},
	url = {https://doi.org/10.1038/s41467-024-45584-6},
	volume = {15},
	year = {2024}}

@article{Verresen2021PredictionTC,
  title = {Prediction of Toric Code Topological Order from Rydberg Blockade},
  author = {Verresen, Ruben and Lukin, Mikhail D. and Vishwanath, Ashvin},
  journal = {Phys. Rev. X},
  volume = {11},
  issue = {3},
  pages = {031005},
  numpages = {23},
  year = {2021},
  month = {Jul},
  publisher = {American Physical Society},
  doi = {10.1103/PhysRevX.11.031005},
  url = {https://link.aps.org/doi/10.1103/PhysRevX.11.031005}
}

@article{wang2024coprime,
  title={Coprime Bivariate Bicycle Codes and their Properties},
  author={Wang, Ming and Mueller, Frank},
  journal={arXiv preprint arXiv:2408.10001},
  url = {https://arxiv.org/abs/2408.10001},
  year={2024}
}

@article{shaw2024lowering,
  title = {Lowering Connectivity Requirements for Bivariate Bicycle Codes Using Morphing Circuits},
  author = {Shaw, Mackenzie H. and Terhal, Barbara M.},
  journal = {Phys. Rev. Lett.},
  volume = {134},
  issue = {9},
  pages = {090602},
  numpages = {5},
  year = {2025},
  month = {Mar},
  publisher = {American Physical Society},
  doi = {10.1103/PhysRevLett.134.090602},
  url = {https://link.aps.org/doi/10.1103/PhysRevLett.134.090602}
}

@article{maan2024machine,
    author = "Maan, Arshpreet Singh and Paler, Alexandru",
    title = "{Machine learning message-passing for the scalable decoding of QLDPC codes}",
    archivePrefix = "arXiv",
    primaryClass = "quant-ph",
    doi = "10.1038/s41534-025-01033-w",
    journal = "npj Quantum Inf.",
    volume = "11",
    number = "1",
    pages = "78",
    year = "2025"
}

@article{cross2024linear,
  title={Linear-Size Ancilla Systems for Logical Measurements in QLDPC Codes},
  author={Cross, Andrew and He, Zhiyang and Rall, Patrick and Yoder, Theodore},
  journal={arXiv preprint arXiv:2407.18393},
  url={https://arxiv.org/abs/2407.18393},
  year={2024}
}

@article{cowtan2024ssip,
  title={SSIP: automated surgery with quantum LDPC codes},
  author={Cowtan, Alexander},
  journal={arXiv preprint arXiv:2407.09423},
  url={https://arxiv.org/abs/2407.09423},
  year={2024}
}

@article{wolanski2024ambiguity,
  title={Ambiguity Clustering: an accurate and efficient decoder for qLDPC codes},
  author={Wolanski, Stasiu and Barber, Ben},
  journal={arXiv preprint arXiv:2406.14527},
  url={https://arxiv.org/abs/2406.14527},
  year={2024}
}

@article{gong2024toward,
  title={Toward Low-latency Iterative Decoding of QLDPC Codes Under Circuit-Level Noise},
  author={Gong, Anqi and Cammerer, Sebastian and Renes, Joseph M},
  journal={arXiv preprint arXiv:2403.18901},
  url={https://arxiv.org/abs/2403.18901},
  year={2024}
}

@ARTICLE{Eberhardt2025Logical,
  author={Eberhardt, Jens Niklas and Steffan, Vincent},
  journal={IEEE Transactions on Information Theory}, 
  title={Logical Operators and Fold-Transversal Gates of Bivariate Bicycle Codes}, 
  year={2025},
  volume={71},
  number={2},
  pages={1140-1152},
  doi={10.1109/TIT.2024.3521638}}

@article{tiew2024low,
    author = "Tiew, Ryan and Breuckmann, Nikolas P.",
    title = "{Low-Overhead Entangling Gates From Generalised Dehn Twists}",
    doi = "10.1109/TIT.2025.3571197",
    journal = "IEEE Trans. Info. Theor.",
    volume = "71",
    number = "7",
    pages = "5452--5468",
    year = "2025"
}

@misc{landahl2011fault,
      title={Fault-tolerant quantum computing with color codes}, 
      author={Andrew J. Landahl and Jonas T. Anderson and Patrick R. Rice},
      year={2011},
      eprint={1108.5738},
      archivePrefix={arXiv},
      primaryClass={quant-ph},
      url={https://arxiv.org/abs/1108.5738}, 
}

@article{voss2024multivariate,
  title = {Multivariate bicycle codes},
  author = {Voss, Lukas and Xian, Sim Jian and Haug, Tobias and Bharti, Kishor},
  journal = {Phys. Rev. A},
  volume = {111},
  issue = {6},
  pages = {L060401},
  numpages = {6},
  year = {2025},
  month = {Jun},
  publisher = {American Physical Society},
  doi = {10.1103/ll5p-z88p},
  url = {https://link.aps.org/doi/10.1103/ll5p-z88p}
}

@INPROCEEDINGS{Wang2024Bivariate,
  author={Wang, Ming and Mueller, Frank},
  booktitle={2024 IEEE International Conference on Quantum Computing and Engineering (QCE)}, 
  title={Rate Adjustable Bivariate Bicycle Codes for Quantum Error Correction}, 
  year={2024},
  volume={02},
  number={},
  pages={412-413},
  doi={10.1109/QCE60285.2024.10331}}

@article{Google2024surface,
    author = {Google Quantum AI and Collaborators},
	da = {2025/02/01},
	doi = {10.1038/s41586-024-08449-y},
	id = {Acharya2025},
	isbn = {1476-4687},
	journal = {Nature},
	number = {8052},
	pages = {920--926},
	title = {Quantum error correction below the surface code threshold},
	ty = {JOUR},
	url = {https://doi.org/10.1038/s41586-024-08449-y},
	volume = {638},
	year = {2025}}

@article{semeghini2021probing,
	author = {G. Semeghini and H. Levine and A. Keesling and S. Ebadi and T. T. Wang and D. Bluvstein and R. Verresen and H. Pichler and M. Kalinowski and R. Samajdar and A. Omran and S. Sachdev and A. Vishwanath and M. Greiner and V. Vuleti{\'c} and M. D. Lukin},
	doi = {10.1126/science.abi8794},
	journal = {Science},
	number = {6572},
	pages = {1242-1247},
	title = {Probing topological spin liquids on a programmable quantum simulator},
	url = {https://www.science.org/doi/abs/10.1126/science.abi8794},
	volume = {374},
	year = {2021}}

@article{Bravyi2010Tradeoffs,
  title = {Tradeoffs for Reliable Quantum Information Storage in 2D Systems},
  author = {Bravyi, Sergey and Poulin, David and Terhal, Barbara},
  journal = {Phys. Rev. Lett.},
  volume = {104},
  issue = {5},
  pages = {050503},
  numpages = {4},
  year = {2010},
  month = {Feb},
  publisher = {American Physical Society},
  doi = {10.1103/PhysRevLett.104.050503},
  url = {https://link.aps.org/doi/10.1103/PhysRevLett.104.050503}
}

@article{Kovalev2013QuantumKronecker,
  title = {Quantum Kronecker sum-product low-density parity-check codes with finite rate},
  author = {Kovalev, Alexey A. and Pryadko, Leonid P.},
  journal = {Phys. Rev. A},
  volume = {88},
  issue = {1},
  pages = {012311},
  numpages = {13},
  year = {2013},
  month = {Jul},
  publisher = {American Physical Society},
  doi = {10.1103/PhysRevA.88.012311},
  url = {https://link.aps.org/doi/10.1103/PhysRevA.88.012311}
}

@Article{Pryadko2022DistanceGB,
AUTHOR = {Wang, Renyu and Pryadko, Leonid P.},
TITLE = {Distance Bounds for Generalized Bicycle Codes},
JOURNAL = {Symmetry},
VOLUME = {14},
YEAR = {2022},
NUMBER = {7},
ARTICLE-NUMBER = {1348},
URL = {https://www.mdpi.com/2073-8994/14/7/1348},
ISSN = {2073-8994},
DOI = {10.3390/sym14071348}
}

@article{liang2025generalized,
  title = {Generalized Toric Codes on Twisted Tori for Quantum Error Correction},
  author = {Liang, Zijian and Liu, Ke and Song, Hao and Chen, Yu-An},
  journal = {PRX Quantum},
  volume = {6},
  issue = {2},
  pages = {020357},
  numpages = {22},
  year = {2025},
  month = {Jun},
  publisher = {American Physical Society},
  doi = {10.1103/rmy6-9n89},
  url = {https://link.aps.org/doi/10.1103/rmy6-9n89}
}

@article{Fowler2012Surfacecodes,
  title = {Surface codes: Towards practical large-scale quantum computation},
  author = {Fowler, Austin G. and Mariantoni, Matteo and Martinis, John M. and Cleland, Andrew N.},
  journal = {Phys. Rev. A},
  volume = {86},
  issue = {3},
  pages = {032324},
  numpages = {48},
  year = {2012},
  month = {Sep},
  publisher = {American Physical Society},
  doi = {10.1103/PhysRevA.86.032324},
  url = {https://link.aps.org/doi/10.1103/PhysRevA.86.032324}
}

@article{Litinski2019gameofsurfacecodes,
  doi = {10.22331/q-2019-03-05-128},
  url = {https://doi.org/10.22331/q-2019-03-05-128},
  title = {A {G}ame of {S}urface {C}odes: {L}arge-{S}cale {Q}uantum {C}omputing with {L}attice {S}urgery},
  author = {Litinski, Daniel},
  journal = {{Quantum}},
  issn = {2521-327X},
  publisher = {{Verein zur F{\"{o}}rderung des Open Access Publizierens in den Quantenwissenschaften}},
  volume = {3},
  pages = {128},
  month = mar,
  year = {2019}
}

@article{berthusen2025toward,
  title = {Toward a 2D Local Implementation of Quantum Low-Density Parity-Check Codes},
  author = {Berthusen, Noah and Devulapalli, Dhruv and Schoute, Eddie and Childs, Andrew M. and Gullans, Michael J. and Gorshkov, Alexey V. and Gottesman, Daniel},
  journal = {PRX Quantum},
  volume = {6},
  issue = {1},
  pages = {010306},
  numpages = {18},
  year = {2025},
  month = {Jan},
  publisher = {American Physical Society},
  doi = {10.1103/PRXQuantum.6.010306},
  url = {https://link.aps.org/doi/10.1103/PRXQuantum.6.010306}
}

@article{lin2025single,
  title={Single-shot and two-shot decoding with generalized bicycle codes},
  author={Lin, Hsiang-Ku and Liu, Xingrui and Lim, Pak Kau and Pryadko, Leonid P},
  journal={arXiv preprint arXiv:2502.19406},
  url={https://arxiv.org/abs/2502.19406},
  year={2025}
}

@article{Bacon2006Operator,
  title = {Operator quantum error-correcting subsystems for self-correcting quantum memories},
  author = {Bacon, Dave},
  journal = {Phys. Rev. A},
  volume = {73},
  issue = {1},
  pages = {012340},
  numpages = {13},
  year = {2006},
  month = {Jan},
  publisher = {American Physical Society},
  doi = {10.1103/PhysRevA.73.012340},
  url = {https://link.aps.org/doi/10.1103/PhysRevA.73.012340}
}

@misc{liang2025selfdual,
      title={Self-dual bivariate bicycle codes with transversal Clifford gates}, 
      author={Zijian Liang and Yu-An Chen},
      year={2025},
      eprint={2510.05211},
      archivePrefix={arXiv},
      primaryClass={quant-ph},
}

@article{bravyi2013subsystemsurfacecodesthreequbit,
  title     = {Subsystem surface codes with three-qubit check operators},
  author    = {Bravyi, Sergey and Duclos-Cianci, Guillaume and Poulin, David and Suchara, Martin},
  journal   = {Quantum Information \& Computation},
  volume    = {13},
  number    = {11--12},
  pages     = {963--985},
  year      = {2013},
  doi       = {10.26421/QIC13.11-12-4},
  eprint    = {1207.1443},
  archivePrefix = {arXiv},
  primaryClass  = {quant-ph}
}

@article{Higgott2021subsystem,
  title = {Subsystem Codes with High Thresholds by Gauge Fixing and Reduced Qubit Overhead},
  author = {Higgott, Oscar and Breuckmann, Nikolas P.},
  journal = {Phys. Rev. X},
  volume = {11},
  issue = {3},
  pages = {031039},
  numpages = {30},
  year = {2021},
  month = {Aug},
  publisher = {American Physical Society},
  doi = {10.1103/PhysRevX.11.031039},
  url = {https://link.aps.org/doi/10.1103/PhysRevX.11.031039}
}

@misc{liang2026generalizedmathbbzptoriccodes,
      title={Generalized $\mathbb{Z}_p$ toric codes as qudit low-density parity-check codes}, 
      author={Zijian Liang and Yu-An Chen},
      year={2026},
      eprint={2602.20158},
      archivePrefix={arXiv},
      primaryClass={quant-ph},
      url={https://arxiv.org/abs/2602.20158}, 
}

@article{Kribs2005Unified,
  title = {Unified and Generalized Approach to Quantum Error Correction},
  author = {Kribs, David and Laflamme, Raymond and Poulin, David},
  journal = {Phys. Rev. Lett.},
  volume = {94},
  issue = {18},
  pages = {180501},
  numpages = {4},
  year = {2005},
  month = {May},
  publisher = {American Physical Society},
  doi = {10.1103/PhysRevLett.94.180501},
  url = {https://link.aps.org/doi/10.1103/PhysRevLett.94.180501}
}

@article{Kribs2006Operator,
    author = "Kribs, David W. and Laflamme, Raymond and Poulin, David and Lesosky, Maia",
    title = "{Operator quantum error correction}",
    eprint = "quant-ph/0504189",
    archivePrefix = "arXiv",
    doi = "10.26421/QIC6.4-5-6",
    journal = "Quant. Inf. Comput.",
    volume = "6",
    number = "4-5",
    pages = "382--399",
    year = "2006"
}

@article{Bravyi2011Subsystem,
  title = {Subsystem codes with spatially local generators},
  author = {Bravyi, Sergey},
  journal = {Phys. Rev. A},
  volume = {83},
  issue = {1},
  pages = {012320},
  numpages = {9},
  year = {2011},
  month = {Jan},
  publisher = {American Physical Society},
  doi = {10.1103/PhysRevA.83.012320},
  url = {https://link.aps.org/doi/10.1103/PhysRevA.83.012320}
}

@misc{zhang2026coupledlayerconstructionquantumproduct,
      title={Coupled-Layer Construction of Quantum Product Codes}, 
      author={Shuyu Zhang and Tzu-Chieh Wei and Nathanan Tantivasadakarn},
      year={2026},
      eprint={2603.08711},
      archivePrefix={arXiv},
      primaryClass={quant-ph},
      url={https://arxiv.org/abs/2603.08711}, 
}

\onecolumngrid
\appendix
\newpage

\section{Subsystem code preliminaries}\label{app: preliminaries}

In this section, we review basic concepts about stabilizer and subsystem codes~\cite{gottesman1997stabilizer, Poulin2005subsystem, Bacon2006Operator, bravyi2013subsystemsurfacecodesthreequbit, Ellison2023paulitopological}.
Throughout this appendix, all Pauli operators are considered up to overall phases.
For binary vectors $u,v\in \mathbb Z_2^n$, define
\begin{equation}
    X(u)=\bigotimes_{i=1}^n X_i^{u_i},
    \qquad
    Z(v)=\bigotimes_{i=1}^n Z_i^{v_i},
    \qquad
    P(u,v)=X(u)Z(v).
\end{equation}
The support of $P(u,v)$ is the set of qubits on which it acts nontrivially,
\begin{equation}
    \operatorname{supp} P(u,v)
    =
    \{i: u_i=1 \text{ or } v_i=1\},
\end{equation}
and its weight is
\begin{equation}
    \operatorname{wt} P(u,v)
    =
    |\operatorname{supp} P(u,v)|.
\end{equation}
Two Pauli operators commute or anticommute according to the symplectic inner product
\begin{equation}
    \langle (u,v),(u',v')\rangle_s
    =
    u\cdot v' + v\cdot u'
    \pmod 2.
\end{equation}
Equivalently,
\begin{equation}
    P(u,v)P(u',v')
    =
    (-1)^{\langle (u,v),(u',v')\rangle_s}
    P(u',v')P(u,v).
\end{equation}

\subsection{Stabilizer codes}

A stabilizer code is specified by an Abelian subgroup
$\mathcal S\subset \mathcal P_n$ of the $n$-qubit Pauli group, with
$-I\notin \mathcal S$.
The code space is the joint $+1$ eigenspace of all stabilizers:
\begin{equation}
    \mathcal C(\mathcal S)
    =
    \{|\psi\rangle : S|\psi\rangle = |\psi\rangle
    \text{ for all } S\in \mathcal S\}.
\end{equation}
If $\mathcal S$ has $s$ independent generators, then
\begin{equation}
    \dim \mathcal C(\mathcal S)=2^{n-s}.
\end{equation}
Thus a stabilizer code encodes
\begin{equation}
    k=n-s
\end{equation}
logical qubits.

The logical Pauli operators are Pauli operators that commute with all stabilizers,
modulo multiplication by stabilizers.
Let
\begin{equation}
    \operatorname{Cent}(\mathcal S)
    =
    \{P\in \mathcal P_n : PS=SP \text{ for all } S\in\mathcal S\}
\end{equation}
denote the Pauli centralizer of $\mathcal S$.
Then the logical Pauli group is
\begin{equation}
    \mathcal L
    =
    \operatorname{Cent}(\mathcal S)/\mathcal S.
\end{equation}
The distance of a stabilizer code is the minimum weight of a nontrivial logical Pauli operator:
\begin{equation}
    d
    =
    \min_{P\in \operatorname{Cent}(\mathcal S)\setminus \mathcal S}
    \operatorname{wt}(P).
\end{equation}

\subsection{Subsystem codes}

Subsystem codes generalize stabilizer codes by allowing part of the code space to be used as an unprotected gauge subsystem.
The stabilizer code space factors as
\begin{equation}
    \mathcal C(\mathcal S)
    \cong
    \mathcal H_L \otimes \mathcal H_G,
\end{equation}
where $\mathcal H_L$ stores the protected logical qubits and $\mathcal H_G$ stores gauge qubits.
Errors or measurements that affect only $\mathcal H_G$ do not change the protected logical information.

A subsystem code is specified by a generally non-Abelian gauge group
$\mathcal G\subset \mathcal P_n$.
The stabilizer group is the center of the gauge group,
\begin{equation}
    \mathcal S
    =
    \mathcal G \cap \operatorname{Cent}(\mathcal G),
\end{equation}
where
\begin{equation}
    \operatorname{Cent}(\mathcal G)
    =
    \{P\in \mathcal P_n : PG=GP \text{ for all } G\in\mathcal G\}.
\end{equation}
Thus stabilizers are precisely those gauge operators that commute with every gauge operator.

After a suitable choice of independent generators, the gauge group can be written in the canonical form
\begin{equation}
    \mathcal G
    =
    \left\langle
    Z_1,\ldots,Z_s,
    X^g_1,Z^g_1,\ldots,X^g_r,Z^g_r
    \right\rangle,
\end{equation}
where $Z_1,\ldots,Z_s$ generate the stabilizer group, and each pair
$X^g_i,Z^g_i$ acts as a Pauli-$X$ and Pauli-$Z$ operator on the $i$th gauge qubit.
The only nontrivial commutation relations among these canonical gauge generators are
\begin{equation}
    X^g_i Z^g_j
    =
    (-1)^{\delta_{ij}}
    Z^g_j X^g_i.
\end{equation}
If $g=\operatorname{rank}(\mathcal G)$ is the number of independent gauge generators and
$s=\operatorname{rank}(\mathcal S)$ is the number of independent stabilizer generators, then
\begin{equation}
    g=s+2r,
\end{equation}
where $r$ is the number of gauge qubits.
Since
\begin{equation}
    \dim \mathcal C(\mathcal S)
    =
    2^{n-s}
    =
    2^k 2^r,
\end{equation}
the number of protected logical qubits is
\begin{equation}
    k=n-s-r.
\end{equation}
Equivalently,
\begin{equation}
    r=\frac{g-s}{2},
    \qquad
    k=n-\frac{g+s}{2}.
    \label{eq:subsystem_counting_general}
\end{equation}
In this work we quote subsystem-code parameters as $[[n,k,d]]$, suppressing the number of gauge qubits.

Subsystem codes have two closely related notions of logical operators.
Bare logical operators commute with all gauge operators:
\begin{equation}
    \mathcal L_{\mathrm{bare}}
    =
    \operatorname{Cent}(\mathcal G)/\mathcal S.
\end{equation}
These operators act on the protected logical subsystem and act trivially on the gauge subsystem.
Dressed logical operators are Pauli operators that commute with all stabilizers, modulo multiplication by gauge operators:
\begin{equation}
    \mathcal L_{\mathrm{dressed}}
    =
    \operatorname{Cent}(\mathcal S)/\mathcal G.
\end{equation}
Multiplying a dressed logical operator by a gauge operator changes only its action on the gauge subsystem.
Thus two Pauli operators $P$ and $PG$, with $G\in\mathcal G$, represent the same dressed logical operator.

The distance of a subsystem code is the minimum weight of a nontrivial dressed logical operator.
For a dressed logical coset $[P]\in \operatorname{Cent}(\mathcal S)/\mathcal G$, define its weight by
\begin{equation}
    \operatorname{wt}([P])
    =
    \min_{G\in\mathcal G}
    \operatorname{wt}(PG).
\end{equation}
Then
\begin{equation}
    d
    =
    \min_{\substack{[P]\in \operatorname{Cent}(\mathcal S)/\mathcal G\\ [P]\neq I}}
    \operatorname{wt}([P]).
    \label{eq:subsystem_distance_general}
\end{equation}
This dressed distance is the relevant distance for quantum error correction, because errors that differ by a gauge operator have the same effect on the protected logical subsystem.

\subsection{CSS subsystem codes}

The SBB codes considered in this work are CSS subsystem codes.
Their gauge group is generated by $X$-type and $Z$-type Pauli operators separately.
Let
\begin{equation}
    H_X^{\rm g}\in \mathbb Z_2^{m_X\times n},
    \qquad
    H_Z^{\rm g}\in \mathbb Z_2^{m_Z\times n}
\end{equation}
be the binary matrices whose rows specify the $X$- and $Z$-type gauge generators.
The corresponding gauge group is
\begin{equation}
    \mathcal G
    =
    \left\langle
    X(u): u\in \operatorname{row}(H_X^{\rm g})
    \right\rangle
    \cdot
    \left\langle
    Z(v): v\in \operatorname{row}(H_Z^{\rm g})
    \right\rangle .
\end{equation}
Define the row spaces
\begin{equation}
    R_X=\operatorname{row}(H_X^{\rm g}),
    \qquad
    R_Z=\operatorname{row}(H_Z^{\rm g}).
\end{equation}
Unlike in an ordinary CSS stabilizer code, one does not require
$H_X^{\rm g}(H_Z^{\rm g})^T=0$.
Instead, the matrix
\begin{equation}
    \Omega
    =
    H_X^{\rm g}(H_Z^{\rm g})^T
    \label{eq:gauge_anticommutation_matrix}
\end{equation}
records the anticommutation relations between $X$- and $Z$-type gauge generators.
Its entries are
\begin{equation}
    \Omega_{ab}
    =
    \begin{cases}
    0, & \text{if the $a$th $X$-gauge generator commutes with the $b$th $Z$-gauge generator,}\\
    1, & \text{if the $a$th $X$-gauge generator anticommutes with the $b$th $Z$-gauge generator.}
    \end{cases}
\end{equation}

The stabilizers are the gauge operators that commute with all gauge operators.
Therefore the $X$-type stabilizer space is
\begin{equation}
    S_X
    =
    R_X\cap R_Z^\perp
    =
    \{u\in R_X : u\cdot v=0 \text{ for all } v\in R_Z\},
    \label{eq:SX_subsystem}
\end{equation}
and the $Z$-type stabilizer space is
\begin{equation}
    S_Z
    =
    R_Z\cap R_X^\perp
    =
    \{v\in R_Z : u\cdot v=0 \text{ for all } u\in R_X\}.
    \label{eq:SZ_subsystem}
\end{equation}
Equivalently,
\begin{equation}
    S_X
    =
    \{u\in R_X : H_Z^{\rm g}u^T=0\},
    \qquad
    S_Z
    =
    \{v\in R_Z : H_X^{\rm g}v^T=0\}.
\end{equation}
The full stabilizer group is
\begin{equation}
    \mathcal S
    =
    \left\langle
    X(u): u\in S_X
    \right\rangle
    \cdot
    \left\langle
    Z(v): v\in S_Z
    \right\rangle .
\end{equation}

Let
\begin{equation}
    a=\operatorname{rank} H_X^{\rm g},
    \qquad
    b=\operatorname{rank} H_Z^{\rm g},
    \qquad
    \rho=\operatorname{rank}\Omega .
\end{equation}
Then $\rho$ counts the number of independent anticommuting gauge pairs.
The ranks of the $X$- and $Z$-type stabilizer spaces are
\begin{equation}
    \operatorname{rank} S_X = a-\rho,
    \qquad
    \operatorname{rank} S_Z = b-\rho.
\end{equation}
Hence
\begin{equation}
    s=\operatorname{rank}(\mathcal S)
    =
    a+b-2\rho,
\end{equation}
and the number of gauge qubits is
\begin{equation}
    r=\rho.
\end{equation}
Using Eq.~\eqref{eq:subsystem_counting_general}, the number of protected logical qubits is therefore
\begin{equation}
    k
    =
    n-a-b+\rho
    =
    n-\operatorname{rank}H_X^{\rm g}
     -\operatorname{rank}H_Z^{\rm g}
     +\operatorname{rank}\!\left(H_X^{\rm g}(H_Z^{\rm g})^T\right).
    \label{eq:css_subsystem_k}
\end{equation}
When $\Omega=0$, the gauge generators all commute, the gauge group is already a stabilizer group, and Eq.~\eqref{eq:css_subsystem_k} reduces to the usual CSS stabilizer-code formula
\begin{equation}
    k
    =
    n-\operatorname{rank}H_X-\operatorname{rank}H_Z.
\end{equation}

The dressed logical operators also have a simple CSS description.
An $X$-type dressed logical operator must commute with all $Z$-type stabilizers, so its binary vector lies in $S_Z^\perp$.
Two such operators are equivalent if they differ by an $X$-type gauge operator in $R_X$.
Thus
\begin{equation}
    \mathcal L_X^{\mathrm{dressed}}
    =
    S_Z^\perp/R_X.
\end{equation}
Similarly,
\begin{equation}
    \mathcal L_Z^{\mathrm{dressed}}
    =
    S_X^\perp/R_Z.
\end{equation}
The corresponding dressed distances are
\begin{equation}
    d_X
    =
    \min_{\substack{x\in S_Z^\perp\\ x\notin R_X}}
    \operatorname{wt}(x+R_X),
    \qquad
    d_Z
    =
    \min_{\substack{z\in S_X^\perp\\ z\notin R_Z}}
    \operatorname{wt}(z+R_Z),
\end{equation}
where
\begin{equation}
    \operatorname{wt}(x+R_X)
    =
    \min_{r_X\in R_X}\operatorname{wt}(x+r_X),
    \qquad
    \operatorname{wt}(z+R_Z)
    =
    \min_{r_Z\in R_Z}\operatorname{wt}(z+r_Z).
\end{equation}
The subsystem-code distance is
\begin{equation}
    d=\min(d_X,d_Z).
\end{equation}

For comparison, bare logical operators are constrained to commute with all gauge generators.
The $X$- and $Z$-type bare logical spaces are
\begin{equation}
    \mathcal L_X^{\mathrm{bare}}
    =
    R_Z^\perp/S_X,
    \qquad
    \mathcal L_Z^{\mathrm{bare}}
    =
    R_X^\perp/S_Z.
\end{equation}
Bare logical operators are often more restrictive than dressed logical operators, because they must act trivially on the gauge subsystem.

\section{Review of the Laurent-polynomial formalism}
\label{app:laurent_formalism}

\begin{figure*}[thb]
    \centering
    \includegraphics[width=0.75\linewidth]{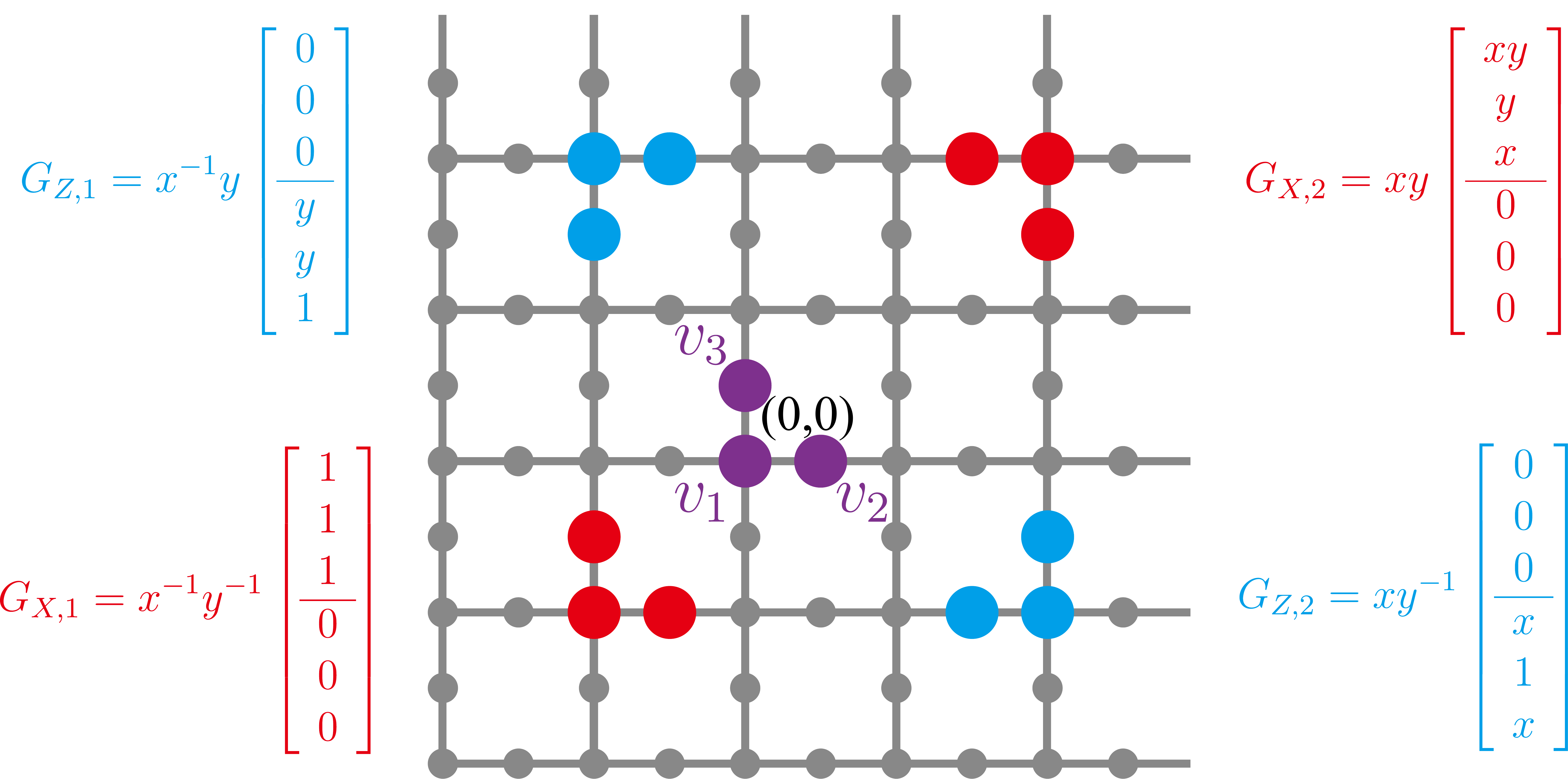}
    \caption{
    Examples of Laurent-polynomial representations of Pauli operators on the square lattice.
    Each unit cell contains three qubit positions, labeled $v_1$, $v_2$, and $v_3$.
    Pauli operators on these qubits are encoded as six-component vectors over
    $\mathbb{Z}_2[x^{\pm1},y^{\pm1}]$, as in Eq.~\eqref{eq: X and Z on v1 and v2 in Appendix}.
    Multiplication by $x^n y^m$ translates an operator by $(n,m)$ lattice units.
    For example, translating the operator supported on the three qubits at the origin to $(-1,-1)$ multiplies its vector by $x^{-1}y^{-1}$.
    Sums of such translated basis vectors give the polynomial-vector representation of gauge generators.
    }
\label{fig:example_poly}
\end{figure*}

For completeness, we review the Laurent-polynomial description of translation-invariant Pauli codes used throughout this work~\cite{haah_module_13, haah2016algebraic, liang2023extracting}.
The general formalism is standard; here we specialize it to qubit CSS subsystem codes on the square lattice with three qubits per unit cell.

We label the three qubits in a unit cell by $v_1$, $v_2$, and $v_3$, as shown in Fig.~\ref{fig:example_poly}.
Ignoring overall phases, the single-qubit Pauli operators are represented by six-component binary vectors:
\begin{equation}
    \mathcal{X}_{v_1} =
    \left[\begin{array}{c}1 \\ 0 \\ 0 \\ \hline 0 \\ 0 \\ 0 \end{array}\right], ~
    \mathcal{Z}_{v_1} =
    \left[\begin{array}{c}0 \\ 0 \\ 0 \\ \hline 1 \\ 0 \\ 0 \end{array}\right], ~
    \mathcal{X}_{v_2} =
    \left[\begin{array}{c}0 \\ 1 \\ 0 \\ \hline 0 \\ 0 \\ 0 \end{array}\right], ~
    \mathcal{Z}_{v_2} =
    \left[\begin{array}{c}0 \\ 0 \\ 0 \\ \hline 0 \\ 1 \\ 0 \end{array}\right], ~
    \mathcal{X}_{v_3} =
    \left[\begin{array}{c}0 \\ 0 \\ 1 \\ \hline 0 \\ 0 \\ 0 \end{array}\right], ~
    \mathcal{Z}_{v_3} =
    \left[\begin{array}{c}0 \\ 0 \\ 0 \\ \hline 0 \\ 0 \\ 1 \end{array}\right].
\label{eq: X and Z on v1 and v2 in Appendix}
\end{equation}
The first three entries record the $X$ support, and the last three entries record the $Z$ support.
Multiplication of Pauli operators corresponds to addition of these vectors over $\mathbb{Z}_2$.

Translation invariance is encoded by two formal variables, $x$ and $y$.
A translation by $(n,m)\in\mathbb{Z}^2$ is represented by multiplication by the monomial $x^n y^m$.
For instance, the Pauli operator $\mathcal{Z}_{v_1}$ translated by $(n,m)$ is represented by
\begin{equation}
    x^n y^m \mathcal{Z}_{v_1}.
\end{equation}
Thus translation-invariant Pauli operators are naturally represented by vectors over the Laurent-polynomial ring
\begin{equation}
    R=\mathbb{Z}_2[x^{\pm1},y^{\pm1}].
\end{equation}

We also use the antipode map
\begin{equation}
    \overline{x^n y^m}=x^{-n}y^{-m},
\end{equation}
extended linearly to all Laurent polynomials.
Throughout this paper, the bar denotes this antipode map, not complex conjugation.

Commutation relations are encoded by the standard symplectic form
\begin{eqs}
    \Lambda=
    \left[\begin{array}{ccc|ccc}
        0 & 0 & 0 & 1 & 0 & 0 \\
        0 & 0 & 0 & 0 & 1 & 0 \\
        0 & 0 & 0 & 0 & 0 & 1 \\
        \hline
        -1 & 0 & 0 & 0 & 0 & 0 \\
        0 & -1 & 0 & 0 & 0 & 0 \\
        0 & 0 & -1 & 0 & 0 & 0
    \end{array}\right]
\end{eqs}
For two Pauli vectors $u,v\in R^6$, we define the symplectic product
\begin{equation}
    u\cdot v := \overline{u}^{\mathsf T}\Lambda v.
\end{equation}
This product is again a Laurent polynomial.
The constant term of $u\cdot v$ determines whether the two Pauli operators commute.
More generally, the coefficient of $x^n y^m$ records the commutation relation between $u$ and a lattice translate of $v$.
Thus $u$ commutes with all translates of $v$ if and only if
\begin{equation}
    u\cdot v=0
\end{equation}
as a Laurent polynomial.

We now describe the gauge-generator data used for translation-invariant CSS subsystem codes.
With two $X$-type and two $Z$-type gauge-generator families per unit cell, the code is specified by twelve Laurent polynomials
$f_i,g_i,h_i\in R$, with $i=1,\ldots,4$:
\begin{eqs}
    G_{X,1} =
    \left[\begin{array}{c}
        \rule{0pt}{1.1em}f_1(x,y) \\
        \rule{0pt}{1.1em}g_1(x,y) \\
        \rule{0pt}{1.1em}h_1(x,y) \\
        \hline
        \rule{0pt}{1.1em}0 \\
        \rule{0pt}{1.1em}0 \\
        \rule{0pt}{1.1em}0
    \end{array}\right],
    \qquad
    G_{X,2} =
    \left[\begin{array}{c}
        \rule{0pt}{1.1em}f_2(x,y) \\
        \rule{0pt}{1.1em}g_2(x,y) \\
        \rule{0pt}{1.1em}h_2(x,y) \\
        \hline
        \rule{0pt}{1.1em}0 \\
        \rule{0pt}{1.1em}0 \\
        \rule{0pt}{1.1em}0
    \end{array}\right], \qquad
    G_{Z,1} =
    \left[\begin{array}{c}
        \rule{0pt}{1.1em}0 \\
        \rule{0pt}{1.1em}0 \\
        \rule{0pt}{1.1em}0 \\
        \hline
        \rule{0pt}{1.1em} f_3(x,y) \\
        \rule{0pt}{1.1em} g_3(x,y) \\
        \rule{0pt}{1.1em} h_3(x,y)
    \end{array}\right],
    \qquad
    G_{Z,2} =
    \left[\begin{array}{c}
        \rule{0pt}{1.1em}0 \\
        \rule{0pt}{1.1em}0 \\
        \rule{0pt}{1.1em}0 \\
        \hline
        \rule{0pt}{1.1em} f_4(x,y) \\
        \rule{0pt}{1.1em} g_4(x,y) \\
        \rule{0pt}{1.1em} h_4(x,y)
    \end{array}\right].
\end{eqs}
Below, we suppress the arguments $(x,y)$ when no confusion can arise.
All lattice translates of $G_{X,1}$, $G_{X,2}$, $G_{Z,1}$, and $G_{Z,2}$ generate the gauge group.

As a simple example, the subsystem surface code~\cite{bravyi2013subsystemsurfacecodesthreequbit} is represented by
\begin{eqs}
    G_{X,1} =
    \left[\begin{array}{c}
        1 \\
        1 \\
        1 \\
        \hline
        0 \\
        0 \\
        0
    \end{array}\right],
    \qquad
    G_{X,2} =
    \left[\begin{array}{c}
        xy \\
        y \\
        x \\
        \hline
        0 \\
        0 \\
        0
    \end{array}\right], \qquad
    G_{Z,1} =
    \left[\begin{array}{c}
        0 \\
        0 \\
        0 \\
        \hline
        y \\
        y \\
        1
    \end{array}\right],
    \qquad
    G_{Z,2} =
    \left[\begin{array}{c}
        0 \\
        0 \\
        0 \\
        \hline
        x \\
        1 \\
        x
    \end{array}\right],
    \label{eq: subsystem surface code}
\end{eqs}
which is illustrated in Fig.~\ref{fig:example_poly}.
Using the symplectic product above, the corresponding commutation matrix is
\begin{equation}
M_c =
\begin{pmatrix}
1 & 1\\
1 & 1
\end{pmatrix}.
\end{equation}
Thus $\det(M_c)=0$, so local stabilizers exist.
Moreover, the nonzero entries are monomials, hence remain units after imposing periodic boundary conditions.
Therefore, no additional nonlocal stabilizers arise on finite tori.

In summary, the Laurent-polynomial formalism converts the analysis of translation-invariant subsystem codes into algebra over $R$: Pauli operators are vectors over $R$, lattice translations are monomial multiplications, commutation is computed by the symplectic product, and local stabilizers are obtained from kernels of the gauge-generator commutation matrix.

%%%%%%%%%%%%%%%%%%%%%%%%%%%%%
\section{Algebraic criteria and proofs}
\label{app:algebraic_proofs}

This appendix proves the algebraic criteria used in the main text, together with their generalizations to rectangular commutation matrices.
We begin with the $2\times2$ reduction and nonlocal-stabilizer criteria, then present their determinantal-ideal extensions, and finally prove the Clifford correspondence to the associated BB code.

%%%%%%%%%%%%%%%%%%%%%%%%%%%%%
\subsection{Proof of Lemma~\ref{lemma: monomial condition}}\label{app: proof of Lemma monomial condition}
\begin{proof}
First note that the Laurent monomials \(x^ay^b\) are precisely the units of \(R\).

We first prove \((1)\Rightarrow (2)\). Suppose there exist \(P,Q\in GL_2(R)\) such that \(PMQ=\begin{psmallmatrix}1&0\\0&0\end{psmallmatrix}\). Since a monomial is a unit in \(R\), the entries of \(PMQ\) generate the whole ring:
\begin{equation}
I_1(PMQ)=R.
\end{equation}
On the other hand, multiplication by invertible matrices only performs invertible row and column operations, so the ideal generated by the entries is unchanged:
\begin{equation}
I_1(M) = I_1(PMQ)= R.
\end{equation}

Now we prove \((2)\Rightarrow (1)\). Assume that
\(
I_1(M)=R
\)
and
\(
\det(M)=0.
\)
Write
\begin{equation}
M=
\begin{pmatrix}
a&b\\
c&d
\end{pmatrix}.
\end{equation}
Since \(I_1(M)=R\), the matrix \(M\) is not zero. By swapping rows and columns if necessary, we may assume \(a\neq 0\). These swaps are invertible operations, so they can be absorbed into the final choices of \(P\) and \(Q\).

Since $R$ is a unique factorization domain, choose
\begin{equation}
    g=\gcd(a,c),
\end{equation}
where the greatest common divisor of two Laurent polynomials is defined only up to multiplication by a Laurent monomial. Then we may write
\begin{equation}
    a=g\alpha,
    \qquad
    c=g\gamma,
\end{equation}
with $\gcd(\alpha,\gamma)=1$. Since $\det(M)=0$, we have
\begin{equation}
    ad=bc.
\end{equation}
Substituting $a=g\alpha$ and $c=g\gamma$, we obtain
\begin{equation}
    g\alpha d=g\gamma b.
\end{equation}
Since $R$ is an integral domain and $g\neq 0$, we can cancel $g$, obtaining
\begin{equation}
    \alpha d=\gamma b.
\end{equation}
Because $\gcd(\alpha,\gamma)=1$ and $R$ is a UFD, the relation $\alpha d=\gamma b$ implies that $\alpha\mid b$. Thus there exists $h\in R$ such that
\begin{equation}
    b=\alpha h.
\end{equation}
Substituting this back into $\alpha d=\gamma b$, we get
\begin{equation}
    \alpha d=\gamma\alpha h.
\end{equation}
Since $R$ is an integral domain and $\alpha\neq 0$, we can cancel $\alpha$, obtaining
\begin{equation}
    d=\gamma h.
\end{equation}
Therefore
\begin{equation}
    M=
    \begin{pmatrix}
        \alpha g&\alpha h\\
        \gamma g&\gamma h
    \end{pmatrix}
    =
    \begin{pmatrix}
        \alpha\\
        \gamma
    \end{pmatrix}
    \begin{pmatrix}
        g&h
    \end{pmatrix}.
\end{equation}
It follows that
\begin{equation}
    I_1(M)
    =
    \langle \alpha g,\alpha h,\gamma g,\gamma h\rangle
    =
    \langle \alpha,\gamma\rangle \langle g,h\rangle .
\end{equation}
By assumption $I_1(M)=R$, so
\begin{equation}
    \langle \alpha,\gamma\rangle \langle g,h\rangle=R.
\end{equation}
This forces
\begin{equation}
    \langle \alpha,\gamma\rangle=R,
    \qquad
    \langle g,h\rangle=R.
\end{equation}
Indeed, if either ideal were proper, it would be contained in a maximal ideal, and then their product would also be contained in that maximal ideal, contradicting the fact that the product is $R$.

Since $\langle \alpha,\gamma\rangle=R$, there exist $r,s\in R$ such that
\begin{equation}
    r\alpha+s\gamma=1.
\end{equation}
Define
\begin{equation}
    P_0=
    \begin{pmatrix}
        r&s\\
        -\gamma&\alpha
    \end{pmatrix}.
\end{equation}
Then
\begin{equation}
    \det(P_0)=r\alpha+s\gamma=1,
\end{equation}
so $P_0\in GL_2(R)$, and
\begin{equation}
    P_0
    \begin{pmatrix}
        \alpha\\
        \gamma
    \end{pmatrix}
    =
    \begin{pmatrix}
        1\\
        0
    \end{pmatrix}.
\end{equation}
Similarly, since $\langle g,h\rangle=R$, there exist $m,n\in R$ such that $gm+hn=1$.
% \begin{equation}
%     gm+hn=1.
% \end{equation}
Define
\begin{equation}
    Q_0=
    \begin{pmatrix}
        m&-h\\
        n&g
    \end{pmatrix}.
\end{equation}
Then
\begin{equation}
    \det(Q_0)=mg+hn=1,
\end{equation}
so $Q_0\in GL_2(R)$, and
\begin{equation}
    \begin{pmatrix}
        g&h
    \end{pmatrix}
    Q_0
    =
    \begin{pmatrix}
        1&0
    \end{pmatrix}.
\end{equation}
Therefore
\begin{equation}
    P_0MQ_0
    =
    P_0
    \begin{pmatrix}
        \alpha\\
        \gamma
    \end{pmatrix}
    \begin{pmatrix}
        g&h
    \end{pmatrix}
    Q_0
    =
    \begin{pmatrix}
        1\\
        0
    \end{pmatrix}
    \begin{pmatrix}
        1&0
    \end{pmatrix}
    =
    \begin{pmatrix}
        1&0\\
        0&0
    \end{pmatrix}.
\end{equation}
%In particular, $P_0MQ_0$ has an entry equal to $1$, which is a monomial.

Thus, there exist invertible matrices $P,Q\in GL_2(R)$ such that $PMQ$ has a monomial entry. This proves $(2)\Rightarrow (1)$, and hence the two conditions are equivalent.
\end{proof}

%%%%%%%%%%%%%%%%%%%%%%%%%%%%%
\subsection{Generalization of Lemma~\ref{lemma: monomial condition}}\label{app: Generalization of Lemma monomial condition}

The main text uses Lemma~\ref{lemma: monomial condition} for the $2\times2$ commutation matrix that appears when there are two $X$-type and two $Z$-type gauge-generator families.
For a more general subsystem code with $p$ $X$-type and $q$ $Z$-type gauge-generator families, the commutation matrix is rectangular.
The entry-ideal condition is then replaced by a determinantal-ideal condition: the ideal generated by the maximal nonzero minors must be the whole ring.
When this holds, invertible row and column operations bring the commutation matrix to canonical rank form.

\begin{lemma}[Unimodular-rank reduction]
\label{lem:unimodular_rank_reduction}
Let
\begin{equation}
R=\mathbb Z_2[x^{\pm1},y^{\pm1}],
\qquad
M\in \operatorname{Mat}_{p\times q}(R).
\end{equation}
For \(j\geq0\), let \(I_j(M)\subseteq R\) be the ideal generated by all
\(j\times j\) minors of \(M\), with \(I_0(M)=R\) and
\(I_j(M)=0\) for \(j>\min(p,q)\). Let
\(
r=\operatorname{rank}(M)
\)
be the generic rank of \(M\), so that
\begin{equation}
I_r(M)\neq0,
\qquad
I_{r+1}(M)=0.
\end{equation}
Then the following are equivalent:
\begin{enumerate}
    \item There exist invertible matrices
    \begin{equation}
    P\in GL_p(R),
    \qquad
    Q\in GL_q(R),
    \end{equation}
    such that
    \begin{equation}
    PMQ=
    \begin{pmatrix}
        I_r&0\\
        0&0
    \end{pmatrix}.
    \end{equation}

    \item The maximal nonzero determinantal ideal is the unit ideal:
    \begin{equation}
    I_r(M)=R.
    \end{equation}
\end{enumerate}
\end{lemma}

\begin{proof}
We first prove \((1)\Rightarrow(2)\). Determinantal ideals are invariant
under invertible row and column operations. Indeed, by the Cauchy--Binet
formula, multiplying by an invertible matrix replaces each \(j\times j\)
minor by an \(R\)-linear combination of \(j\times j\) minors, and applying
the inverse operation gives the reverse inclusion. Therefore
\begin{equation}
I_j(PMQ)=I_j(M)
\end{equation}
for all \(j\). For the canonical matrix
\begin{equation}
D_r:=
\begin{pmatrix}
I_r&0\\
0&0
\end{pmatrix},
\end{equation}
one \(r\times r\) minor is equal to \(1\). Hence
\begin{equation}
I_r(D_r)=R.
\end{equation}
Thus
\begin{equation}
I_r(M)=I_r(PMQ)=I_r(D_r)=R.
\end{equation}

We now prove \((2)\Rightarrow(1)\). Assume
\begin{equation}
I_r(M)=R.
\end{equation}
Since \(r\) is the generic rank, we also have
\begin{equation}
I_{r+1}(M)=0.
\end{equation}
If \(r=0\), then \(M=0\), and the claim is immediate. Hence assume
\(r>0\).

Let
\begin{equation}
\varphi_M:R^q\longrightarrow R^p
\end{equation}
be the \(R\)-linear map represented by \(M\). We first show that
\(\varphi_M\) has constant rank \(r\) locally. Let \(\mathfrak q\subset R\)
be any prime ideal. Since \(I_r(M)=R\), at least one \(r\times r\) minor of
\(M\) is not contained in \(\mathfrak q\), and hence becomes a unit in the
local ring \(R_{\mathfrak q}\). After permuting rows and columns, we may
write \(M\) over \(R_{\mathfrak q}\) in block form
\begin{equation}
M=
\begin{pmatrix}
A&B\\
C&D
\end{pmatrix},
\end{equation}
where \(A\) is an invertible \(r\times r\) matrix over \(R_{\mathfrak q}\).
Using invertible row and column operations over \(R_{\mathfrak q}\), we get
\begin{equation}
\begin{pmatrix}
A&B\\
C&D
\end{pmatrix}
\sim
\begin{pmatrix}
I_r&0\\
0&D-CA^{-1}B
\end{pmatrix}.
\end{equation}
Each entry of the Schur complement \(D-CA^{-1}B\) is, up to multiplication
by \(\det(A)^{-1}\), an \((r+1)\times(r+1)\) minor of \(M\). Since
\begin{equation}
I_{r+1}(M)=0,
\end{equation}
all these entries vanish. Therefore, locally at every prime \(\mathfrak q\),
we have
\begin{equation}
M_{\mathfrak q}\sim
\begin{pmatrix}
I_r&0\\
0&0
\end{pmatrix}.
\end{equation}

It follows that
\begin{equation}
E:=\operatorname{im}(\varphi_M)
\end{equation}
is locally free of rank \(r\), and locally a direct summand of \(R^p\).
Hence \(E\) is a finitely generated projective \(R\)-module. Similarly,
\begin{equation}
\ker(\varphi_M)
\qquad\text{and}\qquad
\operatorname{coker}(\varphi_M)
\end{equation}
are finitely generated projective \(R\)-modules, since locally they are free
of ranks \(q-r\) and \(p-r\), respectively.

Now use that
\begin{equation}
R=\mathbb Z_2[x^{\pm1},y^{\pm1}]
\end{equation}
is a Laurent polynomial ring over a field. By the Quillen--Suslin--Swan
theorem for Laurent polynomial rings, every finitely generated projective
\(R\)-module is free. Therefore
\begin{equation}
E,\qquad \ker(\varphi_M),\qquad \operatorname{coker}(\varphi_M)
\end{equation}
are free \(R\)-modules.

Since \(E\) is projective, the exact sequence
\begin{equation}
0\longrightarrow \ker(\varphi_M)
\longrightarrow R^q
\longrightarrow E
\longrightarrow 0
\end{equation}
splits. Thus
\begin{equation}
R^q\cong E'\oplus \ker(\varphi_M),
\end{equation}
where \(E'\cong E\), and \(\varphi_M\) restricts to an isomorphism
\begin{equation}
E'\xrightarrow{\sim}E.
\end{equation}
Similarly, since \(\operatorname{coker}(\varphi_M)\) is projective, the exact
sequence
\begin{equation}
0\longrightarrow E
\longrightarrow R^p
\longrightarrow \operatorname{coker}(\varphi_M)
\longrightarrow 0
\end{equation}
splits, so
\begin{equation}
R^p\cong E\oplus \operatorname{coker}(\varphi_M).
\end{equation}

Choose bases for the free modules \(E'\), \(\ker(\varphi_M)\), \(E\), and
\(\operatorname{coker}(\varphi_M)\). With respect to the domain decomposition
\begin{equation}
R^q=E'\oplus \ker(\varphi_M)
\end{equation}
and the codomain decomposition
\begin{equation}
R^p=E\oplus \operatorname{coker}(\varphi_M),
\end{equation}
the map \(\varphi_M\) is the identity from \(E'\) to \(E\), and is zero on
\(\ker(\varphi_M)\). Hence its matrix is
\begin{equation}
\begin{pmatrix}
I_r&0\\
0&0
\end{pmatrix}.
\end{equation}
Equivalently, there exist
\begin{equation}
P\in GL_p(R),
\qquad
Q\in GL_q(R),
\end{equation}
such that
\begin{equation}
PMQ=
\begin{pmatrix}
I_r&0\\
0&0
\end{pmatrix}.
\end{equation}
This proves the claim.
\end{proof}

%%%%%%%%%%%%%%%%%%%%%%%%%%%%%%%%%%%%%%%%%%%%%%%%%%%%
\subsection{Proof of Theorem~\ref{thm: nonlocal stabilizer}}\label{app: proof of Theorem nonlocal stabilizer}

\begin{proof}
Write
\begin{equation}
M=
\begin{pmatrix}
a&b\\
c&d
\end{pmatrix}.
\end{equation}
We demonstrate the argument for right kernels. The corresponding statement for left kernels follows by applying the same argument to $M^{\mathsf T}$.

For any torus quotient
\begin{equation}
A=R_{n,m}=R/(x^n-1,y^m-1),
\end{equation}
write
\begin{equation}
\mathcal K_A
:=
\ker_A(M^{(n,m)}:A^2\to A^2)
\end{equation}
for the finite-torus kernel. Also write
\begin{equation}
\mathcal K_R
:=
\ker_R(M:R^2\to R^2)
\end{equation}
for the infinite-plane local kernel.
Let
\begin{equation}
\rho_A:R\to A
\end{equation}
be the quotient map. We extend $\rho_A$ coefficientwise to vectors in $R^2$:
\begin{equation}
u=
\begin{pmatrix}
f\\ g
\end{pmatrix}
\in R^2 ~\Rightarrow~
\rho_A(u)=
\begin{pmatrix}
\rho_A(f)\\
\rho_A(g)
\end{pmatrix}
\in A^2.
\end{equation}
Define the reduction of the infinite-plane local kernel on the torus to be the $A$-submodule
\begin{equation}
\rho_A({\mathcal K}_R)
:=
\left\{
\sum_i t_i\,\rho_A(u_i)
:
t_i\in A,\ u_i\in \mathcal K_R
\right\}
\subseteq A^2.
\end{equation}
In words, $\rho_A({\mathcal K}_R)$ is generated by reducing all infinite-plane local kernel vectors modulo
\begin{equation}
x^n-1,\qquad y^m-1,
\end{equation}
and then allowing arbitrary torus-periodic $A$-linear combinations.

Since every $u_i\in\mathcal K_R$ satisfies
\begin{equation}
Mu_i=0
\end{equation}
over $R$, its reduction satisfies
\begin{equation}
M^{(n,m)}\rho_A(u_i)=0
\end{equation}
over $A$. Hence
\begin{equation}
\rho_A({\mathcal K}_R)\subseteq \mathcal K_A.
\end{equation}
The theorem says that this inclusion is equality for every torus when $I_1(M)=R$, and is strict for some torus when $I_1(M)$ is proper.

First assume
\(
I_1(M)=R.
\)
By Lemma~\ref{lemma: monomial condition}, we have
\begin{equation}
PMQ
=
\begin{pmatrix}
1&0\\
0&0
\end{pmatrix}.
\end{equation}
Now let
\begin{equation}
A:=R_{n,m}
\end{equation}
be any finite-torus quotient. Reducing $P$ and $Q$ modulo $x^n-1$ and $y^m-1$,
% \begin{equation}
% x^n-1,\qquad y^m-1,
% \end{equation}
we get invertible matrices
\begin{equation}
P_A,Q_A\in GL_2(A).
\end{equation}
Thus
\begin{equation}
P_A M^{(n,m)} Q_A
=
\begin{pmatrix}
1&0\\
0&0
\end{pmatrix}.
\end{equation}
It follows that
\begin{equation}
\mathcal K_A
=
Q_A
\begin{pmatrix}
0\\w
\end{pmatrix}, \quad
\forall~w \in A.
\end{equation}
On the infinite plane, the same calculation gives
\begin{equation}
\mathcal K_R
=
Q \begin{pmatrix}
0\\v
\end{pmatrix}, \quad
\forall~v \in R.
\end{equation}
Therefore, after reducing to the torus, the $A$-submodule generated by the reductions of all elements of $\mathcal K_R$ is exactly $\mathcal K_A$.
Hence
\begin{equation}
\rho_A({\mathcal K}_R)
=
\mathcal K_A.
\end{equation}
Since this holds for every $n,m$, imposing periodic boundary conditions creates no right-kernel vectors beyond those coming from infinite-plane local kernels. Applying the same argument to $M^{\mathsf T}$ gives the corresponding statement for left kernels. Thus no additional nonlocal stabilizers appear when
\(
I_1(M)=R.
\)

Now assume that
\(
I_1(M)
\)
is proper. Choose a maximal ideal
\begin{equation}
\mathfrak p\supset I_1(M).
\end{equation}
Set
\begin{equation}
K:=R/\mathfrak p.
\end{equation}
Since $R$ is a finitely generated algebra over $\mathbb Z_2$, the field $K$ is a finite extension of $\mathbb Z_2$. Thus
\begin{equation}
K\cong \mathbb F_{2^e}
\end{equation}
for some $e\geq 1$.
Let
\begin{equation}
\alpha=x\bmod \mathfrak p,
\qquad
\beta=y\bmod \mathfrak p.
\end{equation}
Because $x$ and $y$ are units in $R$, their images $\alpha,\beta$ are nonzero elements of $K$. Hence
\(
\alpha,\beta\in K^\times.
\)
The multiplicative group $K^\times$ has order
\begin{equation}
|K^\times|=2^e-1,
\end{equation}
which is odd. Therefore the multiplicative orders of $\alpha$ and $\beta$ are odd. Choose
\begin{equation}
n=\operatorname{ord}(\alpha),
\qquad
m=\operatorname{ord}(\beta).
\end{equation}
Then
\(
\alpha^n=1,~
\beta^m=1,
\)
and $n,m$ are odd.
Thus,
% \begin{equation}
% x^n-1,\ y^m-1\in \mathfrak p.
% \end{equation}
% Let
\begin{equation}
J:=\langle x^n-1,y^m-1 \rangle \subseteq \mathfrak p.
\end{equation}
Therefore the quotient map
\begin{equation}
R\to K=R/\mathfrak p
\end{equation}
annihilates $J$, so it factors through
\begin{equation}
A=R_{n,m}=R/J.
\end{equation}
Equivalently, we get an induced map
\begin{equation}
A\to K:
\quad
r+J\longmapsto r+\mathfrak p.
\end{equation}
Let
\begin{equation}
\mathfrak n:=\mathfrak p/J
=
\{r+J\in A:r\in\mathfrak p\}.
\end{equation}
Then $\mathfrak n$ is the kernel of the induced map $A\to K$. Hence
\begin{equation}
A/\mathfrak n\cong R/\mathfrak p=K,
\end{equation}
so $\mathfrak n$ is a maximal ideal of $A$.

Recall that, over a field, a polynomial is called square-free if no irreducible factor occurs with multiplicity greater than one. Equivalently, for a polynomial $f$, square-freeness is detected by the derivative criterion
\begin{equation}
f\text{ is square-free}
\quad\Longleftrightarrow\quad
\gcd(f,f')=1.
\end{equation}
Since the characteristic is two and $n$ is odd, we have
\begin{equation}
\frac{d}{dx}(x^n-1)
=
n x^{n-1}
=
x^{n-1}.
\end{equation}
The polynomial $x^{n-1}$ is a power of $x$, while $x$ does not divide $x^n-1$, because $x^n-1$ has nonzero constant term. Hence
\begin{equation}
\gcd(x^n-1,x^{n-1})=1.
\end{equation}
Therefore $x^n-1$ is square-free over $\mathbb Z_2$. Similarly, $y^m-1$ is square-free over $\mathbb Z_2$.
% \begin{equation}
% \frac{d}{dy}(y^m-1)
% =
% m y^{m-1}
% =
% y^{m-1},
% \end{equation}
% and
% \begin{equation}
% \gcd(y^m-1,y^{m-1})=1,
% \end{equation}
% so $y^m-1$ is square-free over $\mathbb Z_2$.

We next explain why this implies that $A$ is a finite product of fields. Since $x^n=1$ and $y^m=1$
% \begin{equation}
% x^n=1,
% \qquad
% y^m=1
% \end{equation}
inside $A$, $A$ is finite-dimensional over $\mathbb Z_2$.
Because $x^n-1$ is square-free over $\ZZ_2$, the quotient
\begin{equation}
B:=\ZZ_2[x^{\pm1}]/(x^n-1)
\end{equation}
has no nonzero nilpotent elements, so $B$ is reduced. Since $x^n=1$ in $B$, we have $x^{-1}=x^{n-1}$, and hence
\begin{equation}
B\cong \ZZ_2[x]/(x^n-1).
\end{equation}
Thus $B$ is finite-dimensional over $\ZZ_2$, hence Artinian. Therefore $B$ is an Artinian reduced ring. By the structure theorem for Artinian reduced rings, $B$ is a finite product of fields. Since $B$ is finite over $\ZZ_2$, each factor is a finite field.
Indeed, if
\begin{equation}
x^n-1=q_1(x)\cdots q_r(x)
\end{equation}
is its factorization into distinct irreducible polynomials, then the Chinese remainder theorem gives
\begin{equation}
B
\cong
\prod_{i=1}^r \mathbb Z_2[x]/(q_i(x)),
\end{equation}
and each factor is a finite field.

Now
\begin{equation}
A
\cong
B[y^{\pm1}]/(y^m-1).
\end{equation}
Writing
\(
B\cong F_1\times\cdots\times F_r
\)
as a product of finite fields, we get
\begin{equation}
A
\cong
\prod_{i=1}^r F_i[y^{\pm1}]/(y^m-1).
\end{equation}
For each finite field $F_i$, the same derivative argument shows that $y^m-1$ is square-free over $F_i$, since the characteristic is still two and $m$ is odd. Therefore each quotient
\begin{equation}
F_i[y^{\pm1}]/(y^m-1)
\end{equation}
is again a finite product of finite fields. Consequently, $A$
is itself a finite product of fields.

Remember $\mathfrak n=\mathfrak p/(x^n-1,y^m-1)$ is a maximal ideal of $A$, and
\begin{equation}
A/\mathfrak n\cong R/\mathfrak p=K.
\end{equation}
Since $A$ is a finite product of fields, localizing at a maximal ideal simply selects the corresponding field factor:
\begin{equation}
A_{\mathfrak n}
\cong
A/\mathfrak n
\cong
K.
\end{equation}
Now consider the finite-torus kernel
\begin{equation}
\mathcal K_A=\ker_A(M^{(n,m)}:A^2\to A^2).
\end{equation}
Localizing at $\mathfrak n$ gives
\begin{equation}
(\mathcal K_A)_{\mathfrak n}
=
\ker_{A_{\mathfrak n}}
\left(
M^{(n,m)}_{\mathfrak n}:A_{\mathfrak n}^2\to A_{\mathfrak n}^2
\right).
\end{equation}
Under the identification
\(
A_{\mathfrak n}\cong K,
\)
all entries of $M$ vanish, because
\begin{equation}
a,b,c,d\in I_1(M)\subseteq \mathfrak p.
\end{equation}
Hence, the localized matrix is the zero matrix over $K$:
\(
M^{(n,m)}_{\mathfrak n}=0.
\)
Therefore
\begin{equation}
(\mathcal K_A)_{\mathfrak n}
=
\ker_K(0:K^2\to K^2)
=
K^2.
\end{equation}
So, in the Fourier sector corresponding to $\mathfrak n$, the full torus kernel is two-dimensional.

We now compare this with the part coming from infinite-plane local kernels. Let
\begin{equation}
\rho_A({\mathcal K}_R)
\subseteq
\mathcal K_A
\end{equation}
be the $A$-submodule generated by reductions of elements of $\mathcal K_R$, as defined at the beginning of the proof.
We claim that
\begin{equation}
(\rho_A({\mathcal K}_R))_{\mathfrak n}
\subsetneq
K^2.
\end{equation}
Suppose, for contradiction, that
\begin{equation}
(\rho_A({\mathcal K}_R))_{\mathfrak n}
=
K^2.
\end{equation}
Let
\begin{equation}
L:=(\mathcal K_R)_{\mathfrak p}
\subseteq R_{\mathfrak p}^2
\end{equation}
be the localization of the infinite-plane kernel at $\mathfrak p$. Explicitly,
\begin{equation}
L=
\left\{
\frac{u}{s}
:
u\in \mathcal K_R,\ s\in R\setminus \mathfrak p
\right\},
\end{equation}
where $R\setminus \mathfrak p$ is multiplicatively closed since $\mathfrak p$ is prime.

Reducing coefficients modulo $\mathfrak p$ gives a map
\begin{equation}
R_{\mathfrak p}^2\to
R_{\mathfrak p}^2/\mathfrak p R_{\mathfrak p}^2
\cong
K^2.
\end{equation}
We now describe the localization of $\rho_A({\mathcal K}_R)$  at $\mathfrak n$. Since
\(
A_{\mathfrak n}\cong K,
\)
we have
\begin{equation}
(A^2)_{\mathfrak n}\cong K^2.
\end{equation}
An element of $(\rho_A({\mathcal K}_R))_{\mathfrak n}$ is a finite sum of terms of the form
\begin{equation}
\frac{t}{s}\,\rho_A(u),
\qquad
t\in A,\quad s\in A\setminus\mathfrak n,\quad u\in\mathcal K_R.
\end{equation}
Under the identification $A_{\mathfrak n}\cong K$, the coefficient
\(
\frac{t}{s}
\)
maps to the scalar
\begin{equation}
\frac{t\bmod\mathfrak n}{s\bmod\mathfrak n}\in K,
\end{equation}
which is well-defined because $s\notin\mathfrak n$. Also, the vector
\(
\rho_A(u)
\)
maps to
\begin{equation}
u\bmod\mathfrak p\in K^2.
\end{equation}
Therefore
\begin{equation}
(\rho_A({\mathcal K}_R))_{\mathfrak n}
=
\operatorname{span}_K
\{u\bmod\mathfrak p:u\in\mathcal K_R\}
\subseteq K^2.
\end{equation}

On the other hand, let
\begin{equation}
L=(\mathcal K_R)_{\mathfrak p}\subseteq R_{\mathfrak p}^2.
\end{equation}
Reducing modulo the maximal ideal $\mathfrak pR_{\mathfrak p}$ gives
\begin{equation}
R_{\mathfrak p}^2
\longrightarrow
R_{\mathfrak p}^2/\mathfrak pR_{\mathfrak p}^2
\cong K^2.
\end{equation}
An element
\begin{equation}
\frac{u}{s}\in L,
\qquad
u\in\mathcal K_R,\quad s\in R\setminus\mathfrak p,
\end{equation}
maps to
\begin{equation}
\frac{u\bmod\mathfrak p}{s\bmod\mathfrak p}\in K^2.
\end{equation}
Since $s\bmod\mathfrak p$ is a nonzero element of $K$, this is a $K$-scalar multiple of $u\bmod\mathfrak p$. Hence the image of $L$ in $K^2$ is precisely
\begin{equation}
\operatorname{span}_K
\{u\bmod\mathfrak p:u\in\mathcal K_R\}.
\end{equation}
Consequently,
\(
(\rho_A({\mathcal K}_R))_{\mathfrak n}
\)
is equal to the image of $L$ modulo $\mathfrak p$ in
\begin{equation}
R_{\mathfrak p}^2/\mathfrak pR_{\mathfrak p}^2
\cong K^2.
\end{equation}
Thus, the assumption
\begin{equation}
(\rho_A({\mathcal K}_R))_{\mathfrak n}=K^2
\end{equation}
means that the image of $L \subseteq R_{\mathfrak p}^2 $ modulo $\mathfrak p$ is all of $K^2 \cong R_{\mathfrak p}^2/\mathfrak p R_{\mathfrak p}^2
$. Equivalently,
\begin{equation}
L+\mathfrak p R_{\mathfrak p}^2
=
R_{\mathfrak p}^2.
\end{equation}

Now define
\begin{equation}
C:=R_{\mathfrak p}^2/L.
\end{equation}
The previous equality says
\begin{equation}
C=\mathfrak p R_{\mathfrak p} C.
\end{equation}
Since $C$ is a finitely generated module over the local ring $R_{\mathfrak p}$, Nakayama's lemma implies
\(
C=0.
\)
Therefore
\begin{equation}
L=R_{\mathfrak p}^2.
\end{equation}

But localization is exact, so $L$ is the kernel of the localized map
\begin{equation}
M_{\mathfrak p}:R_{\mathfrak p}^2\to R_{\mathfrak p}^2.
\end{equation}
Thus
\(
L=R_{\mathfrak p}^2
\)
means that $M_{\mathfrak p}$ kills every vector in $R_{\mathfrak p}^2$. In particular, it kills the two standard basis vectors, so both columns of $M_{\mathfrak p}$ are zero. Hence
\begin{equation}
M_{\mathfrak p}=0.
\end{equation}
Since $R$ is a domain, the localization map
\(
R\to R_{\mathfrak p}
\)
is injective. Therefore
\(
M_{\mathfrak p}=0
\)
would imply
\(
M=0
\)
over $R$, contradicting the hypothesis $M\neq 0$.

This contradiction proves that
\(
(\rho_A({\mathcal K}_R))_{\mathfrak n}
\subsetneq
K^2.
\)
But we already showed that
\(
(\mathcal K_A)_{\mathfrak n}=K^2.
\)
Hence,
\begin{equation}
(\rho_A({\mathcal K}_R))_{\mathfrak n}
\subsetneq
(\mathcal K_A)_{\mathfrak n}.
\end{equation}
If we had
\(
\rho_A({\mathcal K}_R)=\mathcal K_A
\)
before localization, then their localizations at $\mathfrak n$ would also be equal. Since their localizations are not equal, we must have
\begin{equation}
\rho_A({\mathcal K}_R)
\subsetneq
\mathcal K_A.
\end{equation}

Therefore, for the torus sizes $n,m$ chosen above, the finite-torus kernel is strictly larger than the reduction of the infinite-plane local kernel. Choose
\begin{equation}
w\in \mathcal K_A\setminus \rho_A({\mathcal K}_R).
\end{equation}
Then $w$ is a torus kernel vector that does not come from any torus-periodic combination of reduced infinite-plane local kernel vectors.

The corresponding product of gauge generators commutes with all gauge generators on the finite torus because
\(
w\in \mathcal K_A.
\)
Hence it gives a stabilizer. Since
\(
w\notin \rho_A({\mathcal K}_R),
\)
this stabilizer is not generated by the reductions of infinite-plane local stabilizers. Therefore, if the corresponding Pauli operator is nonzero, it is an additional nonlocal stabilizer.

Finally, applying the same argument to $M^{\mathsf T}$ gives the corresponding statement for left kernels.
\end{proof}

%%%%%%%%%%%%%%%%%%%%%%%%%%%%%%%%%%%%%%%%%%%%%%%%%%%%
\subsection{Generalization of Theorem~\ref{thm: nonlocal stabilizer}}\label{app: Generalization of Theorem nonlocal stabilizer}

Theorem~\ref{thm: nonlocal stabilizer} treats the $2\times2$ commutation matrix arising from two $X$-type and two $Z$-type gauge-generator families.
For a general translation-invariant CSS subsystem code with $p$ $X$-type and $q$ $Z$-type gauge-generator families, the commutation matrix is instead a $p\times q$ matrix over $R$.
In this setting, the entry ideal $I_1(M)$ is replaced by the determinantal ideal $I_r(M)$ generated by the maximal nonzero minors, where $r$ is the generic rank of $M$.
The following theorem shows that the same dichotomy holds: if $I_r(M)=R$, periodic boundary conditions introduce no additional kernel vectors beyond the infinite-plane local kernels; if $I_r(M)$ is proper, then some finite torus has extra kernel vectors, which give nonlocal stabilizers whenever the corresponding Pauli operators are nonzero.
We state the result for right kernels, corresponding to $Z$-type stabilizers; the statement for left kernels, corresponding to $X$-type stabilizers, follows by applying the same argument to $M^{\mathsf T}$.

\begin{theorem}[Determinantal-ideal criterion for nonlocal stabilizers]
\label{thm:determinantal_nonlocal_stabilizers}
Let
\begin{equation}
R=\mathbb Z_2[x^{\pm1},y^{\pm1}],
\qquad
M\in \operatorname{Mat}_{p\times q}(R),
\end{equation}
and let
\(
r=\operatorname{rank}(M)
\)
be the generic rank of \(M\). Equivalently,
\begin{equation}
I_r(M)\neq0,
\qquad
I_{r+1}(M)=0,
\end{equation}
where \(I_j(M)\) denotes the ideal generated by all \(j\times j\) minors of
\(M\).

For an \(n\times m\) torus, define
\begin{equation}
R_{n,m}:=R/(x^n-1,y^m-1),
\end{equation}
and let \(M^{(n,m)}\) be the image of \(M\) in
\(
\operatorname{Mat}_{p\times q}(R_{n,m}).
\)
Let
\begin{equation}
\mathcal K_R
:=
\ker_R(M:R^q\to R^p)
\end{equation}
be the infinite-plane right kernel, and let
\begin{equation}
\mathcal K_{n,m}
:=
\ker_{R_{n,m}}
\left(
M^{(n,m)}:R_{n,m}^q\to R_{n,m}^p
\right)
\end{equation}
be the finite-torus right kernel.
Let
\begin{equation}
\rho_{n,m}:R\to R_{n,m}
\end{equation}
be the quotient map. We define
\begin{equation}
\rho_{n,m}(\mathcal K_R)
\subseteq
R_{n,m}^q
\end{equation}
to be the \(R_{n,m}\)-submodule generated by the coefficientwise reductions of
all elements of \(\mathcal K_R\). Explicitly,
\begin{equation}
\rho_{n,m}(\mathcal K_R)
=
\left\{
\sum_i t_i\,\rho_{n,m}(u_i)
:
t_i\in R_{n,m},\ u_i\in\mathcal K_R
\right\}.
\end{equation}
Then:
\begin{enumerate}
    \item If
    \(
    I_r(M)=R,
    \)
    then for every \(n,m\),
    \begin{equation}
    \mathcal K_{n,m}
    =
    \rho_{n,m}(\mathcal K_R).
    \end{equation}
    Thus imposing periodic boundary conditions creates no additional
    right-kernel vectors, and hence no additional nonlocal \(Z\)-type
    stabilizers from this sector.

    \item If
    \(
    I_r(M)\subsetneq R,
    \)
    then there exist positive integers \(n,m\) such that
    \begin{equation}
    \rho_{n,m}(\mathcal K_R)
    \subsetneq
    \mathcal K_{n,m}.
    \end{equation}
    Any nonzero Pauli operator obtained from a vector in
    \begin{equation}
    \mathcal K_{n,m}\setminus \rho_{n,m}(\mathcal K_R)
    \end{equation}
    is an additional nonlocal \(Z\)-type stabilizer.
\end{enumerate}
The analogous statements for \(X\)-type stabilizers hold for left kernels,
or equivalently by applying the theorem to \(M^{\mathsf T}\).
\end{theorem}

\begin{proof}
We first prove the case
\(
I_r(M)=R.
\)
By Lemma~\ref{lem:unimodular_rank_reduction}, there exist invertible matrices
\begin{equation}
P\in GL_p(R),
\qquad
Q\in GL_q(R),
\end{equation}
such that
\begin{equation}
PMQ=
D_r:=
\begin{pmatrix}
I_r&0\\
0&0
\end{pmatrix}.
\end{equation}
Now fix any torus quotient
\begin{equation}
A:=R_{n,m}.
\end{equation}
Reducing \(P\) and \(Q\) modulo \(x^n-1\) and \(y^m-1\), we obtain
invertible matrices
\begin{equation}
P_A\in GL_p(A),~
Q_A\in GL_q(A)
~~\Rightarrow~~
P_A M^{(n,m)} Q_A
=
D_r.
\end{equation}
% Therefore
% \begin{equation}
% P_A M^{(n,m)} Q_A
% =
% D_r.
% \end{equation}
The right kernel of \(D_r:A^q\to A^p\) is
\(
0\oplus A^{q-r}.
\)
Hence,
\begin{equation}
\mathcal K_{n,m}
=
Q_A(0\oplus A^{q-r}).
\end{equation}
On the infinite plane, the same calculation gives
\begin{equation}
\mathcal K_R
=
Q(0\oplus R^{q-r}).
\end{equation}
After reducing to \(A\) and allowing arbitrary \(A\)-linear combinations, we
get exactly
\begin{equation}
Q_A(0\oplus A^{q-r}).
\end{equation}
Thus
\begin{equation}
\rho_{n,m}(\mathcal K_R)
=
\mathcal K_{n,m}.
\end{equation}
Since \(n,m\) were arbitrary, imposing periodic boundary conditions creates no
additional right-kernel vectors. Applying the same argument to
\(M^{\mathsf T}\) gives the corresponding statement for left kernels.

We now prove the converse. Assume
\(
I_r(M)\subsetneq R.
\)
Since \(I_0(M)=R\), this case cannot occur when \(r=0\). Hence \(r>0\).
Choose a maximal ideal
\begin{equation}
\mathfrak p\supset I_r(M).
\end{equation}
Set
\begin{equation}
K:=R/\mathfrak p.
\end{equation}
Since \(R\) is a finitely generated algebra over \(\mathbb Z_2\), the field
\(K\) is a finite extension of \(\mathbb Z_2\). Thus, for some \(e\geq1\),
\begin{equation}
K\cong \mathbb F_{2^e}.
\end{equation}

Let $\alpha=x\bmod\mathfrak p$ and $\beta=y\bmod\mathfrak p$.
Because \(x\) and \(y\) are units in \(R\), their images \(\alpha,\beta\) are
nonzero elements of \(K\). Hence,
\(
\alpha,\beta\in K^\times.
\)
The multiplicative group \(K^\times\) has order
\begin{equation}
|K^\times|=2^e-1,
\end{equation}
which is odd. Therefore the multiplicative orders of \(\alpha\) and \(\beta\)
are odd. Choose
\begin{equation}
n=\operatorname{ord}(\alpha),
\qquad
m=\operatorname{ord}(\beta).
\end{equation}
Then
\begin{equation}
\alpha^n=1,
\qquad
\beta^m=1,
\end{equation}
and \(n,m\) are odd. Equivalently,
\(
x^n-1,\ y^m-1\in\mathfrak p,
\)
so
\begin{equation}
J:=\langle x^n-1,y^m-1\rangle  \subseteq\mathfrak p.
\end{equation}
% Then
% \(
% J\subseteq\mathfrak p.
% \)
Therefore, the quotient map
\begin{equation}
R\to K=R/\mathfrak p
\end{equation}
factors through
\(
A=R_{n,m}=R/J.
\)
Let
\begin{equation}
\mathfrak n:=\mathfrak p/J
=
\{r'+J\in A:r'\in\mathfrak p\}.
\end{equation}
Then \(\mathfrak n\) is the kernel of the induced map \(A\to K\), and hence
\begin{equation}
A/\mathfrak n\cong K.
\end{equation}
In particular, \(\mathfrak n\) is a maximal ideal of \(A\).
Following the same argument in the previous section, we have 
\begin{equation}
A_{\mathfrak n}\cong A/\mathfrak n\cong K.
\end{equation}
Now consider the finite-torus kernel
\begin{equation}
\mathcal K_A
:=
\ker_A(M^{(n,m)}:A^q\to A^p).
\end{equation}
Localizing at \(\mathfrak n\), and using the fact that localization preserves
kernels, gives
\begin{equation}
(\mathcal K_A)_{\mathfrak n}
=
\ker_{A_{\mathfrak n}}
\left(
M^{(n,m)}_{\mathfrak n}:A_{\mathfrak n}^q\to A_{\mathfrak n}^p
\right).
\end{equation}
Under the identification
\(
A_{\mathfrak n}\cong K,
\)
this becomes
\begin{equation}
(\mathcal K_A)_{\mathfrak n}
=
\ker_K
\left(
M(\alpha,\beta):K^q\to K^p
\right).
\end{equation}

Since every \(r\times r\) minor of \(M\) lies in
\(
I_r(M)\subseteq\mathfrak p,
\)
all \(r\times r\) minors vanish after specialization to \(K\). Therefore
\begin{equation}
s:=\operatorname{rank}_K M(\alpha,\beta) < r.
\end{equation}
% satisfies
% \begin{equation}
% s<r.
% \end{equation}
Hence,
\begin{equation}
\dim_K(\mathcal K_A)_{\mathfrak n}
=
q-s.
\end{equation}

We now compare this finite-torus kernel with the part coming from
infinite-plane local kernels. For readability, write
\begin{equation}
\overline{\mathcal K}_A
:=
\rho_{n,m}(\mathcal K_R)
\subseteq
\mathcal K_A.
\end{equation}
Let
\begin{equation}
L:=(\mathcal K_R)_{\mathfrak p}
\subseteq R_{\mathfrak p}^q
\end{equation}
be the localization of the infinite-plane kernel at \(\mathfrak p\). Explicitly,
\begin{equation}
L=
\left\{
\frac{u}{s}:
u\in\mathcal K_R,\ s\in R\setminus\mathfrak p
\right\}.
\end{equation}
Reducing modulo the maximal ideal \(\mathfrak pR_{\mathfrak p}\) gives
\begin{equation}
R_{\mathfrak p}^q
\longrightarrow
R_{\mathfrak p}^q/\mathfrak pR_{\mathfrak p}^q
\cong
K^q.
\end{equation}
Under the identification \(A_{\mathfrak n}\cong K\), the localized module
\(
(\overline{\mathcal K}_A)_{\mathfrak n}
\)
is exactly the \(K\)-linear span of the reductions modulo \(\mathfrak p\) of
elements of \(\mathcal K_R\). Indeed, an element of
\((\overline{\mathcal K}_A)_{\mathfrak n}\) is a finite sum of terms
\begin{equation}
\frac{t}{s}\,\rho_{n,m}(u),
\qquad
t\in A,\quad s\in A\setminus\mathfrak n,\quad u\in\mathcal K_R.
\end{equation}
Under \(A_{\mathfrak n}\cong K\), the coefficient \(t/s\) becomes an element
of \(K\), and the vector \(\rho_{n,m}(u)\) becomes \(u\bmod\mathfrak p\).
Thus
\begin{equation}
(\overline{\mathcal K}_A)_{\mathfrak n}
=
\operatorname{span}_K
\{u\bmod\mathfrak p:u\in\mathcal K_R\}.
\end{equation}
On the other hand, an element
\begin{equation}
\frac{u}{s}\in L,
\qquad
u\in\mathcal K_R,\quad s\in R\setminus\mathfrak p,
\end{equation}
maps to
\begin{equation}
\frac{u\bmod\mathfrak p}{s\bmod\mathfrak p}\in K^q.
\end{equation}
Since \(s\bmod\mathfrak p\neq0\) in \(K\), this is a \(K\)-scalar multiple of
\(u\bmod\mathfrak p\). Therefore
\(
(\overline{\mathcal K}_A)_{\mathfrak n}
\)
is precisely the image of \(L\) in
\begin{equation}
R_{\mathfrak p}^q/\mathfrak pR_{\mathfrak p}^q
\cong
K^q.
\end{equation}

We claim that
\begin{equation}
(\overline{\mathcal K}_A)_{\mathfrak n}
\subsetneq
(\mathcal K_A)_{\mathfrak n}.
\end{equation}
Suppose, for contradiction, that equality holds.
We now work over the local ring \(R_{\mathfrak p}\). Since
\begin{equation}
s=\operatorname{rank}_K M(\alpha,\beta),
\end{equation}
there is an \(s\times s\) minor whose image in \(K\) is nonzero, unless
\(s=0\). If \(s>0\), after permuting rows and columns and applying invertible
row and column operations over \(R_{\mathfrak p}\), we may write
\begin{equation}
M_{\mathfrak p}
\sim
\begin{pmatrix}
I_s&0\\
0&N
\end{pmatrix}.
\end{equation}
If \(s=0\), we simply take
\begin{equation}
N=M_{\mathfrak p}.
\end{equation}
In both cases, every entry of \(N\) lies in the maximal ideal
\(
\mathfrak pR_{\mathfrak p},
\)
because the specialization of \(N\) modulo \(\mathfrak p\) is zero.

Moreover, \(N\) has generic rank
\(
r-s>0.
\)
Indeed, invertible row and column operations do not change generic rank, and
the displayed block form has generic rank
\(
s+\operatorname{rank}(N).
\)
Since \(M\) has generic rank \(r\), we get
\begin{equation}
\operatorname{rank}(N)=r-s.
\end{equation}

The above row and column operations identify kernels and their reductions.
After these operations, the localized kernel of \(M\) is
\begin{equation}
0\oplus \ker_{R_{\mathfrak p}}(N),
\end{equation}
while the specialized \(K\)-kernel is
\begin{equation}
0\oplus K^{q-s},
\end{equation}
because \(N\) vanishes modulo \(\mathfrak p\).

Our assumption
\begin{equation}
(\overline{\mathcal K}_A)_{\mathfrak n}
=
(\mathcal K_A)_{\mathfrak n}
\end{equation}
therefore implies that the image of
\(
\ker_{R_{\mathfrak p}}(N)
\)
modulo \(\mathfrak p\) is all of
\(
K^{q-s}.
\)
Equivalently,
\begin{equation}
\ker_{R_{\mathfrak p}}(N)
+
\mathfrak pR_{\mathfrak p}^{q-s}
=
R_{\mathfrak p}^{q-s}.
\end{equation}

Define
\begin{equation}
C:=
R_{\mathfrak p}^{q-s}/\ker_{R_{\mathfrak p}}(N).
\end{equation}
The previous equality says
\(
C=\mathfrak pR_{\mathfrak p}C.
\)
Since \(C\) is a finitely generated module over the local ring
\(R_{\mathfrak p}\), Nakayama's lemma implies
\(
C=0.
\)
Hence
\begin{equation}
\ker_{R_{\mathfrak p}}(N)
=
R_{\mathfrak p}^{q-s}.
\end{equation}
Thus \(N\) annihilates every vector in its domain, and therefore
\(
N=0
\)
over \(R_{\mathfrak p}\). This contradicts the fact that \(N\) has generic
rank
\(
r-s>0.
\)

Therefore
\(
(\overline{\mathcal K}_A)_{\mathfrak n}
\subsetneq
(\mathcal K_A)_{\mathfrak n}.
\)
If we had
\(
\overline{\mathcal K}_A=\mathcal K_A
\)
before localization, then their localizations at \(\mathfrak n\) would also be
equal. Since their localizations are not equal, we must have
\begin{equation}
\overline{\mathcal K}_A
\subsetneq
\mathcal K_A.
\end{equation}
Equivalently,
\begin{equation}
\rho_{n,m}(\mathcal K_R)
\subsetneq
\mathcal K_{n,m}.
\end{equation}
Choose
\begin{equation}
w\in
\mathcal K_{n,m}\setminus \rho_{n,m}(\mathcal K_R).
\end{equation}
Then \(w\) is a finite-torus right-kernel vector that is not obtained from any
torus-periodic combination of reduced infinite-plane local kernel vectors. In
the subsystem-code interpretation, such a right-kernel vector gives a product
of \(Z\)-type gauge generators that commutes with all \(X\)-type gauge
generators on the finite torus. Hence it gives a \(Z\)-type stabilizer. Since
\begin{equation}
w\notin \rho_{n,m}(\mathcal K_R),
\end{equation}
this stabilizer is not generated by the reductions of infinite-plane local
stabilizers. Therefore, if the corresponding Pauli operator is nonzero, it is
an additional nonlocal \(Z\)-type stabilizer.

Finally, applying the same argument to \(M^{\mathsf T}\) proves the
corresponding statement for left kernels and \(X\)-type stabilizers.
\end{proof}

%%%%%%%%%%%%%%%%%%%%%%%%%%%%%%%%%%%%%%%%%%%%%%%%%%%%
\subsection{Proof of Theorem~\ref{thm: exist Clifford gate}}\label{app: proof of Theorem exist Clifford gate}

\begin{proof}
We work in the polynomial Pauli module
\begin{equation}
    R^{6}=R_X^3\oplus R_Z^3 .
\end{equation}
Let $\bar{\cdot}$ denote the involution $x\mapsto x^{-1}$,
$y\mapsto y^{-1}$, extended coefficientwise to $R$, and write
\begin{equation}
    v^\dagger:=\bar v^{\mathsf T}
\end{equation}
for vectors or matrices over $R$. We denote by $\mathcal U$ the symplectic
matrix representing the Clifford circuit $U$; thus conjugation by $U$ acts on
Pauli vectors by left multiplication by $\mathcal U$.

Since $M_c$ has the canonical form in Eq.~\eqref{eq: canonical Mc}, we have
\begin{equation}
    G_{X,1}\cdot G_{Z,1}=1,
\end{equation}
while $G_{X,1}$ and $G_{Z,1}$ commute with all remaining gauge generators.
Write
\begin{equation}
    G_{X,1}=
    \begin{bmatrix}
        x\\
        0
    \end{bmatrix},
    \qquad
    x=
    \begin{bmatrix}
        f\\
        g\\
        h
    \end{bmatrix},
    \qquad
    G_{Z,1}=
    \begin{bmatrix}
        0\\
        z
    \end{bmatrix},
\end{equation}
with $z\in R^3$. The condition $G_{X,1}\cdot G_{Z,1}=1$ gives
\begin{equation}
    x^\dagger z=1 .
\end{equation}
Hence
\begin{equation}
    1\in \langle \bar f,\bar g,\bar h\rangle .
\end{equation}
Applying the antipode map gives
\begin{equation}
    \langle f,g,h\rangle=R .
\end{equation}
Thus $x=(f,g,h)^{\mathsf T}$ is a primitive column vector, i.e., its entries
generate the unit ideal. We use the standard elementary-completion property of
$R=\mathbb Z_2[x^{\pm1},y^{\pm1}]$: any primitive column of $R^3$ can be
mapped to a standard basis vector by an invertible matrix generated by elementary
local row operations, coordinate permutations, and monomial translations.
Therefore we may choose $E\in GL_3(R)$ such that
\begin{equation}
    Ex=e_1,
    \qquad
    e_1=
    \begin{bmatrix}
        1\\0\\0
    \end{bmatrix}.
\end{equation}
Consider the symplectic transformation
\begin{equation}
    \mathcal S_1=
    \begin{bmatrix}
        E & 0\\
        0 & (E^{-1})^\dagger
    \end{bmatrix}.
\end{equation}
By the choice of $E$, this symplectic transformation is implemented by a
finite-depth Clifford circuit: elementary row operations correspond to CNOT
layers, coordinate permutations to onsite Clifford relabelings, and monomial
factors only translate qubits within the lattice labeling.
It preserves the symplectic pairing, and
\begin{equation}
    \mathcal S_1G_{X,1}
    =
    \begin{bmatrix}
        e_1\\
        0
    \end{bmatrix}.
\end{equation}
Since the pairing with $G_{Z,1}$ is preserved, the image of $G_{Z,1}$ has unit
pairing with this vector. Thus
\begin{equation}
    \mathcal S_1G_{Z,1}
    =
    \begin{bmatrix}
        0\\
        w
    \end{bmatrix},
    \qquad
    w=
    \begin{bmatrix}
        1\\
        z_2\\
        z_3
    \end{bmatrix}
\end{equation}
for some $z_2,z_3\in R$.

We next remove the remaining $Z$-support without changing
$\mathcal S_1G_{X,1}$. Define
\begin{equation}
    \alpha:=\bar z_2,
    \quad
    \beta:=\bar z_3,
    \quad
    A=
    \begin{bmatrix}
        1 & \alpha & \beta\\
        0 & 1 & 0\\
        0 & 0 & 1
    \end{bmatrix}.
\end{equation}
Since $R$ has characteristic $2$, $A^{-1}=A$.
Since $Ae_1=e_1$, the symplectic transformation
\begin{equation}
    \mathcal S_2=
    \begin{bmatrix}
        A & 0\\
        0 & (A^{-1})^\dagger
    \end{bmatrix}
\end{equation}
leaves $\mathcal S_1G_{X,1}$ unchanged. This is again implemented by a
finite-depth Clifford circuit, since $A$ is an elementary shear. Moreover,
\begin{equation}
    (A^{-1})^\dagger
    \begin{bmatrix}
        1\\
        z_2\\
        z_3
    \end{bmatrix}
    =
    \begin{bmatrix}
        1\\
        \bar\alpha+z_2\\
        \bar\beta+z_3
    \end{bmatrix}
    =
    \begin{bmatrix}
        1\\
        0\\
        0
    \end{bmatrix},
\end{equation}
where we used $\bar\alpha=z_2$ and $\bar\beta=z_3$. Therefore, with
\begin{equation}
    \mathcal U:=\mathcal S_2\mathcal S_1,
\end{equation}
$\mathcal U$ is the symplectic representation of a finite-depth Clifford circuit, and we obtain
\begin{equation}
    \mathcal U G_{X,1}
    =
    \left[\begin{array}{c}
        1 \\
        0 \\
        0 \\
        0 \\
        0 \\
        0
    \end{array}\right]~,
    \quad
    \mathcal U G_{Z,1}
    =
    \left[\begin{array}{c}
        0 \\
        0 \\
        0 \\
        1 \\
        0 \\
        0
    \end{array}\right].
\label{eq: U GX1 and U GZ1}
\end{equation}
Equivalently, the corresponding Clifford circuit maps, in each unit cell,
\begin{equation}
    G_{X,1}(\mathbf r)\mapsto X_{q(\mathbf r)},
    \qquad
    G_{Z,1}(\mathbf r)\mapsto Z_{q(\mathbf r)}
\end{equation}
on the same qubit $q(\mathbf r)$.

It remains to check the support of the other gauge generators. Let $H$ be any
remaining gauge generator, including any lattice translate. In the canonical
basis, $H$ commutes with both $G_{X,1}$ and $G_{Z,1}$. Since $\mathcal U$ is
symplectic, $\mathcal U H$ commutes with $\mathcal U G_{X,1}$ and $\mathcal U G_{Z,1}$ in Eq.~\eqref{eq: U GX1 and U GZ1}.
Writing
\begin{equation}
    \mathcal U H=
    \left[\begin{array}{c}
        a_1\\a_2\\a_3\\ b_1\\b_2\\b_3
    \end{array}\right]~,
\end{equation}
these two commutation relations imply
\begin{equation}
    a_1=b_1=0.
\end{equation}
Thus $\mathcal U H$ has no support on the disentangled qubit or any of its translates.

Finally, in the canonical basis the remaining gauge generators lie in the local
left and right kernels of $M_c$, and hence are stabilizers. Since
Eq.~\eqref{eq: canonical Mc} has entry ideal $R$, Theorem~\ref{thm: nonlocal stabilizer}
rules out additional nonlocal stabilizers on finite tori. Therefore these local
generators generate the stabilizer group. After deleting the disentangled qubit
coordinates, their images are precisely the stabilizer generators of the
corresponding bivariate bicycle code.
\end{proof}

%%%%%%%%%%%%%%%%%%%%%%%%%%%%%%%%%%%%%%%%%%%%%%%%%%%%%%%%%%%%%%
\section{Reflection-symmetric SBB codes}
\label{app:reflection_symmetric_search}

In this appendix, we collect the algebraic facts used in the search for reflection-symmetric SBB codes. Throughout, we define $R=\mathbb Z_2[x^{\pm1},y^{\pm1}]$
% \begin{equation}
%     R=\mathbb Z_2[x^{\pm1},y^{\pm1}],
% \end{equation}
and all arithmetic is over $\mathbb Z_2$. We write
\begin{equation}
    \bar p(x,y):=p(x^{-1},y^{-1}),
    \qquad
    p^\sigma(x,y):=p(y,x),
\end{equation}
and define
\begin{equation}
    p^\rho(x,y)=p(y^{-1},x^{-1}).
\end{equation}
For an $X$-type gauge generator
\begin{equation}
    G_X(f,g,h)=(f,g,h,0,0,0)^{\mathsf T},
\end{equation}
we impose reflection symmetry about the diagonal $y=x$ by choosing the
corresponding $Z$-type gauge generator to be
\begin{equation}
    G_Z(f,g,h)
    =
    (0,0,0,f^\sigma,h^\sigma,g^\sigma)^{\mathsf T}.
\end{equation}
Thus a reflection-symmetric SBB candidate is specified by two triples
$(f_1,g_1,h_1)$ and $(f_2,g_2,h_2)$:
\begin{equation}
    G_{X,i}=G_X(f_i,g_i,h_i),
    \qquad
    G_{Z,i}=G_Z(f_i,g_i,h_i),
    \qquad
    i=1,2 .
\end{equation}
More precisely, the gauge operators are:
\begin{eqs}
    G_{X,1} = 
    \begin{bmatrix}
        f_1(x,y) \\
        \rule{0pt}{1.1em}g_1(x,y) \\
        \rule{0pt}{1.1em}h_1(x,y) \\
        \hline
        0 \\
        0 \\
        0
    \end{bmatrix}, 
    \quad
    G_{X,2} = 
    \begin{bmatrix}
        f_2(x,y) \\
        \rule{0pt}{1.1em}g_2(x,y) \\
        \rule{0pt}{1.1em}h_2(x,y) \\
        \hline
        0 \\
        0 \\
        0
    \end{bmatrix}, 
    G_{Z,1} = 
    \begin{bmatrix}
        0 \\
        0 \\
        0 \\
        \hline
        \rule{0pt}{1.1em} f_1^\sigma(x,y) \\
        \rule{0pt}{1.1em} h_1^\sigma(x,y) \\
        g_1^\sigma(x,y)
    \end{bmatrix},
    \quad
    G_{Z,2} = 
    \begin{bmatrix}
        0 \\
        0 \\
        0 \\
        \hline
        \rule{0pt}{1.1em} f_2^\sigma(x,y) \\
        \rule{0pt}{1.1em} h_2^\sigma(x,y) \\
        g_2^\sigma(x,y)
    \end{bmatrix},
    \label{eq: gauge operators}
\end{eqs}
The weight-4 search imposes
\begin{equation}
    \operatorname{wt}(f_i)+\operatorname{wt}(g_i)+\operatorname{wt}(h_i)=4,
    \qquad
    i=1,2 .
\end{equation}
For two triples $w=(f,g,h)$ and $w'=(f',g',h')$, define the reflected pairing
\begin{equation}
    \mu(w,w')
    :=
    \bar f\,{f'}^\sigma
    +
    \bar g\,{h'}^\sigma
    +
    \bar h\,{g'}^\sigma .
\end{equation}
Equivalently,
\begin{equation}
    G_X(f,g,h)\cdot G_Z(f',g',h')=\mu((f,g,h),(f',g',h')).
\end{equation}
The commutation matrix of the four gauge generators therefore has the form
\begin{equation}
    M_c
    =
    \begin{pmatrix}
        a & b\\
        b^\rho & d
    \end{pmatrix},
    \label{eq:reflection_symmetric_Mc}
\end{equation}
where
\begin{equation}
    a=\mu(w_1,w_1),
    \qquad
    b=\mu(w_1,w_2),
    \qquad
    d=\mu(w_2,w_2).
\end{equation}
Here $w_i=(f_i,g_i,h_i)$. The diagonal entries satisfy
\begin{equation}
    a^\rho=a,
    \qquad
    d^\rho=d,
\end{equation}
and the off-diagonal entries are related by reflection:
\begin{equation}
    \mu(w_2,w_1)=\rho\!\left(\mu(w_1,w_2)\right)=b^\rho.
\end{equation}
Suppose that $\det M_c=0$.
Then
\begin{equation}
    (b^\rho,a)M_c=0,
    \qquad
    M_c
    \begin{pmatrix}
        b\\ a
    \end{pmatrix}
    =0 .
\end{equation}
By the kernel--stabilizer correspondence, the corresponding local stabilizers are
\begin{equation}
    S_X=b^\sigma G_{X,1}+\bar a\,G_{X,2},
    \qquad
    S_Z=bG_{Z,1}+aG_{Z,2}.
    \label{eq:reflection_symmetric_stabilizers}
\end{equation}

\begin{lemma}[Reflection symmetry of local stabilizers]
\label{lem:reflection_symmetric_stabilizers}
The local stabilizers in Eq.~\eqref{eq:reflection_symmetric_stabilizers} inherit
the same reflection symmetry as the gauge generators. More explicitly, if
\begin{equation}
    S_X=G_X(F,G,H),
\end{equation}
then
\begin{equation}
    S_Z=G_Z(F,G,H).
\end{equation}
\end{lemma}

\begin{proof}
From Eq.~\eqref{eq:reflection_symmetric_stabilizers},
\begin{equation}
    F=b^\sigma f_1+\bar a f_2,
    \qquad
    G=b^\sigma g_1+\bar a g_2,
    \qquad
    H=b^\sigma h_1+\bar a h_2 .
\end{equation}
Since $a^\rho=a$, we have $(\bar a)^\sigma=a$. Therefore
\begin{equation}
    F^\sigma=b f_1^\sigma+a f_2^\sigma,
    \qquad
    G^\sigma=b g_1^\sigma+a g_2^\sigma,
    \qquad
    H^\sigma=b h_1^\sigma+a h_2^\sigma .
\end{equation}
Hence
\begin{equation}
    G_Z(F,G,H)
    =
    (0,0,0,F^\sigma,H^\sigma,G^\sigma)
    =
    bG_{Z,1}+aG_{Z,2}
    =
    S_Z .
\end{equation}
\end{proof}

We now prove the normal form used to enumerate the first gauge generator in the
search.

\begin{proposition}[weight-4 normal forms for a monomial pivot]
\label{proposition: 1}
\label{prop:weight_four_monomial_pivot}
Let $f,g,h\in R$ satisfy
\begin{equation}
    \operatorname{wt}(f)+\operatorname{wt}(g)+\operatorname{wt}(h)=4,
\end{equation}
and suppose that the reflected self-pairing
\begin{equation}
    \mu(f,g,h)
    :=
    \bar f\,f^\sigma
    +
    \bar g\,h^\sigma
    +
    \bar h\,g^\sigma
\end{equation}
is a monomial.

Then, $f$ is a monomial $x^\alpha y^\beta$ for some $\alpha,\beta\in\mathbb Z$ and $\bar g\,h^\sigma+\bar h\,g^\sigma=0$.
Moreover, after possibly interchanging $g$ and $h$, exactly one of the following
two cases occurs:
\begin{enumerate}
    \item $g=0$ and $\operatorname{wt}(h)=3$;
    \item $g$ is a monomial and $h=ug$
    % \begin{equation}
    %     h=ug
    % \end{equation}
    for some $u\in R^\rho$ with $\operatorname{wt}(u)=2$, where
    \begin{equation}
        R^\rho:=\{p\in R:p^\rho=p\}.
    \end{equation}
\end{enumerate}
Conversely, every triple of the above form satisfies
\begin{equation}
    \mu(f,g,h)
    =
    \bar f\,f^\sigma
    =
    x^{\beta-\alpha}y^{\alpha-\beta},
\end{equation}
and hence has monomial reflected self-pairing.
\end{proposition}

\begin{proof}
Write
\begin{equation}
    T:=\bar g\,h^\sigma+\bar h\,g^\sigma .
\end{equation}
Since
\begin{equation}
    T=Q+Q^\rho,
    \qquad
    Q:=\bar g\,h^\sigma,
\end{equation}
the term $T$ contains no $\rho$-fixed monomial: any $\rho$-fixed monomial
appears twice and cancels in characteristic two. On the other hand,
$\bar f f^\sigma$ is $\rho$-invariant. Thus, if
$\mu(f,g,h)=\bar f f^\sigma+T$ is a monomial, then the $\rho$-fixed part of
$\bar f f^\sigma$ must itself be a single monomial.

Let
\begin{equation}
    f=\sum_{i=1}^r x^{a_i}y^{b_i},
    \qquad
    r=\operatorname{wt}(f).
\end{equation}
Then
\begin{equation}
\begin{aligned}
    \bar f f^\sigma
    &=
    \sum_{i=1}^r x^{b_i-a_i}y^{a_i-b_i}
    +
    \sum_{1\le i<j\le r}
    \left(
        x^{b_j-a_i}y^{a_j-b_i}
        +
        x^{b_i-a_j}y^{a_i-b_j}
    \right).
\end{aligned}
\end{equation}
The second sum consists of $\rho$-paired monomials; if such a pair collapses to
a fixed monomial, it cancels. Hence the fixed-monomial contribution comes from
the first sum.

Since the total weight is four, the possible values of $r$ are
$0,1,2,3,4$. If $r=0$, then the fixed part is zero, contradicting the assumption
that $\mu(f,g,h)$ is a single fixed monomial. If $r=2$, the fixed part is either
zero or a sum of two distinct monomials, again impossible.

If $r=3$, then $\operatorname{wt}(g)+\operatorname{wt}(h)=1$, so
$T=0$. If $r=4$, then $g=h=0$, so again $T=0$. In either case,
$\bar f f^\sigma$ would be a monomial. Since monomials are precisely the units
of $R$, this would imply that $f$ is a unit, contradicting $r=3$ or $r=4$.
Therefore $r=1$, and hence $f=x^\alpha y^\beta$.
% \begin{equation}
%     f=x^\alpha y^\beta .
% \end{equation}
It follows immediately that
\begin{equation}
    \bar f f^\sigma
    =
    x^{\beta-\alpha}y^{\alpha-\beta}.
\end{equation}

Since $\bar f f^\sigma$ already accounts for the unique fixed monomial in
$\mu(f,g,h)$, and since $T$ has no fixed-monomial contribution, we must have $T=0$, that is,
\begin{equation}
    \bar g\,h^\sigma+\bar h\,g^\sigma=0 .
\end{equation}
The remaining weight condition gives
\begin{equation}
    \operatorname{wt}(g)+\operatorname{wt}(h)=3.
\end{equation}
After possibly interchanging $g$ and $h$, we may assume
$\operatorname{wt}(g)\le \operatorname{wt}(h)$. Thus either
$g=0$ and $\operatorname{wt}(h)=3$, or
$\operatorname{wt}(g)=1$ and $\operatorname{wt}(h)=2$.

In the second case, $g$ is a monomial and therefore a unit. Set
\begin{equation}
    u:=h g^{-1}.
\end{equation}
Then $h=ug$ and $\operatorname{wt}(u)=2$. The condition $\bar g\,h^\sigma+\bar h\,g^\sigma=0$
% \begin{equation}
%     \bar g\,h^\sigma+\bar h\,g^\sigma=0
% \end{equation}
becomes
\begin{equation}
    \bar g\,g^\sigma(u^\sigma+\bar u)=0.
\end{equation}
Since $\bar g\,g^\sigma$ is a unit, this is equivalent to $u^\sigma=\bar u$.
% \begin{equation}
%     u^\sigma=\bar u.
% \end{equation}
Applying $\sigma$ to both sides gives
\begin{equation}
    u^\rho=u,
\end{equation}
so $u\in R^\rho$.

Conversely, if $g=0$, then
\begin{equation}
    \bar g\,h^\sigma+\bar h\,g^\sigma=0
\end{equation}
trivially. If $h=ug$ with $u\in R^\rho$, then $u^\sigma=\bar u$, and hence
\begin{equation}
\begin{aligned}
    \bar g\,h^\sigma+\bar h\,g^\sigma
    &=
    \bar g\,u^\sigma g^\sigma
    +
    \bar u\,\bar g\,g^\sigma  \\
    &=
    \bar g\,g^\sigma(u^\sigma+\bar u)
    =
    0 .
\end{aligned}
\end{equation}
Thus in both cases
\begin{equation}
    \mu(f,g,h)=\bar f f^\sigma
    =
    x^{\beta-\alpha}y^{\alpha-\beta},
\end{equation}
which is a monomial.
\end{proof}

Finally, we record the explicit Clifford reduction in the direct monomial-pivot
sector. This is the form used for the examples in the main text.

\begin{lemma}[Explicit Clifford reduction in the direct sector]
\label{lem:explicit_direct_clifford}
Assume $G_{X,1}=G_X(f_1,g_1,h_1)$ satisfies the direct monomial-pivot condition
of Proposition~\ref{prop:weight_four_monomial_pivot}: $f_1$ is a monomial
and $\bar g_1 h_1^\sigma+\bar h_1 g_1^\sigma=0$.
% \begin{equation}
%     \bar g_1 h_1^\sigma+\bar h_1 g_1^\sigma=0 .
% \end{equation}
Set
\begin{equation}
    u:=g_1 f_1^{-1},
    \qquad
    v:=h_1 f_1^{-1}.
\end{equation}
Then there is a finite-depth Clifford circuit that maps
$G_{X,1}$ and $G_{Z,1}$ to a single-qubit Pauli pair, up to the monomial
translations determined by $f_1$ and $f_1^\sigma$. After the translation
normalization $f_1=1$, the pair is mapped exactly to
\begin{equation}
    (1,0,0,0,0,0),
    \qquad
    (0,0,0,1,0,0).
\end{equation}
\end{lemma}

\begin{proof}
In the Laurent-polynomial Pauli module, the first layer of CNOT gates is
represented by
\begin{equation}
U_1=
\begin{bmatrix}
1 & 0 & 0 & 0 & 0 & 0\\
u & 1 & 0 & 0 & 0 & 0\\
v & 0 & 1 & 0 & 0 & 0\\
0 & 0 & 0 & 1 & \bar u & \bar v\\
0 & 0 & 0 & 0 & 1 & 0\\
0 & 0 & 0 & 0 & 0 & 1
\end{bmatrix}.
\end{equation}
It sends
\begin{equation}
    U_1G_{X,1}
    =
    \left[\begin{array}{c}
        f_1\\0\\0\\
        \hline
        0\\0\\0
    \end{array}\right].
\end{equation}
Using
\begin{equation}
    \bar g_1 h_1^\sigma+\bar h_1 g_1^\sigma=0,
\end{equation}
one also obtains
\begin{equation}
    U_1G_{Z,1}
    =
    \left[\begin{array}{c}
        0\\0\\0\\
        \hline
        f_1^\sigma\\h_1^\sigma\\g_1^\sigma
    \end{array}\right].
\end{equation}

The second layer is represented by
\begin{equation}
U_2=
\begin{bmatrix}
1 & v^\rho & u^\rho & 0 & 0 & 0\\
0 & 1 & 0 & 0 & 0 & 0\\
0 & 0 & 1 & 0 & 0 & 0\\
0 & 0 & 0 & 1 & 0 & 0\\
0 & 0 & 0 & v^\sigma & 1 & 0\\
0 & 0 & 0 & u^\sigma & 0 & 1
\end{bmatrix}.
\end{equation}
Therefore, with $U:=U_2U_1$,
\begin{equation}
    UG_{X,1}
    =
    \left[\begin{array}{c}
        f_1\\0\\0\\
        \hline
        0\\0\\0
    \end{array}\right],
    \qquad
    UG_{Z,1}
    =
    \left[\begin{array}{c}
        0\\0\\0\\
        \hline
        f_1^\sigma\\0\\0
    \end{array}\right].
\end{equation}
Since $f_1$ is a monomial, these are single-qubit Pauli operators up to lattice
translation. If we fix the translation redundancy by setting $f_1=1$, they are
exactly the canonical pair on the vertex qubit.
\end{proof}

In the same direct sector, the transformed stabilizers take the BB form
explicitly. Let
\begin{equation}
    a=G_{X,1}\cdot G_{Z,1},
    \qquad
    b=G_{X,1}\cdot G_{Z,2},
\end{equation}
and use the local stabilizers
\begin{equation}
    S_X=b^\sigma G_{X,1}+\bar a\,G_{X,2},
    \qquad
    S_Z=bG_{Z,1}+aG_{Z,2}.
\end{equation}
With $u=g_1f_1^{-1}$ and $v=h_1f_1^{-1}$ as above, the Clifford circuit
$U=U_2U_1$ gives
\begin{equation}
    US_X
    =
    \left[\begin{array}{c}
        0\\
        \bar a\,(g_2+u f_2)\\
        \bar a\,(h_2+v f_2)\\
        \hline
        0\\
        0\\
        0
    \end{array}\right],\qquad
    US_Z
    =
    \left[\begin{array}{c}
        0\\
        0\\
        0\\
        \hline
        0\\
        b h_1^\sigma+a h_2^\sigma\\
        b g_1^\sigma+a g_2^\sigma
    \end{array}\right].
\end{equation}
After deleting the disentangled vertex gauge qubit, the remaining two qubits per
unit cell support the corresponding BB stabilizer generators
\begin{equation}
    \widetilde S_X
    =
    \left[\begin{array}{c}
        \bar a\,(g_2+u f_2)\\
        \bar a\,(h_2+v f_2)\\
        \hline
        0\\
        0
    \end{array}\right],
    \qquad
    \widetilde S_Z
    =
    \left[\begin{array}{c}
        0\\
        0\\
        \hline
        b h_1^\sigma+a h_2^\sigma\\
        b g_1^\sigma+a g_2^\sigma
    \end{array}\right].
    \label{eq: stabilizer for SBB after U}
\end{equation}
In the search normalization $f_1=1$, one has $a=1$, so these formulas simplify to the BB stabilizers used in the main text. 
A Gr\"obner basis computation can then be used to determine the number of logical qubits encoded by the corresponding SBB code.

%%%%%%%%%%%%%%%%%%%%%%%%%%%%%%%%%%%%%%%%%%%%%%%%%%%%%%%%%%%%%%%%%%%%%%%%%%%%%%%%%%%%%%%%%%%%%%%%%%%%%%%%%%%%%%%%
\section{Example data and verification of SBB codes}
\label{app: examples_verification}

This appendix provides the explicit data used for the SBB examples quoted in the main text.
We first verify the guiding $[[75,10,5]]$ code in detail, including its gauge generators, commutation matrix, local stabilizers, Clifford reduction, and associated BB stabilizers.
We then list further low-overhead examples from the finite search.
Together with the reflection rule for obtaining the $Z$-type gauge generators, the data in Table~\ref{tab: more examples} are sufficient to reconstruct each code and reproduce the quoted parameters.

\subsection{Verification of the $[[75,10,5]]$ guiding example}

We begin with the $[[75,10,5]]$ code discussed in the main text.
The gauge generators shown in Fig.~\ref{fig: gauge operator of [[75, 10, 5]] code} are
\begin{eqs}
    G_{X,1} = 
    \left[\begin{array}{c}
        x^2 \\
        \rule{0pt}{1.1em}y^2 \\
        \rule{0pt}{1.1em}x+x^2y \\
        \hline
        0 \\
        0 \\
        0
    \end{array}\right], \qquad
    G_{X,2} = 
    \left[\begin{array}{c}
        1+y^2 \\
        \rule{0pt}{1.1em}x+y \\
        \rule{0pt}{1.1em}0 \\
        \hline
        0 \\
        0 \\
        0
    \end{array}\right], \\
    G_{Z,1} = 
    \left[\begin{array}{c}
        0 \\
        0 \\
        0 \\
        \hline
        \rule{0pt}{1.1em}y^2 \\
        \rule{0pt}{1.1em}y+xy^2 \\
        x^2
    \end{array}\right], \qquad
    G_{Z,2} =  
    \left[\begin{array}{c}
        0 \\
        0 \\
        0 \\
        \hline
        \rule{0pt}{1.1em}1+x^{2} \\
        \rule{0pt}{1.1em}0 \\
        x+y
    \end{array}\right].
    \label{eq: [[75,10,5]] gauge operators}
\end{eqs}
Their commutation matrix is
\begin{eqs}
M_c=
\begin{pmatrix}
 x^{-2}y^2 & x^{-1}y+x^{-1}y^{-1} \\
 x^{-1}y+xy & 1+x^2+y^{-2}+x^2y^{-2}
\end{pmatrix}.
\label{eq:Mc_75}
\end{eqs}
One directly checks that $\det(M_c)=0$, and the upper-left entry $x^{-2}y^2$ is a monomial pivot.
Therefore Theorems~\ref{thm: nonlocal stabilizer} and~\ref{thm: exist Clifford gate} apply: the code has local stabilizers, no additional nonlocal stabilizers on the torus, and a finite-depth Clifford correspondence with a BB stabilizer code.

The local stabilizers are those obtained from the kernel construction.
Explicitly,
\begin{eqs}
S_X=(x^{-2}y+y)G_{X,1}+xG_{X,2},\qquad
S_Z=(xy^{-2}+x)G_{Z,1}+yG_{Z,2}.
\end{eqs}
They are shown in Fig.~\ref{fig: Stabilizers and dressed logical operators of [[75, 10, 5]] code}.
Equivalently, choose
\begin{eqs}
P=
\begin{pmatrix}
1 & 0 \\
x^{2}y^{-1}+y^{-1} & x^{-1}
\end{pmatrix},
\qquad
Q=
\begin{pmatrix}
1 & xy^{-2}+x \\
0 & y
\end{pmatrix}.
\label{eq: P and Q of [[75,10,5]]}
\end{eqs}
Then
\begin{eqs}
P M_c Q=
\begin{pmatrix}
x^{-2}y^2 & 0\\
0 & 0
\end{pmatrix}.
\label{eq: [[75,10,5]] gauge qubits}
\end{eqs}
The vanishing second row and second column encode the local $X$- and $Z$-type stabilizers above.

\begin{figure*}[thb]
    \centering
    \includegraphics[width=0.95\textwidth]{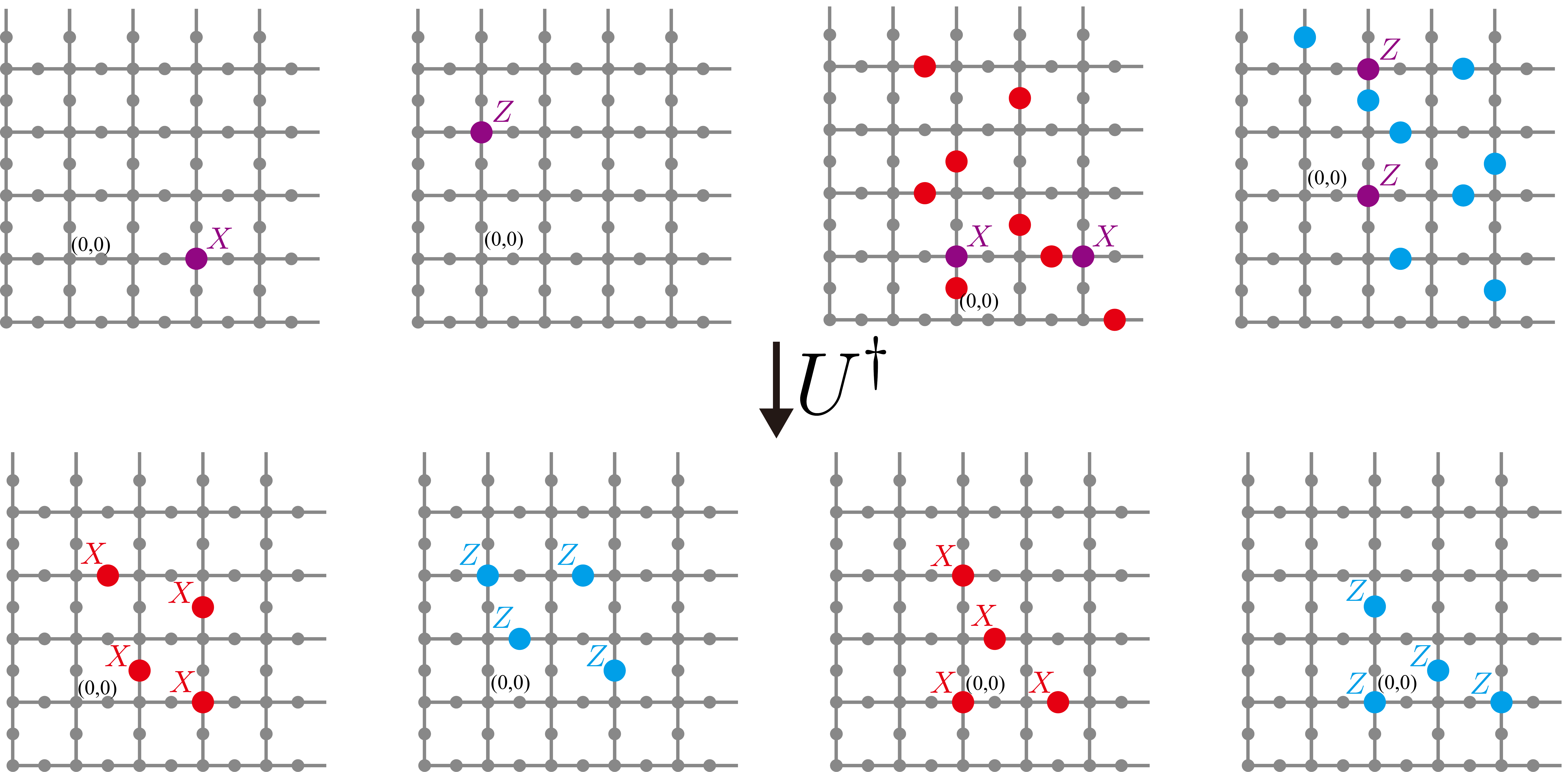}
    \caption{
    Inverse construction of the $[[75,10,5]]$ SBB code from the corresponding BB code.
    The upper row shows the gauge-qubit Pauli pair on each vertex together with the BB stabilizers decorated by these gauge qubits.
    Applying the inverse Clifford circuit $U^\dagger$ gives the four weight-$4$ SBB gauge generators shown in the lower row.
    Thus the weight-$8$ BB stabilizer checks are realized through local weight-$4$ gauge measurements in the SBB description.
    }
\label{fig: Clifford_gate_inverse}
\end{figure*}

We next display the Clifford reduction explicitly.
Let the CNOT circuit in Fig.~\ref{fig: clifford gate and stabilizers} be written as $U=U_2U_1$, where
\begin{eqs}
U_1=
\begin{bmatrix}
1 & 0 & 0 & 0 & 0 & 0\\
x^{-2}y^2 & 1 & 0 & 0 & 0 & 0\\
x^{-1}+y & 0 & 1 & 0 & 0 & 0\\
0 & 0 & 0 & 1 & x^2y^{-2} & x+y^{-1}\\
0 & 0 & 0 & 0 & 1 & 0\\
0 & 0 & 0 & 0 & 0 & 1
\end{bmatrix},\qquad
U_2=
\begin{bmatrix}
1 & y+x^{-1} & x^{-2}y^2 & 0 & 0 & 0\\
0 & 1 & 0 & 0 & 0 & 0\\
0 & 0 & 1 & 0 & 0 & 0\\
0 & 0 & 0 & 1 & 0 & 0\\
0 & 0 & 0 & y^{-1}+x & 1 & 0\\
0 & 0 & 0 & x^2y^{-2} & 0 & 1
\end{bmatrix}.
\end{eqs}
Then
\begin{eqs}
UG_{X,1}=
\left[\begin{array}{c}
        x^2\\0\\0\\0\\0\\0
    \end{array}\right],
% \begin{bmatrix}
% x^2\\0\\0\\0\\0\\0
% \end{bmatrix},
\qquad
UG_{Z,1}=
\left[\begin{array}{c}
        0\\0\\0\\y^2\\0\\0
    \end{array}\right].
% \begin{bmatrix}
% 0\\0\\0\\y^2\\0\\0
% \end{bmatrix}.
\end{eqs}
Thus $U$ maps the gauge pair $G_{X,1},G_{Z,1}$ to a single-qubit Pauli pair on the vertex qubit, up to monomial translations.
After this vertex gauge qubit is removed, the transformed stabilizers act only on the two edge qubits per unit cell:
\begin{eqs}
\widetilde S_X=
\left[\begin{array}{c}
        x^{-1}y^2+x^{-1}y^4+xy+x^2 \\
        \rule{0pt}{1.1em}1+y^2+xy+xy^3 \\
        \hline
        0 \\
        0
    \end{array}\right],\qquad
\widetilde S_Z=x^{2} y^2
\left[\begin{array}{c}
        0 \\
        0 \\
        \hline
        \rule{0pt}{1.1em} 1+y^{-2}+x^{-1}y^{-1}+x^{-1}y^{-3}\\
        xy^{-2}+xy^{-4}+x^{-1}y^{-1}+x^{-2}
    \end{array}\right].
\end{eqs}
These two operators have the BB form and have weight $8$ in this representative.
Therefore the $[[75,10,5]]$ SBB code is a weight-$4$ subsystem realization of the corresponding BB stabilizer code.
Conversely, as shown in Fig.~\ref{fig: Clifford_gate_inverse}, starting from this BB code, adding one gauge qubit at each vertex, and applying $U^\dagger$ produces the four weight-$4$ SBB gauge generators in Eq.~\eqref{eq: [[75,10,5]] gauge operators}.

The equality of logical dimensions can also be checked directly from the BB description.
Let $f(x,y)$ and $g(x,y)$ be the two nonzero components of $\widetilde S_X$, after removing the overall monomial.
On the $5\times5$ torus, the ideal
\begin{eqs}
\langle f(x,y),\,g(x,y),\,1+x^5,\,1+y^5\rangle
\end{eqs}
has Gr\"obner basis
\begin{eqs}
\langle 1+y^5,\,x+y^4\rangle.
\end{eqs}
The quotient ring therefore has five independent monomials.
Since a BB code encodes twice this quotient dimension, the corresponding BB code has $k=10$ logical qubits, in agreement with the $[[75,10,5]]$ SBB code.

\begin{table*}[thb]
\centering
% \footnotesize
\renewcommand{\arraystretch}{1.8}
\setlength{\tabcolsep}{3pt}
\begin{tabular}{ccccccc}
\hline
$[[n,k,d]]$
& $(f_1,g_1,h_1)$
& $(f_2,g_2,h_2)$
& $a_1=(0,\alpha)$
& $a_2=(\beta,\gamma)$
& $S_X$
& $S_Z$ \\
\hline
$[[27,6,3]]$
& $(y^2,y,1+xy)$
& $(1+xy+x^2,1,0)$
& $(0,3)$
& $(3,0)$
& \makecell[l]{\\$(\bar{x}^2+\bar{x}\bar{y}+\bar{x}y)G_{X,1}$\\$+\bar{x}^2y^2G_{X,2}$\\{}}
& \makecell[l]{$(\bar{x}\bar{y}+\bar{y}^2+x\bar{y})G_{Z,1}$\\$+x^2\bar{y}^2G_{Z,2}$} \\
\hline

$[[60,10,4]]$
& $(y^2,x,1+xy)$
& $(x^2+x^2y^2,y^2+xy,0)$
& $(0,5)$
& $(4,-1)$
& \makecell[l]{\\$(\bar{x}y+xy)G_{X,1}$\\$+\bar{x}^2y^2G_{X,2}$\\{}}
& \makecell[l]{$(x\bar{y}+xy)G_{Z,1}$\\$+x^2\bar{y}^2G_{Z,2}$} \\
\hline

$[[75,10,5]]$
& $(x^2,y^2,x+x^2y)$
& $(1+y^2,x+y,0)$
& $(0,5)$
& $(5,0)$
& \makecell[l]{\\$(\bar{x}\bar{y}+x\bar{y})G_{X,1}$\\$+x^2\bar{y}^2G_{X,2}$\\{}}
& \makecell[l]{$(\bar{x}\bar{y}+\bar{x}y)G_{Z,1}$\\$+\bar{x}^2y^2G_{Z,2}$} \\
\hline

$[[90,12,5]]$
& $(y^2,x^2,1+x^2y^2)$
& $(1+x^3y,y^2+xy,0)$
& $(0,6)$
& $(5,2)$
& \makecell[l]{\\$(\bar{x}\bar{y}+y^2)G_{X,1}$\\$+\bar{x}^2y^2G_{X,2}$\\{}}
& \makecell[l]{$(\bar{x}\bar{y}+x^2)G_{Z,1}$\\$+x^2\bar{y}^2G_{Z,2}$} \\
\hline

$[[108,12,6]]$
& $(y^2,x^2,1+x^2y^2)$
& $(1+x^3y,y^2+xy,0)$
& $(0,9)$
& $(4,1)$
& \makecell[l]{\\$(\bar{x}\bar{y}+y^2)G_{X,1}$\\$+\bar{x}^2y^2G_{X,2}$\\{}}
& \makecell[l]{$(\bar{x}\bar{y}+x^2)G_{Z,1}$\\$+x^2\bar{y}^2G_{Z,2}$} \\
\hline

$[[126,14,6]]$
& $(y^2,x^2,x+x^2y)$
& $(x+x^3,y^2+xy,0)$
& $(0,7)$
& $(6,-1)$
& \makecell[l]{\\$(\bar{y}+y)G_{X,1}$\\$+\bar{x}^2y^2G_{X,2}$\\{}}
& \makecell[l]{$(\bar{x}+x)G_{Z,1}$\\$+x^2\bar{y}^2G_{Z,2}$} \\
\hline
\end{tabular}
\caption{
Explicit data for representative weight-$4$ SBB codes found in the finite search.
Each row lists the code parameters, the two $X$-type gauge-generator triples, the two twisted-torus translation vectors, and the local stabilizer combinations $S_X$ and $S_Z$.
The translation vectors define the periodic boundary conditions $y^\alpha=1$ and $x^\beta y^\gamma=1$.
The $Z$-type gauge generators are obtained from the reflection rule
$G_Z(f,g,h)=(0,0,0,f^\sigma,h^\sigma,g^\sigma)^{\mathsf T}$.
The stabilizers are written in the unnormalized convention associated with the listed gauge-generator triples: for
$M_c=\begin{psmallmatrix}a&b\\ c&d\end{psmallmatrix}$ with $\det M_c=0$, we use
$S_X=\bar c\,G_{X,1}+\bar a\,G_{X,2}$ and
$S_Z=b\,G_{Z,1}+a\,G_{Z,2}$ (here $\bar{x}=x^{-1}$ and $\bar{y}=y^{-1}$).
Code distances are computed exactly using the integer-programming approach~\cite{landahl2011fault, Bravyi2024HighThreshold}.
}
\label{tab: more examples}
\end{table*}

%%%%%%%%%%%%%%%%%%%%%%%%%%%
\subsection{Further examples from the finite search}

Table~\ref{tab: more examples} lists representative weight-$4$ SBB codes found in the finite search.
For each row, the second and third columns give the two $X$-type gauge generators,
\begin{equation}
G_{X,1}=G_X(f_1,g_1,h_1),\qquad
G_{X,2}=G_X(f_2,g_2,h_2).
\end{equation}
The corresponding $Z$-type gauge generators are obtained by the combined reflection symmetry used in Appendix~\ref{app:reflection_symmetric_search}:
\begin{equation}
G_Z(f,g,h)=(0,0,0,f^\sigma,h^\sigma,g^\sigma)^{\mathsf T},
\qquad
p^\sigma(x,y)=p(y,x).
\end{equation}
Thus the table gives the full translation-invariant gauge-generator data for each example. 

The last two columns specify the twisted torus.
For translation vectors
\begin{equation}
a_1=(0,\alpha),\qquad a_2=(\beta,\gamma),
\end{equation}
Laurent polynomials are reduced modulo the corresponding periodicities.
The number of unit cells is
\begin{eqs}
\left|\det
\begin{pmatrix}
0 & \beta \\
\alpha & \gamma
\end{pmatrix}
\right|
=
\alpha\beta .
\end{eqs}
Since each unit cell contains three physical qubits, the block length is $n=3\alpha\beta$.
The listed parameters $[[n,k,d]]$ are obtained by constructing the finite binary gauge and stabilizer matrices on the corresponding twisted torus and computing the number of protected logical qubits and the dressed distance, as described in the main text.

As a consistency check, the $[[75,10,5]]$ row has
$a_1=(0,5)$ and $a_2=(5,0)$, so the twisted torus contains
$5\times 5=25$ unit cells and hence
$n=3\times 25=75$ physical qubits. The other rows are obtained in the same way.

For each entry in Table~\ref{tab: more examples}, we also compute the corresponding commutation matrix.
In the unnormalized convention associated with the listed gauge-generator triples, write
\begin{equation}
    M_c=
    \begin{pmatrix}
        a & b\\
        c & d
    \end{pmatrix}.
\end{equation}
When $\det M_c=0$, the local stabilizers are obtained from the kernel formula
\begin{equation}
    S_X=\bar c\,G_{X,1}+\bar a\,G_{X,2},
    \qquad
    S_Z=bG_{Z,1}+aG_{Z,2}.
\end{equation}
Explicitly, the commutation matrices are:
\begin{enumerate}
    \item $[[27,6,3]]$:
    \begin{equation}
        M_c =
        \begin{pmatrix}
            x^2y^{-2}
            &
            x^{-1}y^{-1}+y^{-2}+xy^{-1}
            \\
            ~~xy^{-1}+xy+x^2~~
            &
            ~~x^{-2}+x^{-2}y^2+x^{-1}y^{-1}+y^2+xy~~
        \end{pmatrix}.
    \end{equation}

    \item $[[60,10,4]]$:
    \begin{equation}
        M_c =
        \begin{pmatrix}
            x^2y^{-2}
            &
            xy^{-1}+xy
            \\
            ~~x^{-1}y^{-1}+xy^{-1}~~
            &
            ~~x^{-2}+x^{-2}y^2+1+y^2~~
        \end{pmatrix}.
    \end{equation}

    \item $[[75,10,5]]$:
    \begin{equation}
        M_c =
        \begin{pmatrix}
            x^{-2}y^2
            &
            x^{-1}y^{-1}+x^{-1}y
            \\
            ~~x^{-1}y+xy~~
            &
            ~~y^{-2}+1+x^2y^{-2}+x^2~~
        \end{pmatrix}.
    \end{equation}

    \item $[[90,12,5]]$:
    \begin{equation}
        M_c =
        \begin{pmatrix}
            x^2y^{-2}
            &
            x^{-1}y^{-1}+x^2
            \\
            ~~y^{-2}+xy~~
            &
            ~~x^{-3}y^{-1}+x^{-2}y^2+1+xy^3~~
        \end{pmatrix}.
    \end{equation}

    \item $[[108,12,6]]$:
    \begin{equation}
        M_c =
        \begin{pmatrix}
            x^2y^{-2}
            &
            x^{-1}y^{-1}+x^2
            \\
            ~~y^{-2}+xy~~
            &
            ~~x^{-3}y^{-1}+x^{-2}y^2+1+xy^3~~
        \end{pmatrix}.
    \end{equation}

    \item $[[126,14,6]]$:
    \begin{equation}
        M_c =
        \begin{pmatrix}
            x^2y^{-2}
            &
            x^{-1}+x
            \\
            ~~y^{-1}+y~~
            &
            ~~x^{-3}y+x^{-3}y^3+x^{-1}y+x^{-1}y^3~~
        \end{pmatrix}.
    \end{equation}
\end{enumerate}
For all matrices above, one verifies that
\begin{equation}
    \det M_c=0.
\end{equation}
Moreover, since each matrix has a monomial entry, the entry ideal satisfies
\begin{equation}
    I_1(M_c)=R.
\end{equation}
Therefore, the no-nonlocal-stabilizer criterion applies. The distances listed in Table~\ref{tab: more examples} are the dressed distances of the corresponding subsystem codes.

\end{document}